%% file: main.tex
\documentclass[%
 reprints,nofootinbib
]{revtex4-2}
\usepackage{iceberg}

\usepackage[caption=false]{subfig}
\graphicspath{{./}}

\def\bra{\langle}
\def\ket{\rangle}
\def\lr{\leftrightarrow}
\def\backslash{\symbol{92}}
\def\ve{\varepsilon}
\def\weight{\mathfrak{w}}

\def\run{\{\text{run success}\}}
\newcommand{\noisy}[1]{\tilde{#1}}

\DeclareMathAlphabet{\dutchcal}{U}{dutchcal}{m}{n}
\docsvlist{A,B,C,D,E,F,G,H,I,J,K,L,M,N,O,P,Q,R,S,T,U,V,W,X,Y,Z}

\docsvlist{a,b,c,d,e,f,g,h,j,k,l,m,n,o,p,q,r,s,t,u,v,w,x,y,z}

\DeclareMathOperator{\ch}{ch}
\DeclareMathOperator{\sh}{sh}

\DeclareMathOperator{\cone}{cone}

\DeclareMathOperator{\desc}{desc}
\DeclareMathOperator{\rep}{rep}

\usepackage{accents}

\newcommand{\ubar}[1]{\underaccent{\bar}{#1}}
\newcommand{\super}[1]{^{\scriptscriptstyle(#1)}}

\usepackage{dsfont}
\newcommand{\dnum}[1][1]{\mathds{#1}}

\newtheorem{hypothesis}[definition]{Hypothesis}

\usepackage{tikz-cd} % draw commutative diagram
\usetikzlibrary{matrix,arrows,decorations.pathmorphing}
\usetikzlibrary{fit, shapes.geometric,calc}
\makeatletter
\def\input@path{{./}}

\renewenvironment{algorithm}[1][]{%
  \par\addvspace{\intextsep}%
  \refstepcounter{algorithm}%
  \hrule height.8pt depth0pt \kern2pt%
  \def\caption##1{{\raggedright\textbf{\ALG@name~\thealgorithm}\ ##1\par}%
    \addcontentsline{loa}{algorithm}{\protect\numberline{\thealgorithm}##1}%
    \kern2pt\hrule\kern2pt\nobreak}%
}{%
  \kern2pt\hrule\par\addvspace{\intextsep}%
}
\makeatother

\algdef{SE}[REPEATN]{RepeatN}{EndRepeatN}[1]
{\textbf{repeat} #1 \textbf{times}}
{\textbf{end repeat}}
  
\begin{document}

\title{One-Shot any Code}
\author{Andrew C. Yuan}
% \author{Nouédyn Baspin}
\affiliation{%
 Iceberg Quantum, Sydney
}%
\author{{\scriptsize\texttt{\{andrew\}@iceberg-quantum.com}}}

\begin{abstract}
We first provide a simple proof showing that any single-shot code is self-correcting.
We then show that any $[[n,k]]$ CSS quantum low-density parity-check (QLDPC) code can be transformed into an $[[nm,k]]$ CSS QLDPC code with single-shot (SS) quantum error correction and an efficient local decoder which runs in $O(\log m)$ parallel time. 
Below a constant threshold $p_{\rm RG}>0$ for joint local stochastic physical and measurement noise, the logical failure probability over $T$ correction rounds is bounded by
\begin{equation*}
    O(T n)\exp[-\Omega(m^\alpha)]
\end{equation*}
for a constant $\alpha>0$; we say that the \textit{threshold} is $p_{\rm RG}$ and \textit{error suppression} is $\Omega(m^{\alpha})$.
Moreover, if the input code is equipped with a single-shot decoder with threshold $p_c>0$ and error suppression $\Omega(n^{\beta})$, then the output code \textit{enhances} the error suppression to $\Omega(m^{\alpha} n^{\beta})$, while maintaining the decoder-independent threshold $p_{\rm RG}$, and thus \textit{enhanced} if $p_c < p_{\rm RG}$.
\end{abstract}

\maketitle

\tableofcontents
\input{intro}
\input{prelim}

\input{SS-to-SC}
\input{level}
\input{RG}

\input{run-failure}

\input{enhance}
\input{repair}

\clearpage
\appendix

\input{iterate}
\input{contraction}

\input{dual}
\input{metric}

\bibliography{main}

\end{document}

%% file: intro.tex
\section{Introduction}
\label{sec:intro}

Quantum low-density parity-check (QLDPC) codes offer a route to fault-tolerant quantum computing with low space overhead. 
The construction of asymptotically good QLDPC codes shows that constant encoding rate and linear distance are compatible with bounded check weight and qubit degree \cite{panteleev2022asymptotically,leverrier2022quantum}. 
However, good code parameters alone do not determine the cost of quantum error correction (QEC). Syndrome measurements are themselves noisy, and a decoder must distinguish physical errors from incorrect measurement outcomes. 
The conventional approach obtains this information by repeating syndrome extraction over time, often over a window of order the code distance $d$ \cite{gottesman2013fault}. 
This introduces a temporal cost which can remain substantial even when the spatial cost of encoding is small.

Single-shot QEC replaces this temporal redundancy with structure in the code and its measured checks: one round of noisy syndrome measurements suffices to control the residual error, and the procedure can be repeated without accumulating an uncontrolled error. 
Bombin introduced this principle for three-dimensional gauge color codes and related it to confinement and self-correction \cite{bombin2015single}. 
Campbell developed a general theory based on soundness and showed that suitable choices of checks can give single-shot properties even for an arbitrary stabilizer code, although these checks need not remain low weight \cite{campbell2019theory}.
His argument addresses adversarial noise; it does not establish a stochastic threshold, which also requires controlling the entropy, i.e., the number of possible fault configurations. 
Efficient single-shot decoders are also known for quantum expander codes \cite{fawzi2020constant} and, more recently, quantum Tanner codes \cite{gu2024single}. 
The latter combines single-shot correction with asymptotically good parameters. 
These results motivate a different question: can one equip an arbitrary input CSS QLDPC code with single-shot protection while preserving bounded check weight and qubit degree?

\subsection{Main result}
\label{sec:intro-main-result}

In this work, we first prove that any single-shot stabilizer code is self-correcting, formally stated in Theorem \ref{thm:SS-implies-SC}.
The proof is simple, but interesting, enough so that it is included for completeness.

We then show that it's possible to \textit{one-shot any code}, whose formal statement is given by Theorem \ref{thm:RG-threshold}.
As sketched in Fig. \ref{fig:schematic}, given an $[[n,k]]$ CSS QLDPC code, we can construct an $[[nm,k]]$ CSS QLDPC code with single-shot (SS) quantum error correction and an efficient local renormalization-group (RG) decoder.
The local decoder runs in $O(\log m)$ parallel time, which depends only on the overhead $m$, rather than on the total system size $nm$.
The output preserves the number of logical qubits, and its check weights $\fw$ and qubit degrees $\fd$ are bounded by those of the input code up to some absolute constants, independently of $n$ and $m$. 
More quantitatively, below a constant threshold $p_{\rm RG}>0$, the probability of failure over $T$ rounds of single-shot (SS) decoding is bounded by
\begin{equation}
    \Pr[T\text{ rounds of SS decoding fails}]
    \le O(T n)\left(\frac{p}{p_{\rm RG}}\right)^{\Omega(m^{\alpha})},
    \qquad \alpha >0,
    \label{eq:intro-main-bound}
\end{equation}
for joint local stochastic physical and measurement noise of rate $p$, with constants depending only on $\fw,\fd$.
% As in the formal noise model, failure is assessed with a final noiseless recovery. 
% Thus the error suppression exponent grows polynomially in $m$; at fixed subthreshold noise, the failure probability is exponentially small in a power of the overhead. 
We say that the \textit{threshold} is $p_{\rm RG}$, while the \textit{error suppression} is $\Omega(m^{\alpha})$.
Hence, the single-shot guarantee is sustainable over exponential many repeated rounds and does not require the input code to be single-shot. 
Equation~\eqref{eq:intro-main-bound} is the overhead form of Theorem~\ref{thm:RG-threshold}.

A second guarantee is that the RG decoder \textit{enhances} any single-shot decoder.
Specifically, if an input decoder is provided with threshold $p_c>0$ and error suppression $\Omega(n^{\beta})$, then the RG decoder enahnces the error supression to $\Omega(m^{\alpha} n^{\beta})$, while maintaining the previously introduced decoder independent threshold $p_{\rm RG}$; if $p_c<p_{\rm RG}$, then the threshold is also enhanced.
Moreover, if the input decoder runs in $\tau$ (parallel) time, then the RG decoder runs in $O(\log m) + \tau$ parallel.

% its running time is additional: the full bound is $O(\log m)+T_{\rm seed}$, as in Theorem~\ref{thm:RG-run-time}. Thus the overhead determines the cost of the local RG stages, while enhancement can retain the chosen seed decoder.

\begin{figure}[h]
\centering
\includegraphics[width=0.7\columnwidth]{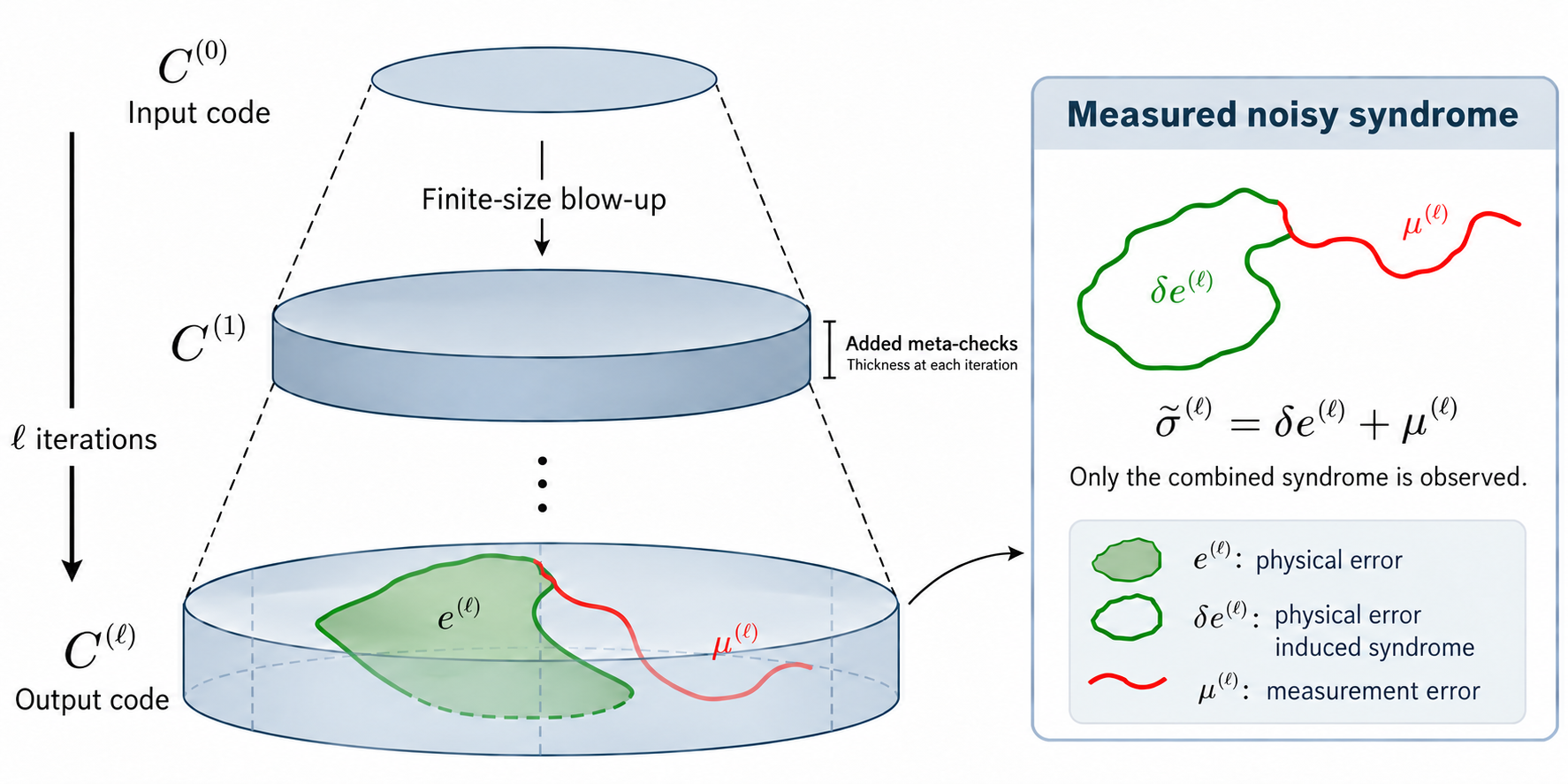}
% \subfloat[\label{fig:input-code}]{%
%     \centering
%     \includegraphics[width=0.43\columnwidth]{input-code-clean.pdf}
% }
% \hspace{0.05\columnwidth}
% \subfloat[\label{fig:local-cone}]{%
%     \centering
%     \includegraphics[width=0.43\columnwidth]{local-cone-clean.pdf}
% }
\caption{Schematic of the RG construction. Physical and measurement errors only occur in the final output code $C\super{\ell}$.
}
\label{fig:schematic}
\end{figure}

\subsection{Literature Comparison}
\label{sec:intro-literature}

Our construction generalizes the recent breakthrough of Balasubramanian, Davydova, and Lin on a passive self-correcting (SC) quantum memory in three dimensions \cite{balasubramanian2026passive}. 
They construct a three-dimensional Pauli stabilizer Hamiltonian with a memory lifetime exponential in a power of system size at sufficiently low, nonzero temperature. 
Their construction recursively replaces local pieces of a seed code and uses an RG decoder to establish protection against thermal errors. 
We build on this recursive construction and its local coarse-graining mechanism, while addressing active correction with faulty syndrome measurements. 
The locality relevant to our decoder is locality in the code; the generalization stated here does not require the input or output code to have a three-dimensional geometric realization.

The extension from passive self-correction to single-shot QEC is nontrivial. 
In the self-correction setting, the errors to be controlled are physical errors, whose exact syndromes are supplied to the recovery argument. 
In single-shot QEC, the observed syndrome also contains measurement faults and need not be a valid physical syndrome. 
Repairing it can produce a physical correction whose residual error enters the next round. 
Consequently, separate bounds on physical errors and measurement errors do not suffice: one must control their combined effect, including interference between the residual correction from one round and fresh errors in the next.
Alternatively, since we also show that any single-shot code is self-correcting, we see that single-shot is a stronger property than self-correction.

This is also where the distinction from Campbell's adversarial guarantee becomes essential \cite{campbell2019theory}: a stochastic threshold requires an entropy bound in addition to control of the weight of each residual error.
One standard approach studied in this context is to impose linear soundness or, more generally, linear confinement \cite{quintavalle2021single,xu2024constant,hong2025preparation,tan2025single}.
Roughly, linear confinement bounds the reduced error weight by its syndrome weight (up to a constant factor), motivating attempts to control the entropy through this relation.
Bombin gives another route to a stochastic single-shot threshold by controlling the entropy of connected syndrome clusters, without requiring linear confinement \cite{bombin2015single}.
Although not explicitly stated in this form, it is quite plausible that his formulation extends to QLDPC codes with \textit{sparse meta-checks}: each meta-check has bounded weight and each check participates in a bounded number of meta-checks, uniformly in the code size.
Such sparsity permits the construction of a bounded degree syndrome graph such that his componentwise validity and recovery compatibility assumptions may hold.
% Such sparsity supplies a bounded-degree syndrome graph on which to control the entropy.
The 4D toric code is an example: it has sparse meta-checks and is self-correcting \cite{alicki2010thermal}, yet a square planar error $e$ of side length $r$ has $|e|\propto r^2$ and syndrome weight $|\delta e|\propto r\propto |e|^{1/2}$, so its confinement is sublinear.
We \textbf{emphasize} that although Bombin provided empirical and qualitative analysis between the similarities, he does not prove that either property implies the other, even when restricted to \textit{sparse meta-checks}.

Our generalization of the passive 3D memory presents both obstructions.
Although the meta-checks introduced at each level are sparse, the full inherited hierarchy, expressed on the final code, is not uniformly sparse: some inherited meta-checks have weight $\Omega(m)$; see Remark~\ref{rem:hidden-meta-checks}.
The construction also has only sublinear confinement, even though the generalized code need not admit a three-dimensional embedding.
Thus neither this approach based on linear confinement nor a direct application of Bombin's bounded-degree syndrome-cluster argument establishes single-shotness for this generalization.

% We address this by combining syndrome repair with local cleaning at every level of the hierarchy. 
% Meta-checks identify violations of syndrome consistency, after which cleaning passes the repaired syndrome to the next coarser code. 
% Repair can enlarge the relevant error clusters, so the clean-syndrome contraction must be strong enough to absorb this loss. 
% In our explicit construction, a three-sheet replacement gives contraction $\kappa=4/9<1/2$, allowing noisy contraction strictly below one. 
% The proof then tracks physical faults together with measurement faults from consecutive rounds in a common witness graph. 
% Any logical failure requires a large connected witness containing a constant fraction of faults. Counting these witnesses in a bounded-degree graph controls the entropy and yields the suppression in Eq.~\eqref{eq:intro-main-bound}. 
% This controls the interaction between the two kinds of faults throughout repeated correction.

% \subsection{Fast surgery}
% \label{sec:intro-surgery}

A further motivation comes from the recent interest in \textit{fast surgery}: measuring logical operators in $O(1)$ syndrome rounds instead of the usual $O(d)$ \cite{gottesman2013fault}. 
Baspin, Berent, and Cohen give a general framework for fast QLDPC surgery using auxiliary chain complexes \cite{baspin2025fast} (later generalized for non-CSS codes \cite{yuan2026non}). 
Recent developments include constant-time surgery on hypergraph product codes with near-constant space overhead \cite{chang2026constant}, and lifted surgery for fast, parallel logical measurements on Abelian group algebra codes \cite{berent2026lifted}. 
These protocols address a central obstacle to fast logical computation, but the time spent in a surgery gadget must be distinguished from the time required to maintain its input and output memories.

Indeed, the fault-tolerance analysis of Ref.~\cite{baspin2025fast} assumes sufficiently reliable base-code syndromes before and after surgery. If the surrounding QEC obtains this reliability through $O(d)$ repeated syndrome rounds, then it can remain the bottleneck even when the surgery gadget itself takes $O(1)$ rounds. 
Our construction removes the need for such temporal repetition for memory correction: combined with a compatible constant-round surgery protocol, both surgery and the intervening QEC can use $O(1)$ quantum measurement rounds. This is a statement about quantum syndrome-extraction depth, with parallel operations on the required code connectivity; the parallel classical decoding takes $O(\log m)$ time, determined by the overhead $m$.

% There are two distinctions when composing these ingredients. First, some fast-surgery results obtain constant \textit{amortized} time by performing many operations within a longer syndrome history \cite{chang2026constant}; their guarantees are different from an isolated constant-round operation. 
% Second, a surgery gadget can introduce a random Pauli frame whose recovery requires more than control of local stochastic residual errors. 
% Single-shot QEC alone therefore does not establish single-shot state preparation or the fault tolerance of every surgery interface. Our result supplies a general single-shot memory on which compatible fast logical measurements can operate, removing repeated QEC as a separate source of quantum time overhead.

\subsection{Organization}
\label{sec:intro-organization}

Section~\ref{sec:prelim} introduces the noise model and decoder conventions. 
Section \ref{sec:SS-implies-SC} provides a simple proof showing that any single-shot stabilizer code is self-correcting; it is self-contained and thus independent of the remainder of this manuscript.
Section~\ref{sec:level} states the level-map, cleaning, and repair properties used in the analysis, with the explicit construction deferred to Appendices~\ref{sec:iterate}--\ref{sec:dual}. 
Section~\ref{sec:RG} defines the RG decoder and reduces repeated correction to an effective input level noise process. 
Section~\ref{sec:run-failure} bounds run failure through syndrome witnesses. Section~\ref{sec:threshold} proves the threshold and enhancement guarantees, culminating in Theorem \ref{thm:RG-threshold}. 
Section~\ref{sec:local-repair} develops local syndrome repair and its parallel implementation, using the geometric estimates in Appendix~\ref{sec:metric}. The run time is given by Theorem \ref{thm:RG-run-time}.

\subsection{Outlook}
\label{sec:intro-outlook}

In light of our simple yet general proof that single-shot implies self-correction, can the converse be proven in a general manner (say, for any QLDPC code)?
Indeed, Bombin \cite{bombin2015single} has given an empirical analysis showing that the two are closely related when dealing with QLDPC codes with \textit{sparse meta-checks}.
This manuscript, in conjunction of \cite{balasubramanian2026passive}, shows that even when meta-checks are non-sparse, the two properties can still closely related.
However, neither has shown a formal equivalence between the two properties, and thus begs the question whether single-shot is equal to or strictly stronger than self-correction?

The passive 3D memory construction has two ingredients: an iterative code blow-up and a Euclidean embedding which positions qubits and checks in $\mathbb{R}^3$ with bounded interaction range and bounded spatial density \cite{balasubramanian2026passive}.
Writing $R$ for the effective qubit blow-up factor per iteration, after $\ell$ iterations the overhead is $R^\ell=\Theta(m)$.
The probabilistic embedding argument, based on the Lov\'asz local lemma, gives an astronomical blow-up factor $R\gtrsim 10^{96}$, corresponding to the square of the subdivision scale $10^{24}$ and squared again due to an alternating sequence so that both $X$ and $Z$ sector errors are protected.
Our generalization focuses on the iterative code construction, i.e., no subdivision occurs, so that a 3D Euclidean embedding is not guaranteed.

% Even without subdivision, due to the alternating construction to protect both $X$ and $Z$ sectors, the blow up is still large, e.g., $R\approx 50^2$ if the surface code is used as an input.
% A natural question is whether a single iterative construction that protects both $X$ and $Z$ sectors simultaneously can be constructed with realistic blow up sizes, e.g., $R\approx 50$.

It then naturally raises the question where one can obtain a simpler 3D Euclidean construction with realistic constants, for both single-shot correction (and thus passive self-correction)?
Ref.~\cite{balasubramanian2026passive} already proposes an explicit alternative with blow up $R \gtrsim 8^4$, corresponding to a subdivision scale $8$.
However, the authors leave a complete proof of exponential memory lifetime for this alternative to future work.
Independently verifying its self-correcting properties remains nontrivial, even with AI-assisted analysis.
A simpler construction with a transparent proof of contraction, locality, and noise stability would therefore be valuable.

In the context of fast surgery, our result supplies a general single-shot memory on which compatible fast logical measurements can operate, removing repeated QEC as a separate source of quantum time overhead.
However, the memory and surgery constructions are currently treated separately, which may incur a large spatial overhead.
If the single-shot construction proves practical, incorporating fast surgery directly into it may therefore be important for reducing the overall spatial overhead.
We leave this integration to future work.

%% file: prelim.tex
\section{Preliminary}
\label{sec:prelim}
\subsection{Notation}

\begin{definition}
    Let $[n]=\{1,...,n\}$. Similarly, define $[n),[0,n],[0,n)$ to be the integers within the corresponding intervals.
\end{definition}

\begin{definition}
    Given linear map $\phi$ between $\F_2$ spaces equipped with bases, let $\|\phi\|_{\rm c}$ be the maximum column weight, and similarly for the maximum row weight $\|\phi\|_{\rm r}$.
\end{definition}
% \begin{definition}
%     If $\cS$ is a collection of subsets, we write $\|\cS\|$ to be the sum of $|s|$ for $s\in \cS$.
% \end{definition}

\begin{definition}[CSS Code]
    \label{def:CSS-code}
    By default, we shall represent a CSS code as a $\F_2$ \textbf{cocomplex} 
    \begin{equation}
        C=X\to Q\to Z
    \end{equation}
    with \textbf{codifferential} $\delta$ satisfying $\delta^2=0$ and $X,Q,Z$ denoting the space of $X$-check, qubits and $Z$-checks, equipped with basis $\cX,\cQ,\cZ$.
    We also reserve the notation $x,q,z$ for the \textbf{cells}, i.e., basis elements, in $\cX,\cQ,\cZ$, respectively, and write $x\prec q\prec z$ to denote the adjacency relation so that $\prec$ defines a partial ordering on cells $\cC\equiv \cX\sqcup \cQ\sqcup \cZ$.
    Note that any \textbf{chain} (element) $c$ in $C$ can be regarded as a subset of cells, and vice-versa.
    Write $\dnum_u c = c|_u$ to denote the \textbf{projection} of $c$ onto the cell $u$ so that $\cong \F_2$; similarly define $\dnum_{\cS} c= c|_{\cS}$ to denote that onto a subset $\cS$ of cells.
    Also write $|c|_{\cS}$ to denote the \textbf{Hamming weight} of the projection $c|_{\cS}$.
    We also write  
    \begin{equation}
        C^\top=Z\to Q\to X
    \end{equation}
    to denote the transpose of $C$ with differential $\partial = \delta^\top$.
    We also write $\|\delta\|_{\rm c}, \|\delta \|_{\rm r}$ to denote the maximum column and row weight.
    In particular, $C$ is \textbf{QLDPC} if the column and row weights of all $\delta$ are bounded; we write $\fw, \fd$ to denote the max check weights and qubit degrees, respectively.
    Unless we are referring to cell (co)complexes, we will not use the term $i$-cell, since CSS codes do not have a canonical convention.
    Instead, we refer to $X,Q,Z$ as the $X$-check, qubit, and $Z$-check \textbf{degree}; in particular, $\deg x=X$ and similarly for $q,z$.
\end{definition}

\begin{definition}[Meta-Checks]
    \label{def:meta-checks}
    In the case where the CSS code has meta-checks, we write
    \begin{equation}
        \bar{C} =  MX \to  X\to Q\to Z \to MZ
    \end{equation}
    so that the bar notation emphasizes that the code has meta-checks $\bar{C}=(\bar{C},\bar{\delta})$, while the \textbf{bare code} is simply $C=(C,\delta)$.
    Also note that if \textbf{syndrome} $\sigma\in Z$ is induced by physical errors $\in \im \delta$, we say it is \textbf{valid}, and if it's consistent with all meta-checks $\in \ker \bar{\delta}$, then it is \textbf{clean}.
    Also, we say that $\bar{C}$ is \textbf{sparse} if the degrees and codegrees are bounded, i.e., the bare code $C$ is QLDPC and the meta-checks $\bar{\delta}$ have bounded column and row weights; we write $\bar{\fw},\bar{\fd}$ for the max meta-check weight and meta-degrees, respectively, to distinguish from the bare $\fw,\fd$.
\end{definition}

\begin{definition}[Repetition]
    Let $R=V\to E$ denote the repetition code on a path with $m\ge 2$ vertices, though for the main result, we will utilize $m=3$ since it is more natural as elaborated in Remark \ref{rem:m-sheets}.
    % It may be worth mentioning that the construction also applies to any tree graph.
\end{definition}

% % ==========================================================
\subsection{Decoder}

% \begin{definition}[MW]
%     Let $\bar{\delta}:Z\to MZ$ be given. Then the \textbf{min-weight (MW)} projection $\im \bar{\delta} \to Z$ is the min-weight $\hat{\mu}\in Z$ such that $\bar{\delta}\hat{\mu}$ is equal to the input.
%     Note that we may extend its domain from $\im \delta$ to $Z$ in some arbitrary manner; here, we take the convention that if the input $\notin \im \bar{\delta}$, the MW projection \textbf{fails}.
%     One can similarly define MW projection relative to $\delta:Q\to Z$, which is commonly used to prove a threshold for noiseless syndrome error correction \cite{gottesman2013fault}, though it is not used in this manuscript.
% \end{definition}

\begin{definition}[Decoder]
    \label{def:decoder}
    Let $C$ be a CSS code.
    A \textbf{(valid) decoder} is a map $\cE:\run\to Q$ defined on a subset $\run \subseteq Z$ which contains the valid syndromes $\im \delta \subseteq \run$, and maps $0\mapsto 0$ and satifies $\delta \cE = \mathrm{id}$ on $\im \delta$.
    We say that $\cE$ \textbf{fails (to run)} on a syndrome $\noisy{\sigma} \in Z$ if $\noisy{\sigma}\notin \run$.
    Even if a decoder runs successfully, it's possible to exhibit a \textbf{logical failure}, i.e., $e+\cE( \delta e) \notin \im \delta$.
    Otherwise, we say it \textbf{decodes} physical error $e$ \textbf{successfully}.
\end{definition}

It is conventional to start with a decoder with domain $ =\im\delta$.
To deal with noisy syndromes $\noisy{\sigma}\notin \im\delta$, the domain is then somehow \textit{enlarged}. One way is to utilize meta-checks as follows.
\begin{definition}[Meta-Check Decoder]
    \label{def:meta-check-decoder}
    Let $\bar{C}$ be a CSS code with meta-checks.
    % A \textbf{syndrome decoder} is a map $\dec: Z\to (Q,Z)$ which maps $0\mapsto (0,0)$ and $\delta \cE + \cM = \rm{id}$.
    A \textbf{(valid) meta-check decoder} is a pair of maps $\cM: Z\to Z$ and $\cE: \run \to Q$ such that 
    \begin{itemize}
        \item $\cM$ is invariant under valid syndromes, i.e., $\cM(\noisy{\sigma} + \delta e) = \cM(\noisy{\sigma})$, maps $0\mapsto 0$ and
        \begin{equation}
            \bar{\delta} \cM (\noisy{\sigma}) = \bar{\delta} \noisy{\sigma},
        \end{equation}
        \item $\cE$ is a valid decoder with domain $\run = \{\noisy{\sigma}:\noisy{\sigma} + \cM(\noisy{\sigma}) \in \im \delta\}$ and
        \begin{equation}
            \cE (\noisy{\sigma}) = \cE(\noisy{\sigma}+\cM(\noisy{\sigma}))
        \end{equation}
    \end{itemize}
    
\end{definition}

\begin{remark}[Run Failure]
    \label{rem:run-failure}
    Note that the meta-check decoders runs successfully with input $\noisy{\sigma} = \delta e + \mu$ iff it runs successfully on $\mu$.
    In a standard noise model, it's conventional to assume that the physical error $e$ and measurement error $\mu$ are jointly local stochastic and thus there is no \textit{interference} between $e,\mu$ when computing the run success (or failure) event; instead, it only depends on $\mu$.
    Hence, we shall focus on the logical failure, while the run failure is considered separately.
\end{remark}

\begin{example}[MW]
    \label{ex:MW}
    The min-weight (MW) decoder relative to $\delta:Q\to Z$ \cite{gottesman2013fault} is a valid decoder. Moreover, if a CSS code has meta-checks (e.g., 3D toric code), then the corresponding MW decoder defined as follows is a valid meta-check decoder; given $\noisy{\sigma}\in Z$, apply MW relative to meta-checks to obtain $\cM (\noisy{\sigma})$, and then apply MW relative to checks on the repaired syndrome $\noisy{\sigma} +\cM (\noisy{\sigma})$ to obtain $\cE(\noisy{\sigma})$. 
    Note that the MW meta-check decoder fails to run if $\noisy{\sigma} +\cM(\noisy{\sigma})$ is clean but not valid.
    However, it's run failure probability is generally well-controlled by the large syndrome distance, i.e., the min weight of clean but not valid syndromes, and thus it's conventional to focus on the logical failure.
\end{example}

\begin{remark}[Valid vs Clean Invariance]
    Note that the valid syndrome invariance of $\cM$ is a weaker property than $\cM$ depending only on the meta-syndrome $\bar{\delta}\tilde{\sigma}$, i.e., invariant under clean syndromes.
    The MW decoder in Example \ref{ex:MW} satisfies the latter stronger property, but our renormalization group (RG) decoder will only possess the former.
    Note that by Remark \ref{rem:run-failure}, the weaker invariance is sufficient when considering the (run) failure probability.
\end{remark}

\begin{definition}[Connectivity]
    Note that $\bar{\delta}:Z\to MZ$ induces a graph $Z(\bar{\delta})$ on $z$ checks with edges $zz'$ if $z,z'$ are in a common meta-check.
    In particular, we can define a \textbf{connectivity} relative to the meta-checks.
    % The case is similar for say $\delta:Q\to Z$, in which we define \textbf{clusters} relative to $Z$-checks.
\end{definition}

\begin{lemma}[Half-Fault]
    \label{lem:half-fault}
    Let $\mu$ be a syndrome with MW projection $\hat{\mu}$, and let $\fs$ be a connected component (relative to meta-checks) of $\mu+\hat{\mu}$. Then
    \begin{equation}
        |\fs| \le 2|\fs \cap \mu|
    \end{equation}
\end{lemma}
\begin{proof}
    Each component $\fs$ is clean, so $\hat{\mu}+\fs$ has the same meta-syndrome as $\hat\mu$. 
    Minimality gives $|\hat\mu\cap\fs|\le |\fs|/2$. Since $\mu$ and $\hat\mu$ partition $\fs$, the statement follows.
\end{proof}

Since we will be dealing with CSS codes, the $X$ and $Z$ sector can be studied separately even when their noise is correlated. Hence, we will restrict our attention to the $Z$ sector.
\begin{definition}[Noise Model]
    \label{def:noise-model}
    Consider a CSS code $C$ equipped with a decoder $\cE$ as in Definition \ref{def:decoder}.
    Suppose that physical errors $e(t)$ occur at integer times $t\in [0,T]$, while measurement errors $\mu(t^+)$ occur at half-integers $t^\pm =t \pm 1/2$, with a (joint) \textbf{local stochastic (LS)} rate $p$.
    Specifically, any finite set $\cS$ of physical and measurement fault locations (including their times) satisfies
    \begin{equation}
        \Pr[\cS \text{ fault}] \le p^{|\cS|}
    \end{equation}
    Proceed by induction on $t$, and suppose that after $t$ rounds, the residue error is $r(t)$, where $r(0)=0$, so that the noisy ($Z$-sector) syndrome is 
    \begin{equation}
        \noisy{\sigma}(t^+) \equiv  \delta (r(t) + e(t)) +\mu(t^+)    
    \end{equation}
    Apply the meta-check decoder on $\noisy{\sigma}(t^+)$, and apply the correction $\cE$ so that the residual error after $t+1$ rounds is
    \begin{equation}
        r(t+1) = r(t) + e(t) +\cE(\noisy{\sigma}(t^+))
    \end{equation}
    The recurrence is well-defined only if the decoder runs successfully on $\noisy{\sigma}(t^+)$.
    We also assume the \textbf{boundary condition} that after $T$ rounds, the $T+1$ measurement is perfect, i.e., $\mu(T^+)=0$.
    In this case, $T$ rounds of ($Z$-sector) \textbf{single-shot (SS)} decoding \textbf{succeeds} if $r(T+1) =0$ modulo stabilizers.
    We thus say that the code $C$ with decoder $\cE,\cM$ is \textbf{single-shot} if there exists constants $p_c,c,a>0$ such that if the error rate $p<p_c$, then
    \begin{equation}
        \Pr[T \text{ rounds of SS decoding fails}] \le O(T+1) \poly(|\cQ|) \left(\frac{p}{p_c}\right)^{c n^a}
    \end{equation}
    We call $p_c$ the (SS) \textbf{threshold} and $\Omega(n^a)$ the (SS) \textbf{error suppression}.
\end{definition}

\subsection{Covers}

\begin{definition}[Cover]
    Let $\cG$ be a graph of bounded degree, and let $\cA,\cS$ be subsets of vertices.
    The $R$-\textbf{neighborhood} $\cN_R (\cS)$ of $\cS$ is the collection of vertices $u$ such that the graph distance $\dist(u,v)\le R$ for some $v\in \cS$.
    We say that $\cA$ is $R$-\textbf{covered} by $\cS$ (relative to $\cG$) if $\cA$ is within the $R$-\textbf{neighborhood} of $\cS$, i.e., $\cA\subseteq \cN_R(\cS)$.
    % Moreover, we say the $R$-cover $\cS$ is \textbf{minimal} relative to $\cA$ if no proper subsets is an $R$-cover, \textbf{matchable} if there exists injective $\iota:\cS\to \cA$ such that $\iota v \in \cN(v)$.
    If $R=1$, we simply say the neighborhood $\cN(\cS)$ and cover.
\end{definition}
\begin{remark}[$R=1$ is General]
    Note it's sufficient to prove statement for $R=1$, since for any finite $R$, we can simply add edges to the graph so that any vertices originally with $R$ distance are now adjacent in the new graph. The new graph would then have bounded degree $\le (1+\deg \cG)^R -1 $.
\end{remark}
\begin{example}[Code Graph]
    \label{ex:code-graph}
    Consider a QLDPC CSS code $C$ with max check weight $\fw$ and qubit degree $\fd$.
    Then the bounded degree graph $\cG(C)$ \textbf{induced} by $C$ is that with vertices $\cC$ and edges $uv$ if $u\prec v$.
    The motivation here is that by definition, if measurement error $\mu\subseteq \cZ$ is local stochastic, then $\Pr[\tau \text{ covered by } \mu]\le p^{|\tau|}$ for any fixed $\tau \subseteq \cZ$.
    Hence, we wish to obtain a similar property for physical errors $e \subseteq \cQ$, so that physical and measurement errors can be treated analogously, as shown in Lemma \ref{lem:cover-probability}.
    One may also define $Z$ \textbf{graphs} $\cG_Z(C)$ with vertices $\cZ$ and edges $zz'$ if there exists cell $w\in \cC$ such that $w\preceq z,z'$.
\end{example}

\begin{lemma}[Cover Size]
    \label{lem:cover-size}
    Let $\cS$ be a cover of $\cA$ relative to graph $\cG$. Then
    \begin{equation}
        |\cA| \le |\cS| (1+\deg \cG) =O(|\cS|)
    \end{equation}
\end{lemma}

\begin{proof}
    Note that
    \begin{align}
        |\cA| &= \sum_{u} 1\{u\in \cA\} \\
        &\le \sum_{u} 1\{u\in \cA\} \sum_{v\in \cS} 1\{u \in \cN(v)\} \\
        &\le |\cS| (1+\deg \cG )
    \end{align}
\end{proof}

\begin{definition}[Independence]
    A subset $\cS$ of vertices in $\cG$ is \textbf{independent} if no vertices in $\cS$ are adjacent.
    A \textbf{maximally independent} subset of $\cS$ is an independent subset of $\cS$ such that adding any more vertex in $\cS$ to the subset causes it be dependent.
\end{definition}

\begin{lemma}[Cover Probability]
    \label{lem:cover-probability}
    Let $e,\mu$ be jointly stochastic with rate $p$ and let $\tau$ be a syndrome.
    Then there exists $p_{\rm cover},c>0$ (depending only on qubit degree $\fd$) such that if $p<p_{\rm cover}$, then
    \begin{equation}
        \Pr[\tau \text{ covered by } e\sqcup \mu] \le \left(\frac{p}{p_{\rm cover}}\right)^{c|\tau|}
    \end{equation}
\end{lemma}
\begin{proof}
    For simplicity, we shall write $c,c'$  as constants that may change from line to line, but only depend on $\fd$.
    If $\tau$ is covered by $e\sqcup \mu$, then let $\cS$ be the maximally independent subset of $\cN(\tau)\cap (e\sqcup \mu)$ (say, determined by a greedy algorithm).
    Then $\cS \subseteq \cN(\tau)\subseteq \cN(\cS)$ and $\cS\subseteq e\sqcup \mu$.
    Hence, $\tau \subseteq \cN_2(\cS)$ and thus by Lemma \ref{lem:cover-size}, 
    \begin{align}
        \Pr[\tau \text{ covered by } e\sqcup \mu] &\le \sum_{\cS} 1\{\cS\subseteq \cN(\tau) \subseteq \cN(\cS)\} \Pr[\cS \subseteq e\sqcup \mu] \\
        &\le p^{c|\tau|} \sum_{\cS}  1\{\cS\subseteq \cN(\tau)\}\\
        &\le p^{c|\tau| } 2^{c'|\tau|}\\
        &= \left(\frac{p}{p_{\rm cover}}\right)^{c|\tau|}
    \end{align}
    where we only considered the subgraph $\cQ\sqcup \cZ$ of $\cG(C)$ (which has degree $\le \fd$), since it contains all realization of errors $e\sqcup \mu$.
    We note that the constants can be further optimized, but for exposition purposes, we have taken this simpler approach.
\end{proof}

% \begin{remark}
%     As suggested by the noise model in Definition \ref{def:noise-model} and Lemma \ref{lem:SS-reduction}, we will care about the joint stochastic rate of $e,\mu_\pm$. 
%     Hence, the previous definition and lemmas can be extended by considering \textit{temporal graphs} with vertices $\cZ^-\sqcup \cQ \sqcup \cZ^+$ and edges $z^-q,qz^+$ if $q\prec z^\pm$ , culminating in an analogous bound to Lemma \ref{lem:cover-probability} given by
%     \begin{equation}
%         \Pr[\tau \text{ covered by } \mu^-\sqcup e \sqcup \mu^+] \le \left(\frac{p}{p_c}\right)^{|\tau|/\fd}
%     \end{equation}
% \end{remark}

\begin{remark}[Cover Equivalence]
    \label{rem:cover-equivalence}
    We note that for QLDPC codes, cover probability and local stochasticity are equivalent, up to constants.
    Indeed, Lemma \ref{lem:cover-probability} shows that local stochastic errors implies the cover probability.
    The converse can be similarly shown by note that
    \begin{equation}
        \Pr[f\subseteq e]\le \Pr [\cN(f) \text{ covered by } e]
    \end{equation}
    and that $|\cN(f)| = \Omega(|f|)$ by Lemma \ref{lem:cover-size}.
\end{remark}

%% file: SS-to-SC.tex
\section{SS implies SC}
\label{sec:SS-implies-SC}

The following section aims to prove that any single-shot code is self-correcting, culminating in Theorem \ref{thm:SS-implies-SC}.
The section is independent from the main result in Theorem \ref{thm:RG-threshold}, and thus self-contained provided that the reader is familiar with the noise model in Definition \ref{def:noise-model}.
Since it is closely related to our main result, possibly interesting to a broad audience and simple enough to prove, we include it for completeness.
The reader may also wish to be familiar with the basic definition of self-correction, which we refer to Refs \cite{chesi2010thermodynamic,balasubramanian2026passive} for a review.
% It is simple enough to prove and closely related to our study of single-shot properties, it warrants being included in this manuscript.
As a corollary of Theorem \ref{thm:RG-threshold} and \ref{thm:SS-implies-SC}, it's clear that any code can also be made self-correcting.

\subsection{Gibbs is LS}

We first show that the Gibbs measure is local stochastic as in Lemma \ref{lem:Gibbs-is-LS}.
\begin{lemma}[MacWilliams, \cite{macwilliams1977theory}]
    \label{lem:macwilliam}
    Let $W\subseteq \F_2^N$ be a subspace, $v\in \F_2^N$ and $0\le y < 1$. Then
    \begin{equation}
        \sum_{w\in W} y^{|v+w|} \le \sum_{w\in W} y^{|w|}
    \end{equation}
    where equality is achieved only if $v\in W$.
\end{lemma}
\begin{proof}
    Regard $\F_2^{N}$ as a finite abelian group so that the Fourier transform is well-defined. Let $f(v)=x^{|v|}$ and note that the Fourier transform is given by
    \begin{align}
        \hat{f}(u) &= \sum_{v\in \F_2^N} u^{|v|} (-1)^{u\cdot v} = \prod_{i=1}^N \sum_{v\in \F_2} (-1)^{u_i \cdot v} y^v \\
        &= (1-y)^{|u|} (1+y)^{N-|u|} > 0
    \end{align}
    Let $1_W(v)=1\{v\in W\}$. Then it has Fourier transform
    \begin{align}
        \hat{1}_W(u) &= \sum_{w\in W} (-1)^{u\cdot w} \\
        &= \prod_{i=1}^{\dim W} (1+(-1)^{u\cdot w_i}) \\
        &= |W| 1_{W^\perp}(u)
    \end{align}
    where $w_i,i=1,...,\dim W$ is a basis of $W$. Hence, we see that
    \begin{align}
        \sum_{w} y^{|v+w|} &= (f*1_W)(v) \\
        &= \frac{1}{2^N} \sum_{u\in \F_2^N} (-1)^{u\cdot v} \hat{f}(u) \hat{1}_W(u) \\
        &= \frac{1}{|W^\perp|} \sum_{u\in W^\perp} (-1)^{u\cdot v} \hat{f}(u)
    \end{align}
    Note that each term in the summation is positive, i.e., $\hat{f}(u)> 0$, coupled to a sign $(-1)^{u\cdot v}$.
    Hence, it is maximized with $v\in W^{\perp \perp} =W$, and thus the statement follows.
\end{proof}

\begin{lemma}[Gibbs is LS]
    \label{lem:Gibbs-is-LS}
    Let $\Pr_{\beta}$ denote the \textbf{Gibbs measure} at \textbf{inverse temperature} $\beta$ over valid syndromes $\sigma\in \im \delta$ of a CSS code, i.e.,
    \begin{equation}
        \Pr_{\beta}[\sigma] =\frac{1}{Z_\beta} e^{-\beta|\sigma|}
    \end{equation}
    where $Z_{\beta}$ is the \textbf{partition function} which normalizes the probability. Then $\sigma$ is local stochastic with rate $e^{-\beta}$, i.e., for any fixed subset $K$ of checks, 
    \begin{equation}
        \label{eq:Gibbs-is-LS}
        \Pr_{\beta}[K\subseteq \sigma] \le e^{-\beta |\sigma|}
    \end{equation}
\end{lemma}
\begin{remark}
    Note that a similar inequality to Eq. \eqref{eq:Gibbs-is-LS} was proven in Theorem 5.13 of \cite{balasubramanian2026passive}. However, it only applies to the special case where $K$ is also a valid syndrome.
    The general case follows from the key observation is Lemma \ref{lem:macwilliam}.
\end{remark}
\begin{proof}
    Note that
    \begin{align}
        \Pr_{\beta}[ \cK \subseteq \sigma] &= \frac{1}{Z_\beta} \sum_{\sigma \in \im \delta} 1\{\cK \subseteq \sigma\} e^{-\beta|\sigma|} \\
        &= \frac{1}{Z_\beta} \sum_{\sigma \in \im \delta} 1\{\cK \subseteq \sigma\} e^{-\beta|\cK|} e^{-\beta |\sigma+\cK|} \\
        &\le e^{-\beta |\cK|} \frac{1}{Z_{\beta}} \sum_{\sigma\in \im \delta} e^{-\beta|\sigma +\cK|} \\
        &\le e^{-\beta |\cK|}
    \end{align}
    where the last inequality uses the MacWilliam identity in Lemma \ref{lem:macwilliam}.
\end{proof}

\subsection{SS implies Stability}

\begin{definition}[Stability, \cite{chesi2010thermodynamic,balasubramanian2026passive}] 
    Consider a CSS code $C$ with decoder $\cE$. A valid syndrome $\sigma \in \im\delta $ is \textbf{stable} (relative to $\cE$) if for any single qubit $q\in \cQ$,
    \begin{equation}
        \cE(\sigma) + \cE(\sigma +\delta q) + q\in \im \delta
    \end{equation}
    If $\sigma$ is unstable, we may write pair $(\sigma,q)$ to be \textbf{unstable} to emphasize which qubit causes a nontrivial logical error.
\end{definition}
\begin{theorem}[SS implies SC]
    \label{thm:SS-implies-SC}
    If a CSS code $C$ with decoder $\cE$ is single-shot, then the Hamiltonian defined from the CSS code $C$ with decoder $\cE$ is self-correcting (SC).
    In particular, there exists (jointly LS) noise model such that
    \begin{equation}
        \Pr[T=1 \text{ round of SS decoding fails}] = p \Pr_{\beta} [\sigma \text{ unstable}]
    \end{equation}
    where $p=e^{-\beta}$ is the rate.
\end{theorem}

\begin{remark}[Straightforward Extension]
    Note that the proof can be extended to any stabilizer code $C$ in a straightforward fashion. In fact, the code doesn't even have to be QLDPC, though in this case, the \textit{physical meaning} of self-correction is debatable.
\end{remark}
\begin{proof}
    Let $\sigma$ denotes the random variable with Gibbs measure at inverse temperature $\beta$ satisfying $p=e^{-\beta}$. 
    For every unstable valid syndrome $\sigma$, let $q(\sigma)$ denote a single qubit such that the pair $\sigma,q(\sigma)$ is unstable; $q(\sigma)$ is set to $=0$ by default if $\sigma$ is stable.
    Consider the noise model in Definition \ref{def:noise-model} over $T=1$ round with errors
    \begin{equation}
        e(0) = 0, \quad \mu(0^+)=\sigma,\quad e(1) = \xi q(\sigma)
    \end{equation}
    where $\xi$ is an independent Bernoulli parameter with rate $p$.
    We claim that this noise model is local stochastic with rate $p$.
    Indeed, let $\cS= \cS(0^+)\sqcup \cS(1) \subseteq \cZ \sqcup \cQ$ denote fault locations.
    If $\cS(1)$ is of size 2 or empty, then the claim follows by Lemma \ref{lem:Gibbs-is-LS}. Hence, consider the case where $\cS(1)$ has a unique element $q$. Then
    \begin{equation}
        \Pr[\cS \subseteq \mu(0^+)\sqcup e(1) ]= \Pr[\cS(0^+)\subseteq \sigma, q(1)=q(\sigma),\xi=1] \le p \Pr[\cS(0^+)\subseteq \sigma]
    \end{equation}
    Hence, the noise model is local stochastic with rate $p$.
    
    Note that the measurement error $\mu(0^+)$ is a valid syndrome and thus we always have run success.
    Hence, the residue at time $t=1$ is $r(1)=\cE(\sigma)$, while the residue at the boundary time $t=T+1=2$ is
    \begin{equation}
        r(2)= \cE(\sigma) + \xi q(\sigma) + \cE(\sigma + \xi \delta q(\sigma))
    \end{equation}
    Hence, under this noise model,
    \begin{align}
        \Pr[T=1 \text{ round of SS decoding fails}] &= \Pr[\cE(\sigma) + \xi q(\sigma) + \cE(\sigma + \xi \delta q(\sigma))\notin \im \delta ] \\
        &= p \Pr_{\beta}[\sigma \text{ unstable}]
    \end{align}
    By Ref. \cite{chesi2010thermodynamic}, the claim then follows.
\end{proof}

%% file: level.tex
\section{Level Data}
\label{sec:level}
The point of this section (specifically, Theorem \ref{thm:validity}) is to summarize the minimal conditions sufficient to derive the single-shot properties of the RG decoder, so that this section can be read independently of the explicit construction.

\subsection{Level Map}

We first introduce the coding framework \cite{yuan2026unified} in a somewhat different format to fit the exposition of the RG decoder.
\begin{definition}[Level Map]
    \label{def:level-map}
    Let $C$ be a CSS code with meta-check extension $\bar{C}$.
    Let $\bar{\Theta}=\bar{\Theta}(C)$ be a CSS code with meta-checks, $\Theta$ be that without the meta-checks.
    Then $C\mapsto \bar{\Theta}(C)$ is a \textbf{level map} if
    \begin{enumerate}
        \item (Layers) Cells $u'$ in $\bar{\Theta}$ can be partitioned into \textbf{layers} (complexes) $\bar{Y}^u$ over cells $u\in \cC$ such that $\bar{Y}^u$ is of bounded size $\Delta_{\sh}$, and the \textbf{shadow} map defined by $u'\in \bar{Y}^{\smash{\sh u'}}$ is monotone, i.e., if $u'\preceq v'$, then $\sh u'\preceq \sh v'$.
        We adopt the notation
        \begin{itemize}
            \item For a chain $c$ in $\bar{\Theta}$, write $\dnum^u c= c|^{u}$ to denote its restriction to layer $Y^u$ and similarly define $\dnum^{\cS} c=c|^{\cS}$ for a subset of cells $\cS$ in $C$, and $|c|^{\cS}$ denote the Hamming weight of $c|^{\cS}$.
            \item A map $h$ on chains in $\bar{\Theta}$ is \textbf{layer-local} if $\dnum^u h(c)$ only depends on $c$ implicitly via $\dnum^{\preceq u}c$. Similarly extend the definition to maps $\bar{\Theta}\to C$ (and $C\to \bar{\Theta}$) by treating each cell in $C$ as a \textit{layer} whose shadow is itself, i.e., replace $\dnum^u$ with $\dnum_u$.
        \end{itemize}
        
        % \item (Meta-Checks) Two $z$ checks in $\Theta(C)$ in a common same meta-check have comparable shadows.
        \item (SDR) $C$ is a \textbf{strong deformation retract (SDR)} of $\bar{\Theta}$, i.e., there exists layer-local $(\iota,\pi,\phi)$ as sketched below
        \begin{equation}
        \label{eq:SDR}
        \begin{tikzpicture}[baseline]
        \matrix(a)[matrix of math nodes, nodes in empty cells, nodes={minimum size=25pt},
        row sep=2em, column sep=3em,
        text height=1.25ex, text depth=0.25ex]
        {&X& Q  & Z  &\\
         MX(\bar{\Theta}) & X(\Theta) & Q(\Theta)  & Z(\Theta) & MZ(\bar{\Theta})\\};
        \path[-left to,font=\scriptsize,transform canvas={yshift=0.2ex}]
        (a-1-2) edge (a-1-3)
        (a-1-3) edge (a-1-4)
        (a-2-1) edge (a-2-2)
        (a-2-2) edge (a-2-3)
        (a-2-3) edge (a-2-4)
        (a-2-4) edge (a-2-5);
        \path[left to-,font=\scriptsize,transform canvas={yshift=-0.2ex}]
        (a-2-1) edge node[below]{$\phi$}  (a-2-2)
        (a-2-2) edge node[below]{$\phi$}  (a-2-3)
        (a-2-3) edge node[below]{$\phi$}  (a-2-4)
        (a-2-4) edge node[below]{$\phi$}  (a-2-5);
        \path[-left to,font=\scriptsize,transform canvas={xshift=0.2ex}]
        (a-1-2) edge node[right]{$\iota$}  (a-2-2)
        (a-1-3) edge node[right]{$\iota$}  (a-2-3)
        (a-1-4) edge node[right]{$\iota$}  (a-2-4);
        \path[left to-,font=\scriptsize,transform canvas={xshift=-0.2ex}]
        (a-1-2) edge node[left]{$\pi$}  (a-2-2)
        (a-1-3) edge node[left]{$\pi$}  (a-2-3)
        (a-1-4) edge node[left]{$\pi$}  (a-2-4);
        \end{tikzpicture}
        \end{equation}
        such that \textbf{inclusion} $\iota$ and \textbf{projection} $\pi$ are chain maps, $\phi^2=\pi \phi =\phi \iota=0$ and
        \begin{equation}
            \pi \iota = \mathrm{id}_{C}, \quad \iota \pi =\mathrm{id}_{\bar{\Theta}} + \phi\bar{\delta} +\bar{\delta} \phi
        \end{equation}
    \end{enumerate}
    The sequence (or hierarchy) of codes $\bar{\bm{C}}$ defined inductively by $\bar{C}\super{i}=\bar{\Theta}(C\super{i-1}),i\in [\ell]$ is the \textbf{(hierarchy) tower} obtained from the \textbf{seed code} $C\super{0}$.
    We further say that $\bar{\Theta}$ \textbf{preserves sparsity} if $C\super{0}$ being QLDPC implies that the codes $\bar{C}\super{i}$ are (uniformly) sparse, i.e., with column and row weights bounded independent of $i$.
\end{definition}

\begin{remark}[Relation with \cite{yuan2026unified}]
    \label{rem:relation-with-embedding-framework}
    If $\bar{\Theta}$ is a level map, then it's clear that $\bar{\Theta}$ is a (height-2) cone in the code embedding framework \cite{yuan2026unified}; indeed, monotonicity of the shadow map implies that the codifferentials $\bar{\delta}$ of $\bar{\Theta}$ are lower-triangular (relative to the decomposition $\bar{Y}^\cX\oplus \bar{Y}^\cQ\oplus \bar{Y}^\cZ$).
    Also note that if a map is linear, then layer-locality coincides with lower-triangularity (e.g., $\bar{\delta}$ is layer-local), but since we will also deal with nonlinear maps, the layer-local definition is more appropriate.
    Also, each layer $\bar{Y}^u$, viewed as a CSS code with meta-checks, is \textbf{contractible} in the following sense; 
    \begin{equation}
        H_{\bullet}(\bar{Y}^u) \cong \F_2 1\{\bullet = \deg u\}
    \end{equation}
    This follows immediately from showing that the restriction $\dnum^u\iota \dnum_u,\dnum_u \pi \dnum^u, \dnum^u \phi\dnum^u$ satisfy the SDR properties.
    The converse also holds, provided that each layer is contractible, which follows from the homological perturbation lemma \cite{balasubramanian2026passive}.
    Examples such as the Layer codes \cite{williamson2023layer,yuan20264d} and the \textit{almost}-optimal geometrically local codes \cite{li2026almost,lin2023geometrically} are level maps, provided that we remove the restriction that each layer is bounded in size, since no hierarchy tower is constructed in these examples.
\end{remark}
\begin{remark}[Hidden Meta-Checks]
    \label{rem:hidden-meta-checks}
    Note that $\bar{\Theta}$ also \textit{implicitly} inherits the meta-checks of $\bar{C}$ in the following sense; as in Eq. \eqref{eq:SDR}, note that $\bar{\delta}\pi:Z(\Theta)\to MZ(\bar{C})$ satisfies $\bar{\delta}\pi\delta=0$ and thus can be added to $\bar{\Theta}$ in addition to its original meta-checks $MZ(\bar{\Theta})$.
    However, the inherited meta-checks will generally increase in weight, and thus if a tower $\bar{\bm{C}}$ is constructed, then the inherited meta-checks from previous levels (e.g., $0\to \ell$) will be unbounded in weight.
    Hence, we shall generally not inherit meta-checks.
\end{remark}

\subsection{Cleaning}
\begin{definition}[Cleaning]
    \label{def:cleaning}
    A level map $\bar{\Theta}$ is equipped with a ($Z$-sector) \textbf{cleaning} $f: Z(\Theta)\cap \ker \bar{\delta} \to Q(\Theta)$ if
    \begin{enumerate}
        \item (Locality) $f$ is layer-local (but not necessarily linear) with $f(0)=0$
        \item (Descent) $\sigma +\delta f(\sigma)\in \im \iota$ and if the \textbf{(clean syndrome) descent} $\desc \sigma \equiv \pi (\sigma +\delta f(\sigma))$, then
        \begin{equation}
            |\desc\sigma| \le \kappa |\sigma|
        \end{equation}
        with \textbf{contraction} $\kappa$.
        Also let the \textbf{physical descent} $\desc e \equiv \pi(e+f(\delta e))$ so that $\delta \desc = \desc \delta$
        \item (Passthrough) If $\sigma$ is a clean syndrome and $\noisy{\tau}\in Z(C)$, then $f(\sigma + \iota \noisy{\tau}) = f(\sigma)$.
        In particular, this implies similar passthrough properties for the syndrome and physical descents.
        \item (Bounded Fill) There exists constant $c_f>0$ such that for any clean syndrome $\sigma$,
        \begin{equation}
            |\sh \{\succeq f(\sigma)\} \backslash \sh \sigma| \le c_f|\sigma|
        \end{equation}
    \end{enumerate}
    An $X$-sector cleaning can be similarly defined, but since decoding $X$ and $Z$ sectors can be independently studied, we shall restrict our attention to the $Z$ sector.
\end{definition}
\begin{remark}[Descent]
    Note that
    \begin{itemize}
        \item the descent is linear modulo $\im \delta$, i.e., $\desc (\sigma+\tau) = \desc \sigma +\desc \tau +\im \delta$ for clean syndromes $\sigma,\tau$.
        \item the descent satisfies $\iota \desc e +e+f(\delta e) \in \im \delta$
    \end{itemize}
    
\end{remark}

\begin{remark}[Descent is Noisy]
    Note that the projection $\pi$ is a chain map only into $C$ and not into $\bar{C}$
    Hence, $\bar{\delta} \desc \ne 0$ in general since we cannot commute $\bar{\delta},\pi$ at the meta-check level.
\end{remark}
\begin{lemma}[Composition]
    \label{lem:level-map-composition}
    Let $\bar{\Theta}_1,\bar{\Theta}_2$ be level maps with corresponding items. Define the composition $\bar{\Theta} =\bar{\Theta}_2\circ \bar{\Theta}_1$ as the CSS code $\Theta_2 \circ \Theta_1$ with $X$ and $Z$-sector meta-checks\footnote{The output $\bar{\Theta}$ includes the \textit{intermediate} meta-checks in $\bar{\Theta}_1$. This does not contradict Remark \ref{rem:hidden-meta-checks} since we only \textit{inherit} a bounded number of times.} 
    \begin{equation}
        \bar{\delta} = 
        \begin{pmatrix}
            \iota_2 \bar{\delta}_1 & \bar{\delta}_2
        \end{pmatrix},\quad
        \bar{\delta}=
        \begin{pmatrix}
            \bar{\delta}_1 \pi_2 \\
            \bar{\delta}_2
        \end{pmatrix}
    \end{equation}
    Then
    $\bar{\Theta}$ is a level map with items
    \begin{align}
        \sh   = \sh_1\circ \sh_2&, \\
        \iota = \iota_2\iota_1&,\quad \pi=\pi_1\pi_2 \\
        f(\sigma) = f_2(\sigma) + \iota_2f_1(\desc_2 \sigma)&, \quad \desc = \desc_1\circ \desc_2 \\
        \kappa=\kappa_1\kappa_2&, \quad c_{f} = c_{f,2} +\kappa_2 c_{f,1}
    \end{align}
\end{lemma}

\begin{proof}
    Most equalities are straightforward (though tedious) to check and thus omitted.
    The only somewhat tricky one is the fill bound $c_f$, and thus we elaborate. Note that
    \begin{equation}
        \sh \{\succeq f(\sigma)\} \subseteq \sh \{\succeq f_2(\sigma)\} \cup \sh_1 \{\succeq f_1(\desc_2 \sigma)\}
    \end{equation}
    Also note that by definition
    \begin{equation}
        \sh_1 \desc_2\sigma \subseteq \sh \sigma \cup \sh \{\succeq f_2(\sigma)\}
    \end{equation}
    Combining the two implies that we can choose $c_f =c_{f,2} +\kappa_2 c_{f,1}$.
\end{proof}

% The point of Lemma \ref{lem:level-map-composition} is that if a single level map does not have contraction $\kappa<1/2$
\begin{theorem}[Validity]
    \label{thm:validity}
    Let $\bar{\Gamma}$ denote the level map in Definition \ref{def:gluing} with $m=3$ sheets.
    Then $\bar{\Theta}=\bar{\Gamma}_*\circ\bar{\Gamma}$ is a level map which preserves sparsity and has $Z$-sector cleaning $f$, with contraction $\kappa=4/9$ and a (uniform) fill bound $c_f<3$.
\end{theorem}

\begin{proof}
    Note that by Theorem \ref{thm:contraction}, \ref{thm:fill-bound}, and \ref{thm:dual}, the level $\bar{\Gamma}$ with its associated cleaning has contraction $\kappa_1=4/9$ and fill bound $c_{f,1} <3$, while $\bar{\Gamma}_*$ has contraction $\kappa_2=1$ and fill bound $c_{f,2}=0$.
    Then the statement follows from composition in Lemma \ref{lem:level-map-composition}.
    Also note that $\bar{\Theta}$ preserves sparsity by repeated application of Lemma \ref{lem:sparsity}.
\end{proof}

\begin{remark}[$X$ sector]
    Though we restrict our attention to RG (and single-shot) properties of the $Z$ sector, it's worth mentioning that the $X$ sector corresponds to the level map
    \begin{equation}
        C^{\top} \to \bar{\Theta}(C)^{\top} = (\bar{\Gamma} \circ \bar{\Gamma}_*)(C^{\top})
    \end{equation}
    In particular, $\bar{\Theta}_*\equiv \bar{\Gamma}\circ \bar{\Gamma}_*$ has cleaning with the same contraction $\kappa$ and filling $c_f<3$ as $\bar{\Theta}$.
\end{remark}

\subsection{Repair}
Note that cleaning $f$ deals only with clean syndromes, and thus we must have a repair rule $Z\mapsto \ker\bar{\delta}$ to deal with general noisy syndromes $\noisy{\sigma} \in Z$.
The proof-oriented repair rule is conventionally taken to be the MW correction (e.g., the 3D or 4D toric codes).
However, since the level map induces a natural local repair (postponed to Section \ref{sec:local-repair}), we generalize the sufficient properties of a syndrome repair using meta-checks to accommodate both scenarios.

\begin{definition}[Repair]
    \label{def:repair}
    A level map $\bar{\Theta}$ is equipped with a ($Z$-sector) \textbf{syndrome repair rule} $\nu:Z(\Theta)\to Z(\Theta)$ with \textbf{energy (slack)} $0\le \ve <1$ if it satisfies the following.
    \begin{itemize}
        \item (Valid) $\nu$ is invariant under valid syndromes, maps $0\mapsto 0$ and $\bar{\delta} \nu(\noisy{\sigma}) = \bar{\delta}\noisy{\sigma}$
        \item (Near half fault) If $\fs$ is any connected component (relative to meta-checks) of $\rep \noisy{\sigma}\equiv \noisy{\sigma} + \nu (\noisy{\sigma})$, then
        \begin{equation}
            |\fs| \le \frac{2}{1-\ve} |\noisy{\sigma}\cap \fs|
        \end{equation}
    \end{itemize}
\end{definition}

\begin{example}[MW]
    By Lemma \ref{lem:half-fault}, the (global) min-weight correction is a syndrome repair rule with $\ve =0$.
\end{example}

\begin{remark}[$m$ Sheets]
    \label{rem:m-sheets}
    Note that the main reason we generalized the $m=2$ sheet construction in \cite{balasubramanian2026passive} with contraction $=1/2$ to $m=3$ sheets is so that the latter has contraction $\kappa <1/2$, where the constant 2 derives from the near half fault property of the syndrome repair rule.
    However, by Lemma \ref{lem:level-map-composition}, we could've simply used the 2-sheet construction $\bar{\Theta}$, but treat the composition $\bar{\Theta}^2$ as a single level map so that the composed contraction is $\kappa=1/4<1/2$.
    Both methods are valid, but the former answers the natural question whether a single level can have contraction $<1/2$, and also provides better constants.
\end{remark}

%% file: RG.tex
\section{Renormalization Group}
\label{sec:RG}

In this section, we set up the RG decoder in Definition \ref{def:RG-decoder} which will be shown to be sustainably single-shot (SS), as stated in Theorem \ref{thm:RG-threshold}.
As discussed previously, we shall consider the RG decoder which repairs noisy syndromes in a manner that satisfies Definition \ref{def:repair}. 
The reader may wish to have the global min-weight projection as a concrete example in mind, though it's worth mentioning that a local repair is possible, and postponed to Section \ref{sec:local-repair}.
Specifically, we shall assume that 
We shall also assume that the seed code $C\super{0}$ is locally trivial as follows

\begin{remark}[Setup]
    \label{rem:setup}
    The tower $\bar{\bm{C}}$ is constructed with a sparsity preserving level map $\bar{\Theta}$ and QLDPC seed code $C\super{0}$ so that the tower $\bar{\bm{C}}$ is sparse.
    It is also equipped with cleaning $f$ and repair rule $\nu$ as in Definition \ref{def:cleaning} and \ref{def:repair}.
    The corresponding constants are collected as
    \begin{equation}
        \label{eq:constants}
        \Delta_{\sh},\fw,\fd,\kappa,c_f,\ve
    \end{equation}
    where $\Delta_{\sh}$ is the uniform bound on layer size, $\fw,\fd$ denote the max check weight and qubit degrees of sparse $\bar{\bm{C}}$, respectively.
    We write $\bar{\fw},\bar{\fd}$ denotes those of the meta-checks to distinguish from $\fw,\fd$, since they will only be need in Section \ref{sec:local-repair}. Also define the \textbf{noisy contraction} constant as
    \begin{equation}
        \noisy{\kappa} = \frac{2\kappa}{1-\ve}
    \end{equation}
    We further assume that the seed code $C\super{0}$ is \textbf{locally trivial}, i.e., every physical error $h$ in the seed code $C\super{0}$ supported within one seed $X$-check or one qubit such that $\delta h=0$ is a stabilizer $h\in \im \delta$.
    This is a property of the seed code, not of a decoder. For example, it holds if the distance in the error sector exceeds the largest seed $X$-check weight.
\end{remark}

\subsection{Decoder}

\begin{definition}[Lifted Chains]
    Let $\bm{e}=(e\super{i})$ be a sequence of $i$-level physical chains. Then define the \textbf{lifted} physical as
    \begin{equation}
        \Sigma \bm{e} = \iota \left(\sum_i e_i\right)
    \end{equation}
    where we slightly abuse notation, and also use $\iota$ to denote the iterated inclusion from $C\super{i}\to C\super{\ell}$. Similarly define \textbf{lifted} syndromes.
\end{definition}

\begin{definition}[RG Decoder]
    \label{def:RG-decoder}
    Let $\noisy{\sigma}\super{i} = \delta e\super{i} +\mu\super{i}$ be a noisy syndrome with $i$-level physical $e\super{i}$ and syndrome $\mu\super{i}$.
    Then the $i$-\textbf{level pass} for $i\ge 1$ performs
    \begin{itemize}
        \item (R) Repair the syndrome so that $\rep \noisy{\sigma}\super{i}\equiv \noisy{\sigma}\super{i} +\nu\super{i}$ is a clean syndrome, where $\nu\super{i}=\nu(\noisy{\sigma}\super{i})=\nu(\mu\super{i})$
        \item (C) Clean $\rep \noisy{\sigma}\super{i} \mapsto \noisy{\sigma} \super{i-1} = \desc \rep \noisy{\sigma}\super{i}$ and record cleaning $f\super{i}=f(\rep \noisy{\sigma}\super{i})$
    \end{itemize}
    If $i=0$, then  we apply a fixed arbitrary noisy seed decoder to obtain correction $\nu\super{0},f\super{0}$; to emphasize the distinction between the seed level $i=0$ and $i\ge 1$, we may also write $\cM_0,\cE_0$ instead.
    Note the seed decoder may fail to run; we state successful-run conditions where needed\footnote{Recall by Definition \ref{def:decoder}, the difference between a successful run and successful decoding}.
    % In particular, the seed syndrome correction $\nu\super{0}$ is always defined, while $f\super{0}$ is defined only when the seed decoder runs successfully on the descended syndrome $\noisy{\sigma}\super{0}$.
    In particular, given noisy syndrome $\noisy{\sigma}\super{\ell}$, apply the level pass in descending order $\ell,\ell-1,...$ to obtain the \textbf{syndrome history} $\noisy{\bm{\sigma}}=(\noisy{\sigma}\super{i})$. Similarly, define the \textbf{measurement} $\bm{\nu}$, \textbf{physical} $\bm{f}$, and \textbf{repair history} $\rep \noisy{\bm{\sigma}}$.
    Further define
    \begin{equation}
        \cE(\noisy{\sigma}\super{\ell}) = \Sigma \bm{f}, \quad \cM(\noisy{\sigma}\super{\ell}) = \Sigma \bm{\nu}
    \end{equation}
    We may also write $\hat{\cE}$ to only sum over $f\super{i},i\ne 0$, i.e., all but the seed level.
    The map $\noisy{\sigma}\super{\ell}\mapsto \cE,\cM$ is the \textbf{RG decoder}, since $\cE,\cM$ will be the final \textit{approximations} for the physical error $e\super{\ell}$ and measurement error $\mu\super{\ell}$.
\end{definition}

\begin{remark}[Valid Syndromes]
    Suppose that at some level $\noisy{\sigma}\super{i}$ is a valid syndrome, then by the descent, all lower levels $\noisy{\sigma}\super{j},j<i$ are valid syndromes.
    In particular, the RG decoder is a valid meta-check decoder as in Definition \ref{def:meta-check-decoder} so that
    \begin{equation}
        \delta \cE(\sigma\super{\ell}) = \sigma\super{\ell}, \quad \sigma\super{\ell} \in \im \delta
    \end{equation}
\end{remark}

% In a conventional noisy syndrome $\noisy{\sigma}=\delta e+\mu$ decoder with meta-checks (e.g., 4D toric code), the syndrome correction $\nu$ is independent of the physical error $e$ since it only depends on the meta-syndrome $\delta \noisy{\sigma}$, while the physical correction $\hat{e}$ depends on $\noisy{\sigma}$ only through the repaired syndrome $\noisy{\sigma} +\nu$.
% In the following, we show that the RG decoder exhibits the same property.
\begin{lemma}[meta-check decoder]
    \label{lem:RG-noisy-decoder}
    Let $\noisy{\sigma}\super{\ell} = \delta e\super{\ell} +\mu\super{\ell}$ be an $\ell$-level noisy syndrome. Then the RG decoder applied to $\noisy{\sigma}\super{\ell}$ outputs 
    \begin{itemize}
        \item the same measurement history $\bm{\nu}$ as that applied to just $\mu\super{\ell}$; in particular,
        \begin{equation}
            \cM(\noisy{\sigma}\super{\ell}) = \cM(\mu\super{\ell}) 
        \end{equation}
        \item if the decoder runs successfully, the same physical history $\bm{f}$ as that applied to the repaired $\noisy{\sigma}\super{\ell} +\cM(\mu\super{\ell})$; in particular,
        \begin{equation}
            \cE(\noisy{\sigma} \super{\ell}) = \cE(\noisy{\sigma}\super{\ell} + \cM(\mu\super{\ell}))
        \end{equation}
    \end{itemize}
    In particular, the RG decoder is a meta-check decoder for $C\super{\ell}$ as in Definition \ref{def:meta-check-decoder}.
\end{lemma}

\begin{proof}
    Adopt the notation convention of Definition \ref{def:RG-decoder} with input $\noisy{\sigma}\super{\ell}$.
    Let us proceed by induction, and assume that the coset $[\noisy{\sigma}\super{i}]\equiv \noisy{\sigma}\super{i}/ \im \delta$ is determined by coset $[\noisy{\sigma}\super{\ell}]$.
    Indeed, since the syndrome descent is linear modulo $\im \delta$, we have
    \begin{align}
        \noisy{\sigma}\super{i-1} &= \desc(\noisy{\sigma}\super{i} + \nu\super{i}) \\
        &= \desc (\mu\super{i} + \nu\super{i}) +\desc \delta e\super{i}  + \im \delta \\
        &= \desc \rep \mu\super{i} + \im \delta
    \end{align}
    where $\mu\super{i}$ is any representation of the coset $[\noisy{\sigma}\super{i}]$ so that $\noisy{\sigma}\super{i}=\mu\super{i} +\delta e\super{i}$, and we used the fact $\nu\super{i}$ depends on $\noisy{\sigma}\super{i}$ only implicitly through the coset.
    Hence, the coset $[\noisy{\sigma}\super{i-1}]$ is determined by the coset $[\noisy{\sigma}\super{i}]$ and thus by induction follows.
    Now consider the physical corrections. 
    In the case where the decoder runs successfully, we see that
    \begin{equation}
        \noisy{\sigma}\super{\ell} + \cM(\mu\super{\ell}) = \delta \cE(\noisy{\sigma}\super{\ell})
    \end{equation}
    and thus the left side is valid.
    Note that at the $\ell$ level, we have
    \begin{equation}
        \noisy{\sigma}\super{\ell} + \cM(\mu\super{\ell}) = \noisy{\sigma}\super{\ell} +\nu\super{\ell} + \iota\underbrace{\left(\cdots \right)}_{i<\ell}
    \end{equation}
    where the underbraced term denote syndromes in $i<\ell$ levels.
    By the passthrough property of the level map,
    \begin{equation}
        \desc (\noisy{\sigma}\super{\ell} + \cM(\mu\super{\ell})) = \desc \rep \noisy{\sigma}\super{\ell} + \nu\super{\ell-1} + \iota\underbrace{\left(\cdots \right)}_{i<\ell-1}
    \end{equation}
    Hence, the statement follows by induction.
    % Finally, by Theorem \ref{thm:contraction}, we see that
    % \begin{equation}
    %     \noisy{\sigma}\super{\ell} = \cM(\noisy{\sigma}\super{\ell}) + \delta \cE(\noisy{\sigma})
    % \end{equation}
    % Hence, the statement follows.
\end{proof}

\begin{lemma}[Run Failure]
    \label{lem:RG-run-failure}
    The RG decoder runs successfully on $\noisy{\sigma}\super{\ell}$ iff the seed decoder runs successfully on $\noisy{\sigma}\super{0}$, where $\noisy{\sigma}\super{0}$ is that obtained from $\noisy{\sigma}\super{\ell}$ in Definition \ref{def:RG-decoder}.
    Equivalently, the totally repaired input is valid iff the seed's repaired syndrome is valid.
    % In particular, under Hypothesis \ref{hyp:seed-decoder}, the RG decoder always runs successfully.
\end{lemma}
\begin{proof}
    Note that each $i$-level pass always runs successfully for $i\ge 1$ so that
    \begin{equation}
        \noisy{\sigma}\super{\ell} + \iota\left(\sum_{i=1}^{\ell} \nu\super{i}\right) +\delta \iota\left(\sum_{i=1}^{\ell} f\super{i}\right) = \iota\noisy{\sigma}\super{0}
    \end{equation}
    The statement then follows by definition and that $\pi\iota$ is the identity.
    % If the seed decoder runs successfully, then $\noisy{\sigma}\super{0} =\nu\super{0} + \delta f\super{0}$ and thus
    % \begin{equation}
    %     \noisy{\sigma}\super{\ell}  = \cM(\mu\super{\ell})+\delta \cE(\noisy{\sigma}\super{\ell}) 
    % \end{equation}
    % Hence, the RG decoder runs successfully.
    % Conversely, if the RG decoder runs successfully, then $\noisy{\sigma}\super{\ell}  = \cM(\mu\super{\ell})+\delta \cE(\noisy{\sigma}\super{\ell})$ and thus
    % \begin{equation}
    %     \iota (\noisy{\sigma}\super{0} + \nu\super{0} + \delta f\super{0}) =0
    % \end{equation}
    % Since $\pi\iota$ is the identity, we see that the seed decoder runs successfully.
\end{proof}

\begin{remark}[Run Failure]
    \label{rem:RG-run-failure}
    By Remark \ref{rem:run-failure}, Lemma \ref{lem:RG-noisy-decoder} and \ref{lem:RG-run-failure}, we see that when investigating the run failure on $\noisy{\sigma}\super{\ell}$, it's equivalent to investigate the measurement error $\mu\super{\ell}$ and its seed level descent $\mu\super{0}$.
    A particularly interesting case is when $\mu\super{0}=0$, since it would imply a successful run.
\end{remark}

\subsection{Reduction}

In this subsection, we show that the noise model at the top level $C\super{\ell}$ is reduced to an \textit{effective} noise model at the seed level $C\super{0}$ by the RG decoder.
\begin{definition}[Noisy Descent]
    Let $\mu\super{\ell}$ be a noisy syndrome with history $\bm{\mu}$.
    Define the \textbf{noisy (physical) descent} of $e\super{i}$ relative to $\mu\super{\ell}$ as
    \begin{equation}
        \desc_{\bm{\mu}} e\super{i} = \pi (e\super{i} + f(\delta e\super{i} + \rep \mu\super{i})+ f(\rep \mu\super{i}))
    \end{equation}
    Note that $\desc_{\bm{\mu}} 0 =0$, and if $\mu\super{\ell}=0$, then $\desc_{\bm{\mu}} = \desc$ is the physical descent.
    If we started with top level physical $e\super{\ell}$, then the sequence $\bm{\desc}_{\bm{\mu}} e\super{\ell}$  defined inductively is the \textbf{noisy descent (sequence)} of $e\super{\ell}$ relative to $\mu\super{\ell}$.
\end{definition}

\begin{remark}[Passthrough]
    \label{rem:noisy-descent-passthrough}
    Note that $\desc_{\bm{\mu}}$ satisfies
    \begin{equation}
        \desc_{\bm{\mu}} (e\super{i} + \iota e\super{i-1}) = \desc_{\bm{\mu}} e\super{i} + e\super{i-1}
    \end{equation}
\end{remark}

\begin{remark}[Stabilizer Equivalent]
    \label{rem:noisy-descent-stabilizer-equivalent}
    Note that $\desc_{\bm{\mu}}$ satisfies
    \begin{equation}
        \iota \desc_{\bm{\mu}} e\super{i}  + e\super{i} + f(\delta e\super{i} + \rep \mu\super{i}) + f(\rep\mu\super{i}) \in \im \delta
    \end{equation}
\end{remark}
\begin{lemma}[Perturbation]
    \label{lem:perturbation}
    Let $\noisy{\sigma}\super{\ell} = \delta e\super{\ell} + \mu\super{\ell}$ be a noisy syndrome with noisy history $\noisy{\bm{\sigma}}$, and let $\bm{\mu}$ denote the noisy history of $\mu\super{\ell}$. 
    Let $\bm{e}_{\bm{\mu}}$ denote the noisy descent sequence of $e\super{\ell}$ relative to $\mu\super{\ell}$. Then
    \begin{equation}
        \noisy{\sigma}\super{i} = \delta e_{\bm{\mu}}\super{i} + \mu\super{i}, \quad i\in [0,\ell]
    \end{equation}
\end{lemma}

\begin{proof}
    It's clear that the equality is true for $i=\ell$ and thus we proceed by induction.
    Note that
    \begin{align}
        \delta e_{\bm{\mu}}\super{i-1} &= \pi(\delta e_{\bm{\mu}}\super{i} + \delta f( \rep \noisy{\sigma}\super{i}) + \delta f(\rep \mu\super{i})) \\
        &= \pi (\delta e_{\bm{\mu}}\super{i} + \rep \noisy{\sigma}\super{i} + \iota \noisy{\sigma}\super{i-1} + \rep \mu\super{i} + \iota \mu\super{i-1})\\
        &= \noisy{\sigma}\super{i-1} +\mu\super{i-1}
    \end{align}
    where we used the fact that $\rep \noisy{\sigma}\super{i} = \delta e_{\bm{\mu}}\super{i} + \rep\mu\super{i}$, and $\pi \iota =\mathrm{id}$.
\end{proof}

\begin{lemma}[Seed Level Reduction]
    \label{lem:seed-lvl-reduction}
    Let $\bm{e}_\pm$ be the noisy descent sequence of $e_\pm\super{\ell}$ relative to $\mu\super{\ell}$, and let $\noisy{\bm{\sigma}}$ be the history of $\noisy{\sigma}_\pm\super{\ell} =\delta e_\pm\super{\ell} + \mu\super{\ell}$.
    Then, modulo stabilizers,
    \begin{equation}
        e_+\super{\ell} + e_-\super{\ell} + \hat{\cE}(\noisy{\sigma}_+\super{\ell}) + \hat{\cE}(\noisy{\sigma}_-\super{\ell}) = \iota (e_+\super{0} +e_-\super{0}) + \im \delta
    \end{equation}
\end{lemma}

% \begin{remark}
%     By local triviality in Remark \ref{rem:setup}, if $\noisy{\sigma}\super{0}_\pm=0$, then
%     \begin{equation}
%         e_+\super{\ell} + e_-\super{\ell} + \cE(\sigma_+\super{\ell}) + \cE(\sigma_-\super{\ell}) \in \im \delta
%     \end{equation}
% \end{remark}
\begin{proof}
    The statement follows repeated application of Remark \ref{rem:noisy-descent-stabilizer-equivalent} and Lemma \ref{lem:perturbation}.
\end{proof}

\begin{theorem}[Noise Model Reduction]
    \label{thm:noise-model-reduction}
    Consider the noise model in Definition \ref{def:noise-model} with physical $e\super{\ell}(t)$ and measurement errors $\mu\super{\ell}(t)$ at the top level $C\super{\ell}$.
    Then the residual after $t$ rounds, modulo stabilizers, satisfies
    \begin{equation}
        r\super{\ell}(t) = \hat{\cE}(\mu\super{\ell}(t^-)) + \iota s\super{0}(t) + \im \delta
    \end{equation}
    where $s\super{0}(t)$ is the seed level \textbf{residue} defined with $s\super{0}(0)=0$ and inductively as
    \begin{equation}
        s\super{0}(t+1) = s\super{0}(t) + E\super{0}(t) + \cE_0( \delta (s\super{0}(t) + E\super{0}(t)) + \mu\super{0}(t^+))
    \end{equation}
    where $E\super{0}(t)$ is the seed level noisy descent of \textbf{effective} error 
    \begin{equation}
        E\super{\ell}(t) \equiv e\super{\ell}(t) + \hat{\cE}(\mu\super{\ell}(t^-))
    \end{equation}
    relative to $\mu\super{\ell}(t^+)$; note that it only depends on adjacent time steps $t$ and $t^\pm$.
    $\mu\super{0}(t^+)$ is the seed level output of the history of $\mu\super{\ell}(t^+)$.
    We refer to $E\super{0}(t),\mu\super{0}(t^+)$ as the \textbf{effective} seed level physical and measurement errors.
\end{theorem}

\begin{proof}
    Note that by induction, we have
    \begin{equation}
        r(t+1) = r(t) + e(t) + \cE(\delta (r(t) +e(t)) + \mu(t^+))
    \end{equation}
    Apply Lemma \ref{lem:seed-lvl-reduction} by replacing 
    \begin{align}
        e_+\super{\ell} &\mapsto e(t) +\hat{\cE}(\mu(t^-)) + \iota s(t) \\
        e_-\super{\ell} &\mapsto 0\\
        \mu\super{\ell} &\mapsto \mu\super{\ell}(t^+)
    \end{align}
    where we note that the physical is stabilizer equivalent to the residue $r(t)$ by induction.
    By the passthrough property in Remark \ref{rem:noisy-descent-passthrough}, the statement then follows.
\end{proof}

\begin{remark}
    By Theorem \ref{thm:noise-model-reduction} and boundary condition $\mu\super{\ell}(T^+)=0$, we see that if the seed level residue $s\super{0}(T+1) \in \im \delta$, then so is the top level residue $r\super{\ell}(T+1)$.
    Hence, the problem has been reduced to considering the effective seed level errors $E\super{0}(t),\mu\super{0}(t^+)$; both of which are \textbf{memory-less}, in the sense that they only depend on at most adjacent time steps $t,t^\pm$.
    In particular, we will show that every $E\super{0}(t)$ is equal to some $\ubar{E}\super{0}(t)$ modulo seed level stabilizers, such that $\ubar{E}\super{0}(t),\mu\super{0}(t^+),t\in [0,T]$ are jointly local stochastic with rate $\ll p$, as in Theorem \ref{thm:noise-transfer}.
\end{remark}
\subsection{Stability}

\begin{definition}
    Two $i$-level chains \textbf{agree off} cell $w\super{i-1}$ in $C\super{i-1}$ if their difference has shadow in the \textbf{up set} $\{\succeq w\super{i-1}\}$.
\end{definition}
\begin{lemma}[Local Stability]
    \label{lem:RG-stability}
    Let $\bm{\mu}$ be the history of $\mu\super{\ell}$.
    If $e_\pm\super{i} \in Q\super{i}$ agree off a cell $w\super{i-1} \in \cC\super{i-1}$, then $\desc_{\bm{\mu}} e_\pm\super{i}$ differs only in $\{\succeq w\super{i-1}\}$. 
    In particular, if $e_\pm\super{\ell}$ differ by $\iota q$ where $q\in \cQ\super{j}$ is a single qubit in some $j\le \ell$ level and $\bm{e}_\pm$ denote the noisy descents sequence relative to $\bm{\mu}$, then their seed level noisy descents differ only in 
    \begin{equation}
        e_+\super{0} + e_-\super{0} \subseteq \{\succeq w_q\}
    \end{equation}
    where $w_q=\sh \circ \cdots \circ \sh q$ is the seed level (compositive) shadow of $q$.
\end{lemma}
\begin{proof}
    Note that $\delta,\pi,f$ are layer-local and thus so is the noisy descent, i.e.,
    \begin{equation}
        (\desc_{\bm{\mu}} e_\pm)|_{u} = (\desc_{\bm{\mu}} e_{\pm}|^{\preceq u} )|_u
    \end{equation}
    Hence, we see that if $u\not\succeq w$, then $v\not\succeq w$ for any $v\preceq u$. Hence, $(\desc_{\bm{\mu}} e_\pm)|_u$ can only differ in $\{\succeq w\}$.
    Now consider the case where $e_+\super{\ell} + e_-\super{\ell}=\iota q\super{j}$.
    By the passthrough property of the noisy descent, we see that
    \begin{equation}
        e_+\super{i} = e_-\super{i} + \iota q\super{j}, \quad i\ge j
    \end{equation}
    In particular, this implies that
    \begin{equation}
        f(e_+\super{i} +\mu\super{i}) = f(\delta e_-\super{i} +\mu\super{i}), \quad i > j
    \end{equation}
    We also see that $e_\pm\super{j}$ only differ by $q\super{j}$ and thus must agree off some cell $w\super{j-1}$.
    We can apply induction, so that $e_\pm\super{i}$ agree off $w\super{i-1}$ where $w\super{i-1}=\sh w\super{i}$.
    Hence, $e_{\pm}\super{0}$ are equal everywhere except possibly in $\{\succeq w\super{0}\}$.
\end{proof}

\begin{example}[Intuition]
    \label{ex:stability-intuition}
    Let us provide a simple example to describe how we will utilize Lemma \ref{lem:RG-stability}.
    Consider the first pass $t=0$ in Theorem \ref{thm:noise-model-reduction} so that the seed level effective error $e\super{0}$ is the noisy descent of $e\super{\ell}$ relative to measurement error $\mu\super{\ell}$, where we have omitted the $t=0$ dependence.
    If $e\super{0}$ is a stabilizer, then we are done and thus we want to consider the opposite.
    
    Consider an increasing sequence of subsets $e_{\bullet}\super{\ell}\subseteq e\super{\ell}$ from $0\to e\super{\ell}$ by adding one qubit $q\in e\super{\ell}$ at a time, with noisy descent $e_{\bullet}\super{0}$ relative to $\mu\super{\ell}$.
    By Lemma \ref{lem:RG-stability}, the difference $e_{\bullet}\super{0} + e_{\bullet-1}\super{0}$ is \textit{local}.
    By Lemma \ref{lem:perturbation},
    \begin{equation}
        \delta (e_{\bullet}\super{0} + e_{\bullet-1}\super{0})= \noisy{\sigma}\super{0}_{\bullet} + \noisy{\sigma}\super{0}_{\bullet-1}
    \end{equation}
    If $e\super{0}$ is not a stabilizer, by Remark \ref{rem:setup}, there must be some intermediate point with $\noisy{\sigma}\super{0}_{\bullet} \ne 0$, and thus the goal is to show that even if $\noisy{\sigma}_{\bullet}\super{\ell}\ne 0$, its seed level descendent $\noisy{\sigma}\super{0}_{\bullet}$ has a high probability of being zero.
    This mirrors the $\mu\super{0}=0$ case for run failure in Remark \ref{rem:RG-run-failure}.
\end{example}

%% file: run-failure.tex
\section{Run Failure}
\label{sec:run-failure}
By Remark \ref{rem:RG-run-failure} (and tangentially, Theorem \ref{thm:noise-model-reduction}), the run failure of the RG decoder is governed by that of measurement error $\mu\super{\ell}$, without \textit{interference} from the physical errors and residue.
Hence, this section is devoted to the study of the history of $\mu\super{\ell}$ as a graphical representation, so that the probability of a seed level history $\mu\super{0}$ will be well-controlled by a connected component.
For exposition purposes, this section will mainly aim to show that $\mu\super{0}=0$ with high probability; its local stochastic behavior will be postponed to Theorem \ref{thm:noise-transfer}.

\subsection{Syndrome Graph}
The intuition is to construct a bounded degree graph $\cG$ that corresponds to the space of syndromes so that any syndrome $\noisy{\sigma}\super{\ell}$ (with history $\noisy{\bm{\sigma}}$) corresponds to a subgraph $\cG(\noisy{\bm{\sigma}})\subseteq \cG$. The connectivity (edges) of the graph $\cG$ is chosen so that restriction to each connected component of the subgraph commutes with cleaning and descent.
One can then provide bounds on the number of connected components to prove error correction properties.
In our case, the noisy history $\noisy{\bm{\sigma}}$ has a natural hierarchy structure, and thus it's natural to consider a graph $\cG$ with a corresponding hierarchy structure that emphasizes the layer-locality of the embedding. 
Specifically, 

\begin{definition}
    Let the \textbf{RG syndrome graph} $\cG$ be that with vertices $\cV = \bigsqcup_{i\ge 1} \cV\super{i}$ and edges defined as follows.
    \begin{itemize}
        \item \textbf{Vertices}. $\cV\super{i},i\in [\ell]$ are the cells $\cC\super{i-1}$ in level $i-1$ (meta-checks are not included)
        % \item \textbf{\textit{Horizontal} Nearest-Neighbor (NN) Edges}. $uv$ for $u,v\in \cV\super{i}$ if $u\preceq v$
        \item \textbf{\textit{Horizontal} Edges}. $uv$ for $u,v\in \cV\super{i}$ if there exists $w \in \cV\super{i}$ such that $u,v\preceq w$
        \item \textbf{\textit{Vertical} Edges}. $uv$ where $u\in \cV\super{i-1}$ and $v\in \cV\super{i}$ if $u=\sh v$
    \end{itemize}
    Note that the degree due to vertical edges is bounded by the layer sizes $\Delta_{\sh}$.
    Let $\Delta_{\cG}$ denote the max degree of $\cG$, which depends only on $\Delta_{\sh},\fw,\fq$.
\end{definition}

% Given a history, write $\cK\super{i}=\cK\cap\cV\super{i}$ for a component's vertices at level $i$. For $i\ge1$, $\dnum^{\cK}$ restricts chains to layers with shadows in $\cK\super{i}$; $\dnum_{\cK}$ restricts chains in $C\super{i-1}$ to cells in $\cK\super{i}$. At the seed, use cell restriction $\dnum_{\cK\super{1}}$. These define a restricted history, which need not be the output of rerunning the decoder on the restricted input. In particular, we do not assume that repair commutes with restriction.

\begin{remark}[Meta-Check vs RG Connectivity]
    \label{rem:meta-check-vs-RG-connectivity}
    By layer locality, we see that if $z,z'$ are in a common meta-check, then their shadows are equal or adjacent in the RG graph.
    Hence, the connectivity of the RG graph will inherit Lemma \ref{lem:half-fault} as used in Lemma \ref{lem:noisy-contraction}.
\end{remark}
% \begin{remark}[Next Nearest-Neighbors]
%     In the definition of the horizontal edges, one may suspect that it's more natural to only define \textit{nearest-neighbors}, i.e., $uv$ is an edge if $u\preceq v$. However, as we shall see in the proof of Theorem \ref{lem:correction-invariance}, or specifically, Eq. \eqref{eq:nnn-used}, this is insufficient. Instead, \textit{next nearest-neighbors} are necessary.
%     To emphasize this distinction, let $\Delta_{\rm{nn}}$ denote the max degree due to nearest neighbors and $\Delta_{\rm{nnn}}\le \Delta_{\rm{nn}}^2$ be that due to next nearest neighbors.
% \end{remark}

% Next, we want to define a mapping $\noisy{\bm{\sigma}}\mapsto \cG(\noisy{\bm{\sigma}})$ from a history to a collection of vertices (and thus a subgraph) of the RG syndrome graph, such that if $\cK$ is a connected component of $\cG(\noisy{\bm{\sigma}})$ where $\cK\super{i}=\cK\cap \cV\super{i}$ denotes the intersection, then the sequence 
% \begin{equation}
%     \noisy{\bm \sigma}|^\cK \equiv (\noisy{\sigma}\super{i}|^\cK), \quad \noisy{\sigma}\super{i}|^\cK \equiv \noisy{\sigma}\super{i}|^{\cK\super{i}}
% \end{equation}
% is the output of the RG decoder with input $\noisy{\sigma}\super{\ell}|^\cK$.
% In other words, each connected component corresponds to a \textit{history}.
% As will become apparent below, we define the mapping as follows

% Given a history, 

\begin{definition}[RG Subgraph]
    \label{def:RG-subgraph}
    Given a history $\noisy{\bm{\sigma}}$ (or noisy syndrome $\noisy{\sigma}\super{\ell}$), let the \textbf{RG syndrome subgraph} $\cG(\noisy{\bm{\sigma}})$ denote the collection of vertices in $v\in \cG$ (and thus corresponds to a subgraph) such that either $\noisy{\sigma}\super{i}|^v, (\rep \noisy{\sigma}\super{i})|^v, (\delta f\super{i})|^v$ or $f\super{i}|^{v}$ is nonzero, i.e.,
    \begin{equation}
        \cG(\noisy{\bm{\sigma}}) = \bigcup_i \sh(\noisy{\sigma}\super{i}\cup \rep \noisy{\sigma}\super{i} \cup \delta f\super{i} \cup f\super{i})
    \end{equation}
    If $\cK$ denotes a connected component of $\cG(\noisy{\bm{\sigma}})$, write $\cK\super{i}=\cK\cap\cV\super{i}$ for a component's vertices at level $i$. For $i\ge1$, $\dnum^{\cK}$ restricts chains to layers with shadows in $\cK\super{i}$; $\dnum_{\cK}$ restricts chains in $C\super{i-1}$ to cells in $\cK\super{i}$. 
    These define a restricted history, which need not be the output of rerunning the decoder on the restricted input. In particular, we do not assume that repair commutes with restriction.
\end{definition}

\begin{lemma}[Cell-Layer-Invariance]
    \label{lem:cell-layer-invariance}
    Let $\noisy{\bm{\sigma}}$ be a noisy history with RG syndrome graph $\cG(\noisy{\bm{\sigma}})$ and a connected component $\cK$. Then for any $i\in [\ell)$,
    \begin{equation}
        \noisy{\sigma}\super{i}|_{\cK} = \noisy{\sigma}\super{i}|^{\cK}
    \end{equation}
    For consistency, we also write $\noisy{\sigma}\super{0}|^{\cK} \equiv \noisy{\sigma}\super{0}|_{\cK}$.
    % A similarly inequality also holds for $\sigma_{\rm rep}$.
\end{lemma}
\begin{proof}
    Note that we can regard the left and right side as a subsets of $z$ checks in $C\super{i}$.
    If $\noisy{\sigma}\super{i}|_z\ne0$ for $z\in \cK\super{i+1}$, then $\noisy{\sigma}\super{i}|^u\ne 0$ for $u=\sh z$.
    By definition, $u\in \cG(\noisy{\bm{\sigma}})$ but since $z\in \cK\super{i+1}$, we see that $u\in \cK\super{i}$. Hence, 
    \begin{equation}
        \noisy{\sigma}\super{i}|_{\cK} \subseteq \noisy{\sigma}\super{i}|^{\cK}
    \end{equation}
    Conversely, consider $z \in \noisy{\sigma}\super{i}|^\cK$ so that $\noisy{\sigma}\super{i}|_z\ne 0$ for some $\sh z \in \cK\super{i}$.
    By the descent, we see that either $\rep \noisy{\sigma}\super{i+1}|^z\ne 0$ or $(\delta f(\rep \noisy{\sigma}\super{i+1}))|^z \ne 0$.
    Hence, $z\in \cG(\noisy{\bm{\sigma}})$ and thus $z\in \cK\super{i+1}$.
    Hence, the other direction is also proven.
\end{proof}
\begin{lemma}[Correction Invariance]
    \label{lem:correction-invariance}
    Let $\noisy{\bm{\sigma}}$ be a history with syndrome subgraph $\cG(\noisy{\bm{\sigma}})$. If $\cK$ is a connected component of $\cG(\noisy{\bm \sigma})$ and $i\in [\ell]$, then $\dnum^\cK \rep \noisy{\sigma}\super{i}$ is clean and
    \begin{equation}
        \delta f(\dnum^\cK \rep \noisy{\sigma}\super{i}) = \dnum^\cK\delta f(\rep \noisy{\sigma}\super{i}), \quad f(\dnum^\cK \rep \noisy{\sigma}\super{i}) = \dnum^\cK f(\rep \noisy{\sigma}\super{i})
    \end{equation}
    In particular,
    \begin{equation}
        \dnum_{\cK} \desc \rep \noisy{\sigma}\super{i}  = \desc \dnum^{\cK} \rep \noisy{\sigma}\super{i}
    \end{equation}
\end{lemma}
\begin{proof}
    % Let's first prove the simpler second equality. Indeed, by layer-locality,
    % \begin{equation}
    %     (\delta f(\rep \noisy{\sigma}\super{i})) |^{\cK} = (\delta f(\rep \noisy{\sigma}\super{i}|^{\preceq \cK})) |^{\cK}
    % \end{equation}
    % Note that
    % \begin{equation}
    %     \rep \noisy{\sigma}\super{i} = \sum_{\cK'}\rep \noisy{\sigma}\super{i}|^{\cK'}
    % \end{equation}
    % where the summation is over all connected components $\cK'$ of $\cG(\noisy{\bm \sigma})$. Note that by definition, $\cK'\cap \{\preceq \cK\}$ is empty if $\cK'\ne \cK$. Hence, we have
    % \begin{equation}
    %     (\delta f(\rep \noisy{\sigma}\super{i})) |^{\cK} = (\delta f(\dnum^{\cK} \rep \noisy{\sigma}\super{i})) |^{\cK}
    % \end{equation}
    Note that by Remark \ref{rem:meta-check-vs-RG-connectivity}, a meta-check can only be adjacent to at most one $\dnum^{\cK} \rep \noisy{\sigma}\super{i}$ among connected components $\cK$ of $\cG(\noisy{\bm{\sigma}})$ and thus $\dnum^{\cK} \rep \noisy{\sigma}\super{i}$ is a clean syndrome.
    For the first equality, if $v\in \cK\super{i}$, then by layer-locality,
    \begin{equation}
        \delta f(\dnum^{\cK} \rep \noisy{\sigma}\super{i})|^v = \delta f(\dnum^{\cK,\preceq v} \rep \noisy{\sigma}\super{i})|^v 
    \end{equation}
    Note that if $u\preceq v\in \cK\super{i}$ is such that $\rep \noisy{\sigma}\super{i}|^u\ne 0$, then $u\in \cK\super{i}$ and thus
    \begin{align}
        \dnum^{\cK,\preceq v} \rep \noisy{\sigma}\super{i} &= \dnum^{\preceq v} \rep \noisy{\sigma}\super{i}\\
        \dnum^{v} \delta f(\dnum^{\cK} \rep \noisy{\sigma}\super{i})) &= \dnum^{v} \delta f(\rep \noisy{\sigma}\super{i})
    \end{align}
    Note that if $v\not\succeq \cK\super{i}$, then the first equality trivially holds. 
    Hence, consider the final scenario where $v\succeq \cK\super{i}$ but $\notin \cK\super{i}$.
    Then
    \begin{align}
        \rep \noisy{\sigma}\super{i} &= \sum_{\cK'} \dnum^{\cK'} \rep \noisy{\sigma}\super{i} \\
        \dnum^{\preceq v} \rep \noisy{\sigma}\super{i} &= \sum_{\cK'} \dnum^{\cK',\preceq v} \rep \noisy{\sigma}\super{i} \\
        \dnum^{\preceq v} \rep \noisy{\sigma}\super{i} &= \dnum^{\cK,\preceq v} \rep \noisy{\sigma}\super{i}
        \label{eq:nnn-used}
    \end{align}
    where the last equality is the only time we use the \textit{next nearest-neighbor} property of the horizontal edges.
    Hence,
    \begin{equation}
        \dnum^{v} \delta f(\dnum^{\cK} \rep \noisy{\sigma}\super{i}) =\dnum^v  \delta f(\rep \noisy{\sigma}\super{i}) =0
    \end{equation}
    The statement then follows.
    A similar argument proves the second equality.
\end{proof}

% \begin{remark}[Valid Syndrome]
%     If input $\noisy{\sigma}\super{\ell}$ is a valid syndrome with valid descent $\noisy{\sigma}\super{0}|_{\cK}$, then so is $\noisy{\sigma}\super{\ell}|^{\cK}$ and rerunning the decoder on this restricted input reproduces the restricted syndrome history, with syndrome subgraph $\cK$.
% \end{remark}

\begin{lemma}[Noisy Contraction]
    \label{lem:noisy-contraction}
    Let $\noisy{\bm{\sigma}}$ be a noisy history with RG syndrome subgraph $\cG(\noisy{\bm{\sigma}})$ and connected component $\cK$.
    If $i\in [\ell]$, then
    \begin{equation}
        |\noisy{\sigma}\super{i-1}|^{\cK} \le \kappa | \rep \noisy{\sigma}\super{i}|^{\cK} \le \noisy{\kappa} |\noisy{\sigma}\super{i}|^{\cK}, \quad \noisy{\kappa} = \frac{2\kappa}{1-\ve} <1
    \end{equation}
    where $\noisy{\kappa}$ is the noisy contraction in Remark \ref{rem:setup}.
\end{lemma}

\begin{proof}
    By Definition \ref{def:repair}, we see that if $\fs$ is a connected component (relative to meta-checks) of $\rep \noisy{\sigma}\super{i}$, then
    \begin{equation}
        |\fs| \le \frac{2}{1-\ve} |\noisy{\sigma}\super{i}\cap \fs|
    \end{equation}
    Also note that by Remark \ref{rem:meta-check-vs-RG-connectivity}, each $\fs$ must have shadow contained in a unique connected component of $\cG(\noisy{\bm{\sigma}})$,
    Hence, by definition
    \begin{align}
        \dnum^{\cK} \rep \noisy{\sigma}\super{i} &=\sum_{\sh \fs \subseteq \cK} \fs \\
        | \rep \noisy{\sigma}\super{i}|^{\cK} &\le \frac{2}{1-\ve}\sum_{\sh \fs \subseteq \cK} |\noisy{\sigma}\super{i}\cap \fs| \\
        &\le \frac{2}{1-\ve}|\noisy{\sigma}\super{i}|^{\cK}
    \end{align}
    The statement then follows from the descent.
\end{proof}

\subsection{Witness}

\begin{definition}[Witness]
    Let $\noisy{\bm{\sigma}}$ be a noisy history with RG subgraph $\cG(\noisy{\bm{\sigma}})$. 
    A component $\cK$ with $\noisy\sigma\super{0}|_{\cK}\ne0$ is a \textbf{witness}. 
    A witness exists iff the descended seed syndrome $\noisy{\sigma}\super{0}\ne 0$
    % We write $\noisy{\bm{\sigma}}_{\star}=\noisy{\bm{\sigma}}|^{\cK}$ to be that restricted to layers in the witness $\cK$.
\end{definition}

\begin{lemma}[Component Size]
    \label{lem:witness-size}
    If noisy history $\noisy{\bm{\sigma}}$ has connected component $\cK$, then 
    \begin{equation}
         |\cK| = \Omega(\noisy{\kappa}^{-\ell})  |\noisy{\sigma}\super{0}|^{\cK},\quad |\cK| = O(|\noisy{\sigma}\super{\ell}|^{\cK})
    \end{equation}
    where hidden constants in $\Omega,O$ depends only on $\noisy{\kappa},c_f,\ve$.
    Note that if $\noisy{\sigma}\super{\ell}$ is valid, then we can replace $\noisy{\kappa}\mapsto \kappa$.
\end{lemma}
\begin{proof}
    Let
    \begin{equation}
        \noisy{\tau}\super{i} = \noisy{\sigma}\super{i}|^{\cK},\quad  \rho\super{i} = (\rep \noisy{\sigma}\super{i})|^{\cK}
    \end{equation}
    Note that if $i\ge 1$, then
    \begin{align}
        \cK\super{i} &= \sh (\noisy{\tau}\super{i} \cup \rho\super{i} \cup \delta f(\rho\super{i})\cup f(\rho\super{i})) \\
        &= \sh \noisy{\tau}\super{i} \cup \sh ( \rho\super{i} \backslash \noisy{\tau}\super{i}) \cup \sh (\delta f(\rho\super{i}) \cup f(\rho\super{i})) \backslash \sh (\rho\super{i}) 
    \end{align}
    By definition, it's straightforward to check that
    \begin{equation}
        |\sh \noisy{\tau}\super{i}| \le |\noisy{\tau}\super{i}|, \quad  |\sh ( \rho\super{i} \backslash \noisy{\tau}\super{i})| \le |\rho\super{i} \backslash \noisy{\tau}\super{i}| \le \frac{1+\ve}{1-\ve} |\noisy{\tau}\super{i}|
    \end{equation}
    By the fill bound property of the level map, we have
    \begin{equation}
        |\cK\super{i}| \le \frac{2}{1-\ve} (1+c_f) |\noisy{\tau}\super{i}|
    \end{equation}
    Hence, by Lemma \ref{lem:noisy-contraction},
    \begin{equation}
        |\cK| \le c|\noisy{\bm{\tau}}|  \le \frac{1+c_f}{1-\noisy{\kappa}} \frac{2}{1-\ve} |\noisy{\tau}\super{\ell}|
    \end{equation}
    In the case where $\cK$ is the witness, by Lemma \ref{lem:cell-layer-invariance}, we see that
    \begin{equation}
        |\cK\super{i}| \ge |\noisy{\sigma}\super{i-1}|_{\cK} =|\noisy{\sigma}\super{i-1}|^{\cK}
    \end{equation}
    Hence, by Lemma \ref{lem:noisy-contraction},
    \begin{equation}
        |\cK| \ge \sum_{i=1}^{\ell} |\noisy{\tau}\super{i-1}| =\frac{\noisy{\kappa}^{-\ell} - 1}{\noisy{\kappa}^{-1}-1} |\noisy{\sigma}\super{0}|^{\cK}
    \end{equation}
    The case where $\noisy{\sigma}\super{\ell}$ is valid can be similarly derived.
\end{proof}

\begin{lemma}[Witness Entropy]
    \label{lem:witness-entropy}
    The number of connected clusters of the RG syndrome graph $\cG$ meeting $\cV\super{1}$ and of size $M$ is bounded by 
    \begin{equation}
        |\cC\super{0}| (e\Delta_{\cG})^M
    \end{equation}
\end{lemma}
\begin{proof}
    Follows from \cite{gottesman2013fault} by fixing any vertex $v\in \cV\super{1}$ and counting the number of connected clusters of size $M$ containing $v$.
\end{proof}

\begin{theorem}[Witness Probability]
    \label{thm:witness-probability}
    Let $\mu\super{\ell}$ be local stochastic with noisy history $\bm{\mu}$ and RG subgraph $\cG(\bm{\mu})$. 
    Then there exists $p_{\rm wit}>0$, depending only on constants in Eq. \eqref{eq:constants}, such that if $p/p_{\rm wit}$ is sufficiently small (depending only on $\noisy{\kappa},\ve,c_f$), then 
    \begin{equation}
        \Pr[\cG(\bm{\mu}) \text{ has witness}] \le 2|\cC\super{0}| \left(\frac{p}{p_{\rm wit}}\right)^{c\noisy{\kappa}^{-\ell}}
    \end{equation}
    where $c>0$ depends only $\noisy{\kappa},\ve,c_f$.
\end{theorem}
\begin{remark}[Successful Run]
    Note that by Remark \ref{rem:RG-run-failure}, the RG decoder fails to run with probability bounded by Theorem \ref{thm:witness-probability}, and thus succeeds with high probability. 
    Hence, we can focus on the logical failure.
\end{remark}
\begin{proof}
    For simplicity, we shall write $a,b$ as constants depending only on $\noisy{\kappa},\ve,c_f$.
    Note that
    \begin{equation}
        \Pr[\cG(\bm{\mu}) \text{ has witness}] \le \sum_{\cK \text{ cluster meets } \cV\super{1}} \Pr[\cG(\bm{\mu}) \text{ has witness } \cK]
    \end{equation}
    Hence, by Lemma \ref{lem:witness-entropy}, it's sufficient to bound each term in the summation, while by Lemma \ref{lem:witness-size}, we only need to consider clusters $\cK$ with size $|\cK|\ge b \noisy{\kappa}^{-\ell}$
    By the upper bound in Lemma \ref{lem:witness-size}, we see that if $\cG(\bm{\mu})$ has witness $|\cK|=K$, then
    \begin{equation}
        |\mu\super{\ell}|^{\cK} \ge a K
    \end{equation}
    Hence, we can sum over all syndromes $\tau\super{\ell}$ with shadow contained in $\cK\super{\ell}$ and consider the event where $\tau\super{\ell} = \mu\super{\ell}|^{\cK}$ so that
    \begin{align}
        \Pr[\cG(\bm{\mu}) \text{ has witness } \cK] &\le \sum_{\sh \tau\super{\ell} \subseteq \cK\super{\ell}:|\tau\super{\ell}|\ge a |\cK|} \Pr[\tau\super{\ell} \subseteq \mu\super{\ell}] \\
        &\le \left(\frac{p}{p_1}\right)^{a K}, \quad p_1 =2^{-\Delta_{\sh}/a}
    \end{align}
    Hence, by Lemma \ref{lem:witness-entropy}, we have
    \begin{align}
        \Pr[\cG(\bm{\mu}) \text{ has witness}] &\le \sum_{K\ge b\noisy{\kappa}^{-\ell}} \sum_{\cK \text{ cluster meets } \cV\super{1}: |\cK|=K} \left(\frac{p}{p_1}\right)^{a K} \\
        &\le |\cC\super{0}|\frac{\rho^{b \noisy{\kappa}^{-\ell}}}{1-\rho}
    \end{align}
    where
    \begin{equation}
        \rho = \left(\frac{p}{p_{\rm wit}}\right)^{a},\quad p_{\rm wit} = (e\Delta_{\cG})^{-1/a} p_1
    \end{equation}
    Set $\rho \le 1/2$ and the statement follows.
\end{proof}

%% file: enhance.tex
\section{Threshold and Enhancement}
\label{sec:threshold}
Consider the setup in Remark \ref{rem:setup} and the noise model reduction in Theorem \ref{thm:noise-model-reduction}.
In this section, we shall show that the effective seed level errors $E\super{0}(t),\mu\super{0}(t^+)$ are jointly local stochastic with rate $\ll p$.

\subsection{Active History}

% By Lemma \ref{lem:seed-lvl-reduction} and Theorem \ref{thm:noise-model-reduction}, the seed decoder (disregarding the residue) at time $t$ observes seed level syndrome $\noisy{\sigma}\super{0}$, which is obtained via the history with input
% \begin{equation}
%     \label{eq:t-round-syndrome}
%     \noisy{\sigma}\super{\ell} = \delta E\super{\ell} + \mu\super{\ell}_+
% \end{equation}
% where we have omitted the time $t$ dependence for simplicity.
% Note that $\delta E\super{\ell}$ is a valid lifted syndrome, and thus similar to Section \ref{sec:threshold}, we can adopt the idea of an active subgraph.
By Theorem \ref{thm:noise-model-reduction}, the effective error at time $t$ is the noisy descent of $E\super{\ell} = e\super{\ell} + \hat{\cE}(\mu\super{\ell}_-)$ relative to $\mu_+\super{\ell}$ where 
\begin{equation}
    e\super{\ell} =e\super{\ell}(t), \quad \mu_\pm\super{\ell} = \mu\super{\ell}(t^\pm)
\end{equation}
so we have omitted the time $t$ dependence for simplicity.
Note that $\cE(\mu_-\super{\ell})$ has a special \textit{lifted} structure, and thus it warrants a slight variation of the RG syndrome subgraph in Definition \ref{def:RG-subgraph}.
For later reference, the corresponding noisy syndrome is
\begin{equation}
    \label{eq:t-round-syndrome}
    \noisy{\sigma}\super{\ell} = \delta E\super{\ell} + \mu_+\super{\ell}
\end{equation}
\begin{lemma}[Active History]
    \label{lem:active-history}
    Let $\bm{\rho}$ be a sequence of $i$-level valid syndromes with lifted syndrome $\Sigma \bm{\rho}$, and let $\mu\super{\ell}$ be an $\ell$-level syndrome. 
    Let $\noisy{\sigma}\super{\ell} = \mu\super{\ell} + \Sigma \bm{\rho}$ have noisy history $\noisy{\bm{\sigma}}$. Then
    \begin{equation}
        \noisy{\sigma}\super{i} = \noisy{a}\super{i} + \iota \left(\sum_{j<i} \rho\super{j}\right)
    \end{equation}
    where $\noisy{a}\super{i}$ is initialized with $\noisy{a}\super{\ell} = \mu\super{\ell} +\rho\super{\ell}$ and defined inductively as
    \begin{equation}
        \noisy{a}\super{i-1} = \desc \rep \noisy{a}\super{i} + \rho\super{i-1}
    \end{equation}
    With slight abuse of notation\footnote{Technically, $\noisy{\bm{a}}$ depends on both $\bm{\rho},\mu\super{\ell}$ and not just on $\noisy{\bm{\sigma}}$}, we refer to $\noisy{\bm{a}} = (\noisy{a}\super{i})$ as the \textbf{active history} of $\noisy{\bm{\sigma}}$, since the cleaning $f(\rep \noisy{\sigma}\super{i}) = f(\rep \noisy{a}\super{i})$.
\end{lemma}

\begin{proof}
    Follows from the passthrough property of the cleaning in Definition \ref{def:cleaning} and the valid syndrome invariance of the syndrome repair.
\end{proof}

\begin{definition}[Active Subgraph]
    Let $\bm{\rho}$ be a sequence of valid syndromes and $\noisy{\sigma}\super{\ell} = \mu\super{\ell} + \Sigma\bm{\rho}$ induce history $\noisy{\bm{\sigma}}$, cleanings $\bm{f}$ and active history $\noisy{\bm{a}}$.
    Then the \textbf{active} RG syndrome subgraph $\cG_{\rm act} (\noisy{\bm{\sigma}})$ is the subgraph induced by vertices
    \begin{equation}
        \cG_{\rm act} (\noisy{\bm{\sigma}}) = \noisy{a}\super{0}\cup \bigcup_{i \ge 1} \sh (\noisy{a}\super{i} \cup \rep\noisy{a} \cup f\super{i} \cup \delta f\super{i})
    \end{equation}
    Note that $\cG(\noisy{\bm{\sigma}})$ in Definition \ref{def:RG-subgraph} does not include $\noisy{\sigma}\super{0}$ explicitly since it is implicitly included in the shadows, but for the active subgraph, $\rho\super{0}$ changes this.
\end{definition}

\begin{lemma}[Witness Equivalence]
    \label{lem:witness-equivalence}
    Let $\noisy{\sigma}\super{\ell} = \mu\super{\ell} + \Sigma \bm{\rho}$ where $\bm{\rho}$ is a sequence of valid syndromes, and $\noisy{\bm{\sigma}}$ denote its history.
    Then the RG subgraph $\cG(\noisy{\bm{\sigma}})$ has a witness iff the active RG subgraph $\cG_{\rm act}(\noisy{\bm{\sigma}})$ has a witness.
\end{lemma}

\begin{proof}
    By Lemma \ref{lem:active-history}, $\noisy{\sigma}\super{0} = \noisy{a}\super{0}$ and thus the statement follows.
\end{proof}

Unlike $\cG(\noisy{\bm{\sigma}})$, the witness of $\cG_{\rm act}(\noisy{\bm{\sigma}})$ may be small, since it does not necessarily extend to the top level, e.g., only $\rho\super{0}\ne 0$.
However, such sources can be \textit{attached} to connected components of $\cG(\bm{\rho})$.
In the case of Eq. \eqref{eq:t-round-syndrome}, the goal is thus to estimate the size of the \textit{temporal} subgraph, i.e., the active subgraph + attachments due to $\bm{\mu}_\pm$, which we postpone to the next subsection.

\begin{lemma}[Cell-Layer Invariance]
    \label{lem:active-cell-layer-invariance}
    Let $\bm{\rho}$ is a sequence of valid syndromes and $\noisy{\sigma}\super{\ell} = \mu\super{\ell} +\Sigma\bm{\rho}$ have history $\noisy{\bm{\sigma}}$ and active history $\noisy{\bm{a}}$.
    Let $\cK$ be a connected component of the active RG subgraph $\cG_{\rm act} (\noisy{\bm{\sigma}})$. 
    If $i\in [\ell)$, then
    \begin{equation}
        \noisy{a}\super{i}|_{\cK} = \noisy{a}\super{i}|^{\cK}
    \end{equation}
    For notation consistency, we also write $\noisy{a}\super{0}|^{\cK} = \noisy{a}\super{0}|_{\cK}$.
\end{lemma}
\begin{proof}
    Similar to Lemma \ref{lem:cell-layer-invariance} and thus omitted.
\end{proof}
\begin{lemma}[Correction Invariance]
    \label{lem:active-correction-invariance}
    Let $\bm{\rho}$ is a sequence of valid syndromes and $\noisy{\sigma}\super{\ell} = \mu\super{\ell} +\Sigma\bm{\rho}$ have history $\noisy{\bm{\sigma}}$ and active history $\noisy{\bm{a}}$.
    Let $\cK$ be a connected component of the active RG subgraph $\cG_{\rm act} (\noisy{\bm{\sigma}})$. 
    If $i\in [\ell]$, then $\dnum^{\cK} \rep \noisy{a}\super{i}$ is clean and
    \begin{equation}
        \dnum^{\cK} f(\rep \noisy{a}\super{i}) = f(\dnum^{\cK} \rep \noisy{a}\super{i}), \quad \dnum^{\cK}\delta f(\rep \noisy{a}\super{i}) = \delta f( \dnum^{\cK} \rep \noisy{a}\super{i})
    \end{equation}
    In particular,
    \begin{equation}
        \dnum_{\cK} \desc \rep \noisy{a} = \desc \dnum^{\cK}\rep\noisy{a}\super{i}
    \end{equation}
\end{lemma}
\begin{proof}
    Follows similarly as in Lemma \ref{lem:correction-invariance}.
\end{proof}

\begin{lemma}[Equivalent Repair]
    \label{lem:equivalent-repair}
    Let $\bm{\rho}$ be valid syndromes, $\bm{\mu}$ be the history of $\mu\super{\ell}$, and $\noisy{\bm{a}}$ be the active history of $\noisy{\sigma}\super{\ell} = \mu\super{\ell} +\Sigma \bm{\rho}$.
    Then $\noisy{a}\super{i},\mu\super{i}$ differ by a valid syndrome for all $i$. In particular, the syndrome repair satisfies
    \begin{equation}
        \nu( \noisy{a}\super{i} ) = \nu (\mu\super{i} )
    \end{equation}
\end{lemma}
\begin{proof}
    By definition, $\noisy{a}\super{\ell} = \mu\super{\ell} + \rho\super{\ell}$ satisfies the statement.
    Let us now proceed by induction and note that
    \begin{align}
        \noisy{a}\super{i-1} &=\desc \rep \noisy{a}\super{i} + \rho\super{i-1} \\
        &= \desc (\rep \mu\super{i} + \im \delta) + \im \delta \\
        &= \desc \rep \mu\super{i} + \im \delta \\
        &= \mu\super{i-1} + \im \delta
    \end{align}
    where we used the fact that $\desc$ is linear on clean syndromes modulo stabilizers. The statement then follows by induction.
\end{proof}
\begin{lemma}[Component Size]
    \label{lem:active-component-size}
    Let $\bm{\rho}$ is a sequence of valid syndromes and $\noisy{\sigma}\super{\ell} = \mu\super{\ell} + \Sigma \bm{\rho}$ have history $\noisy{\bm{\sigma}}$ and active history $\noisy{\bm{a}}$. Let $\cK$ be a connected component of the active RG subgraph $\cG_{\rm act} (\noisy{\bm{\sigma}})$. Then
    \begin{equation}
        |\cK| = O\left(|\mu\super{i}|^{\cK} +\sum_{i=0}^{\ell} |\rho\super{i}|^{\cK}\right) 
    \end{equation}
    where the constant in $O$ only depends on $\kappa,c_f,\ve$.
    Moreover,
    \begin{equation}
        (\noisy{\kappa}\super{\ell}+\kappa^{\ell}) |\mu\super{\ell}|^{\cK} + \sum_{j=0}^{\ell} \kappa^j |\rho\super{j}|^{\cK} \ge |\noisy{a}\super{0}|_{\cK}
    \end{equation}
\end{lemma}
\begin{proof}
    Similar to Lemma \ref{lem:witness-size}, we can obtain
    \begin{equation}
        |\cK\super{i}| \le (1+c_f) \frac{2}{1-\ve} |\noisy{a}\super{i}|^{\cK} + 1\{i=1\}|\noisy{a}\super{0}|_{\cK}
    \end{equation}
    By Lemma \ref{lem:active-correction-invariance} and \ref{lem:equivalent-repair}, we have
    \begin{align}
        |\noisy{a}\super{i}|^{\cK} &\le |\desc \rep \noisy{a}\super{i+1}|^{\cK} + |\rho\super{i}|^{\cK} \\
        &\le \kappa |\nu (\mu\super{i+1}) + \noisy{a}\super{i+1})|^{\cK} + |\rho\super{i}|^{\cK} \\
        &\le \kappa \frac{1+\ve}{1-\ve} |\mu\super{i+1}|^{\cK} + \kappa |\noisy{a}\super{i+1}|^{\cK} + |\rho\super{i}|^{\cK} 
    \end{align}
    By inudction, it's straightforward to check that
    \begin{equation}
        |\noisy{a}\super{i}|^{\cK} \le \frac{1+\ve}{1-\ve} \sum_{j>i} \kappa^{j-i} |\mu\super{j}| + \kappa^{\ell-i} |\mu\super{\ell}| + \sum_{j\ge i} \kappa^{j-i} |\rho\super{j}|
    \end{equation}
    By Lemma \ref{lem:noisy-contraction}, we have $|\mu\super{j}| \le \noisy{\kappa}^{\ell-j}|\mu\super{\ell}|$ and thus
    \begin{equation}
        |\noisy{a}\super{i}|^{\cK} \le  (\noisy{\kappa}^{\ell-i} + \kappa^{\ell-i}) |\mu\super{\ell}| + \sum_{j\ge i} \kappa^{j-i} |\rho\super{j}|
    \end{equation}
    Take $i=0$ and the second inequality follows.
    Also note that
    \begin{equation}
        |\cK| \le (1+c_f) \frac{2}{1-\ve} \left( \left(\frac{1}{1-\noisy{\kappa}} + \frac{1}{1-\kappa}\right)|\mu\super{\ell}|^{\cK} + \frac{1}{1-\kappa}\sum_{i=0}^{\ell} |\rho\super{i}|^{\cK}\right)
    \end{equation}
\end{proof}

\subsection{Temporal Graph}

As noted in Theorem \ref{thm:noise-model-reduction}, the seed level effective error $E\super{0}(t)$ is memory-less, i.e., only depends on adjacent time steps.
This warrants the introduction of a \textit{temporal} graph.

\begin{definition}[Temporal Graph]
    \label{def:temporal-graph}
    The \textbf{temporal} RG syndrome graph $\cJ$ is constructed by three independent copies of the RG syndrome graph $\cG$ with vertices labeled by $\cV\times \{-,0,+\}$, which correspond to the times of measurement errors and physical errors.
    In additional to the \textbf{spatial} edges in each copy of $\cG$, add the \textbf{temporal} edges $(u,0)(u',\pm)$ if $u=u'$ or $uu'$ is a spatial edge in $\cG$.
    Note that its max degree is $\Delta_{\cJ}\le 3\Delta_{\cG}+2$.
    For any subset $\cS\subseteq \cJ$, we wrte $\cS_0,\cS_\pm$ to denotes it restrictions to copies $0,\pm$ respectively.
\end{definition}

Intuitively, we shall let the active subgraph $\cG_{\rm act}(\noisy{\bm{\sigma}})$ live in the copy $\cG\times 0$. 
As discussed previously, the active subgraph may have a small witness due to the sources, which is a result of measurement errors at times $t^-$.
Hence, we will examine the subgraphs $\cG(\bm{\mu}_\pm)$ \textit{attached} to the active subgraph, which themselves live in the copies $\cG\times \pm$.

\begin{definition}[Temporal Subgraph]
    \label{def:temporal-subgraph}
    Let $e\super{\ell},\mu_\pm\super{\ell}$ be given and let $\bm{\mu}_\pm,\bm{f}_\pm$ be the history and cleaning of $\mu_\pm\super{\ell}$. 
     Define the disjoint union of qubits
    \begin{equation}
        \Lambda(e\super{\ell},\mu_-\super{\ell}) = e\super{\ell} \sqcup \bigsqcup_{i=1}^{\ell} f_-\super{i}, \quad f_-\super{i}\equiv f(\rep \mu_-\super{i})
    \end{equation}
    We say that any subset $\lambda\subseteq \Lambda$ is \textbf{admissible} relative to $e\super{\ell},\mu_-\super{\ell}$. It then defines a top level physical
    \begin{equation}
        E\super{\ell}(\lambda) = \lambda \cap e\super{\ell} + \iota \left(\sum_{i=1}^{\ell} \lambda \cap f_-\super{i}\right)
    \end{equation}
    Let $\noisy{\sigma}\super{\ell}(\lambda) = \mu_+\super{\ell} + \delta E\super{\ell}(\lambda)$ have history $\noisy{\bm{\sigma}}(\lambda)$ and active history $\noisy{\bm{a}}$.
    The \textbf{temporal} RG subgraph $\cJ(\lambda,\mu_+\super{\ell})$ is then defined as the subgraph of $\cJ$ induced by vertices in
    \begin{equation}
        \cJ(\lambda,\mu_+\super{\ell}) = \cG(\bm{\mu}_-) \sqcup \cG_{\rm act}(\noisy{\bm{\sigma}}(\lambda)) \sqcup \cG(\bm{\mu}_+)
    \end{equation}
    where we have implicitly assumed that $\cG(\bm{\mu}^\pm)$ live in copies labeled by $\pm$ and the active subgraph lives in copy $0$.
    The \textbf{terminal checks} is the collection of seed level checks
    \begin{equation}
        \cT(\lambda, \mu_+\super{\ell}) =(\noisy{a}\super{0} \times 0) \sqcup (\mu\super{0}_+\times +)
    \end{equation}
    Note that copy $-$ does not matter since we are only interested in induced seed levels errors at time steps $0,+$.
    A \textbf{temporal witness} is a connected component meeting the terminal checks.
\end{definition}

\begin{lemma}[Temporal Witness]
    \label{lem:temporal-witness-size}
    Let $\lambda$ be admissible relative to $e,\mu_-$ and $\cJ(\lambda,\mu_+\super{\ell})$ have nonempty temporal witness $\cK$ and terminal checks $\cT$.
    Then
    \begin{equation}
        \quad N(\cK)= \Omega( |\cK|), \quad |\cK| = \Omega(\noisy{\kappa}^{-\ell})|\cK\cap \cT|
    \end{equation}
    where $\Omega,O$ only depend on constants in Eq. \eqref{eq:constants}, and $N(\cK)$ denotes the \textbf{top level} weight,
    \begin{equation}
        N(\cK) = |\delta \ubar{e}\super{\ell}|^{\cK_0} + |\mu_+\super{\ell}|^{\cK_+} + |\mu_-\super{\ell}|^{\cK_-}
    \end{equation}
    % and
    % \begin{equation}
    %     c_R = \frac{1-2\kappa}{2c_f}, \quad C_f = 1+\frac{3}{1-\kappa} + \frac{3c_f}{(1-\kappa)(1-2\kappa)} + \frac{6}{1-2\kappa}
    % \end{equation}
\end{lemma}

\begin{proof}
    By Lemma \ref{lem:active-component-size}, can upper bound the size of $|\cK_0|$ via the measurement error $|\mu_+\super{\ell}|^{\cK_0}$ and the valid syndromes $|\delta \ubar{E}\super{i}|^{\cK_0}$.
    Note that by definition, if $z\in (\delta \ubar{f}_-\super{i})|^{\cK_0}$ for $i\in [\ell)$, then $z \succ \ubar{f}_-\super{i}$.
    Since $\ubar{f}_-\super{i}\subseteq f_-\super{i}$, we see that $z\succ f_-\super{i}$ and thus $\sh z \succeq \sh f_-\super{i}$.
    Note that $\sh f_-\super{i}$ is contained in $\cG(\bm{\mu}_-)$.
    Hence, by connectivity of temporal graph, we see that
    \begin{align}
        |\delta \ubar{f}_-\super{i}|^{\cK_0} &\le \sum_{z} 1\{\sh z \in \cK_0\} \sum_{q} 1\{z\succ q\} 1\{q\in f_-\super{i}\} \\
        &\le \sum_{q} 1\{q\in f_-\super{i}|^{\cK_-}\} \sum_{z} 1\{z\succ q\} \\
        &\le \fd |f_-\super{i}|^{\cK_-}
    \end{align}
    Including the top level case, we have for $i\in [\ell]$, that
    \begin{equation}
        |\delta \ubar{E}\super{\ell}|^{\cK_0} \le |\delta e\super{\ell}|^{\cK_0}  + \fd \Delta_{\sh} (1+c_f) \frac{2}{1-\ve} \sum_{i=1}^{\ell} |\mu_-\super{i}|^{\cK_-}
    \end{equation}
    where we used the fill bound of cleaning map and the near half-fault property of the repair syndrome.
    Also note that if $z\in \mu_+\super{\ell}|^{\cK_0}$, then $\sh z\in \cK_+$ by connectivity of the temporal graph. Hence,
    \begin{equation}
        |\mu_+\super{\ell}|^{\cK_0} \le |\mu_+\super{\ell}|^{\cK_+}
    \end{equation}
    Hence, by Lemma \ref{lem:active-component-size}, 
    \begin{equation}
        |\cK_0| = O(N(\cK))
    \end{equation}
    where we used the noisy contraction.
    The case is similar for $|\cK_\pm|$ by Lemma \ref{lem:witness-size} and thus we obtain the upper bound $|\cK|= O(N(\cK))$.
    One can similarly show that $N(\cK)=\Omega(\noisy{\kappa}^{-\ell})|\cK\cap \cT|$ using Lemma \ref{lem:active-component-size}.
    A lower bound of $|\cK|$ is easily achieved by noting that $|\cK|\ge N(\cK)/\Delta_{\sh}$.
    % Put 
    % \begin{equation}
    %     h_i=|\mu_+\super i|^{\cK_+}+|\mu_-\super i|^{\cK_-}, \quad u=|\delta\ubar e|^{\cK_0}
    % \end{equation}
    % Each $\cK_\pm$ is a union of complete measurement-history components, so Lemma \ref{lem:noisy-contraction} gives $h_i\le\noisy{\kappa}^{\ell-i} h_{\ell}$.
    % By definition of $\Delta_{\sh}$, the fill bound $c_f$ and near half fault, 
    % \begin{equation}
    %     |f(\rep \mu_{\pm}\super i)|^{\cK_\pm} \le L|\mu_\pm\super i|^{\cK_\pm},\qquad
    %     L=\Delta_{\sh}(1+c_f) \frac{2}{1-\ve}.
    % \end{equation}
    % A source check in $\cK_0$ produced by $g_\pm$ is adjacent to its producing qubit. Their shadows are equal or spatially adjacent, so the temporal edges place that qubit in $\cK_\pm$. Consequently,
    % \begin{equation}
    %     |s\super i|^{\cK_0}\le \fd Lh_i+ 1\{i=\ell\} u,\qquad s\super0=0.
    % \end{equation}
    % Apply Lemma~\ref{lem:active-component-size} to all active components in $\cK_0$ and Lemma~\ref{lem:witness-size} to the attached measurement components. This gives
    % \begin{equation}
    %     |\cK|\le \frac{1+c_f}{1-\kappa} u + \frac{1+c_f}{1-\noisy{\kappa}} \frac{2}{1-\ve} \left( \fd \Delta_{\sh} \frac{1+c_f}{1-\kappa}  + 1 \right)h_{\ell}\le K(u+h_{\ell})
    % \end{equation}
    % for a uniform $K$. 
    % Since at least one active component is a witness,
    % \begin{align}
    %     1&\le\sum_{i=1}^{\ell}\kappa^i|s\super i|^{\cK_0}\\
    %     &\le\kappa^\ell u+\fd L h_{\ell}\sum_{i=1}^{\ell}\kappa^i\noisy{\kappa}^{\ell-i} \\
    %     &\le H\noisy{\kappa}^\ell (u+h_{\ell}),\qquad H=1+ \fd L.
    % \end{align}
    % The statement then follows.
\end{proof}

% \begin{lemma}[Temporal Witness Entropy]
%     \label{lem:temporal-witness-entropy}
%     The number of connected clusters of the temporal RG syndrome graph $\cJ$ meeting $\cV\super{1}\times 0$ and of size $M$ is bounded by 
%     \begin{equation}
%         |\cV\super{1}| (e\Delta_{\cJ}^2)^M
%     \end{equation}
% \end{lemma}
% \begin{proof}
%     Similar to Lemma \ref{lem:witness-entropy} and thus omitted.
% \end{proof}

\begin{lemma}[Terminal Probability]
    \label{lem:terminal-probability}
    Let $e\super{\ell},\mu\super{\ell}_\pm$ be joint local stochastic with rate $p$. Let $\cS$ be a fixed subset of checks $z\times t$ in $\cJ$ where $t=0,+$.
    Then there exists constants $p_{\rm term},a,c>0$, depending only on the constants in Eq. \eqref{eq:constants}, such that if $p/p_{\rm term} \le a$, then
    \begin{equation}
        \Pr[\cA_{\cS}] \le  2 q^{|\cS|}, \quad q = \left(\frac{p}{p_{\rm term}}\right)^{c \noisy{\kappa}^{-\ell}}
    \end{equation}
    where $\cA_{\cS}$ is the event that there admissible $\lambda$ relative to $e\super{\ell},\mu_-\super{\ell}$ such that $\cS \subseteq \cT(\lambda,\mu_+\super{\ell})$, and $c$ is some constant.
\end{lemma}

\begin{proof}
    If $\cA_{\cS}$ occurs, then there exists a union $\cK$ of connected clusters in $\cJ$ such that $\cS\subseteq \cK$, each connected cluster in the union intersects $\cS$, and that each connected cluster is a witness of $\cJ(\lambda,\mu_+\super{\ell})$.
    If the union has fixed size $|\cK|= K$, then the number (entropy) of such unions is upper bounded by \cite{gottesman2013fault}
    \begin{equation}
        e^{|\cS|-1} (e\Delta_{\cJ})^{K-|\cS|}
    \end{equation}
    By Lemma \ref{lem:temporal-witness-size}, we see that given such a fixed union $\cK$, we have
    \begin{equation}
        |\cK| =\Omega(\noisy{\kappa}^{-\ell})  |\cS|, \quad N(\cK) =\Omega(|\cK|)
    \end{equation}
    Hence,
    \begin{equation}
        \Pr[\cA_{\cS}] \le \sum_{K=\Omega(\noisy{\kappa}^{-\ell})|\cS|} \sum_{|\cK|=K} \Pr[\cJ(\lambda,\mu_+\super{\ell}) \text{ has witnesses } \cK]
    \end{equation}
    If the union $\cK$ consists of witness of $\cJ(\lambda,\mu_+\super{\ell})$, then by Lemma \ref{lem:temporal-witness-size}, there must exist $\ell$-level syndromes $\tau_t,t=0,\pm$ with shadows $\sh \tau_t \subseteq \cK_{t}$ such that
    \begin{equation}
        \tau_0 = \delta \ubar{e}\super{\ell}|^{\cK_0}, \quad \tau_\pm=\mu_\pm\super{\ell}|^{\cK} 
    \end{equation}
    and that
    \begin{equation}
        |\tau|\equiv |\tau_0| + |\tau_+| +|\tau_-| = \Omega(|\cK|)
    \end{equation}
    In particular, we see that $\tau_0$ is covered by $e\super{\ell}$ and $\tau_\pm$ are covered by $\mu_\pm\super{\ell}$.
    Hence,
    \begin{equation}
        \Pr[\cA_{\cS}] \le \sum_{K=\Omega(\noisy{\kappa}^{-\ell})|\cS|} \sum_{|\cK|=K} \sum_{|\tau|=\Omega(|\cK|):\sh \tau_t \subseteq \cK_t}\Pr[\tau_0 \text{ covered by } e\super{\ell},\tau_\pm \text{ covered by } \mu_\pm\super{\ell}] 
    \end{equation}
    The statement the follows by Lemma \ref{lem:cover-probability}.
\end{proof}

\subsection{Noise Transfer}

\begin{lemma}[Terminal Dependence]
    \label{lem:terminal-dependence}
    Consider the notation in Definition \ref{def:temporal-subgraph}. For every $q\in \Lambda$, let $w_q\in \cC\super{0}$ denotes it seed level (composite) shadow, i.e., $w_q =\sh\circ \cdots \circ \sh q$.
    Then $\noisy{a}\super{0}(\lambda)|_z$ depends on $\lambda$ only implicitly through qubits $q$ with $w_q \preceq z$, i.e.,
    \begin{equation}
        \noisy{a}\super{0}(\lambda)|_z = \noisy{a}\super{0}(\lambda\cap \{q:w_q\preceq z\})|_z
    \end{equation}
\end{lemma}
\begin{proof}
    Let $q\in \Lambda$ be such that $w_q\not\preceq z$ and consider the symmetric difference $\lambda +q$. 
    Let $E\super{0}(\lambda)$ be the seed level noisy descent of $E\super{\ell}(\lambda)$ relative to $\mu_+\super{\ell}$.
    By Lemma \ref{lem:RG-stability}, we see that $E\super{0}(\lambda) + E\super{0}(\lambda+q)$ is supported within $\{\succeq w_q\}$.
    Also note that by Lemma \ref{lem:active-history}, we have
    \begin{equation}
        \noisy{a}\super{0}(\lambda) +\noisy{a}\super{0}(\lambda+q) = \delta  (E\super{0}(\lambda) + E\super{0}(\lambda+q))
    \end{equation}
    Hence, the difference is supported within $\{\succeq w_q\}$. 
    Since $z\not\succeq w_q$, we see that $\noisy{a}\super{0}(\lambda)|_z = \noisy{a}\super{0}(\lambda+q)|_z$ and thus the claim follows.
\end{proof}

% \begin{definition}[$Z$ Graph]
%     \label{def:Z-graph}
%     Let $\cG_Z(C\super{0})$ denote the graph with vertices $\cZ\super{0}$ and edges $zz'$ if there exists cell $w\in \cC\super{0}$ such that $w\preceq z,z'$.
%     Note that $\cG_Z(C\super{0})$ is of bounded degree, depending only on $\fw,\fd$.
%     We refer to $\cG_Z(C\super{0})$ the $Z$ graph \textbf{induced} by $C\super{0}$.
% \end{definition}

\begin{lemma}[Admissible Cover]
    \label{lem:admissible-cover}
    Consider the notation in Definition \ref{def:temporal-subgraph} and let $\cG_Z(C\super{0})$ denote the $Z$ graph of the seed code as in Example \ref{ex:code-graph}.
    Let $\cB(e\super{\ell},\mu_\pm\super{\ell})\subseteq \cZ\super{0}$ denote the collection of seed level $z$ checks such that $\noisy{a}\super{0}(\lambda)|_z\ne 0$ for some admissible $\lambda$.
    For any $\cA \subseteq \cB$, there exists $\cS\subseteq \cA$ is a maximmaly independent (relative to the $Z$ graph) subset, then there exists admissible $\lambda\subseteq \Lambda$ such that
    \begin{equation}
        \cS\subseteq \noisy{a}\super{0}(\lambda)
    \end{equation}
\end{lemma}
\begin{proof}
    Let $\cS$ be a maximally independent subset of $\cA$ (relative to the $Z$ graph).
    For every $z\in \cS\subseteq \cB$, we can choose admissible $\lambda_z$ such that the seed level active history $\noisy{a}\super{0}(\lambda_z)|_z=1$.
    Define
    \begin{equation}
        \lambda_{\cS} = 
        \begin{cases}
            \lambda_z \cap \{q: w_q\preceq z\} & z\in \cS \\
            0 & \text{otherwise}
        \end{cases}
    \end{equation}
    Since all checks in $\cS$ are non-adjacent, we see that $\lambda_{\cS}$ is well-defined. By the terminal dependence, we see that 
    \begin{equation}
        \cS \subseteq \noisy{a}\super{0}(\lambda_{\cS})
    \end{equation}
\end{proof}
\begin{theorem}[Noise Transfer]
    \label{thm:noise-transfer}
    Consider the noise model reduction in Theorem \ref{thm:noise-model-reduction}. Then there exists $p_{\rm RG},a,c>0$, depending only on the constants in Eq. \eqref{eq:constants}, such that if the joint stochastic rate $p$ of top level errors satisfies $p/p_{\rm RG} \le a$, then effective seed level errors are local stochastic with rate
    \begin{equation}
        q = \left(\frac{p}{p_{\rm RG}}\right)^{c \noisy{\kappa}^{-\ell}}
    \end{equation}
    The same estimate holds jointly for disjoint copies of this three-time gadget whose underlying physical and measurement fault locations are disjoint, even when those faults are correlated.
\end{theorem}
\begin{proof}
    By Remark \ref{rem:cover-equivalence}, it's sufficient to consider a fixed subset of seed level $z$-checks at distinct time steps, i.e., fault locations
    \begin{equation}
        \tau = \bigsqcup_{t} \tau(t) \sqcup \tau(t^+), \quad  \tau(t),\tau(t^+)\subseteq \cZ\super{0}
    \end{equation}
    For exposition purpose, we shall only consider $\tau=\tau(t)$ for a fixed time step $t$ (and thus omitted); the general case can be similarly derived.
    Specifically, we shall upper bound the probability
    \begin{equation}
        \Pr[\tau \text{ covered by } \ubar{E}\super{0}]
    \end{equation}
    where $\ubar{E}\super{0}$ is a to be determined seed level physical, which is stabilizer equivalent to the effective seed level error $E\super{0}$ as defined in Theorem \ref{thm:noise-model-reduction}.
    As before, let us define
    \begin{equation}
        E\super{\ell} = e\super{\ell} +\hat{\cE}(\mu_-\super{\ell})
    \end{equation}
    Consider an increasing sequence $\lambda_{\bullet}$ of admissible subsets relative to $e\super{\ell},\mu_-\super{\ell}$ constructed in some deterministic manner such that the increment $\lambda_i \backslash \lambda_{i-1}$ is a single qubit $q_i$ in $\Lambda(e\super{\ell},\mu_-\super{\ell})$ for each $i$.
    Let
    \begin{equation}
        \noisy{\sigma}\super{\ell}(\lambda_{\bullet}) = \delta E\super{\ell}(\lambda_{\bullet}) + \mu_+\super{\ell}
    \end{equation}
    By Lemma \ref{lem:RG-stability}, if $E\super{0}(\lambda_{\bullet})$ denotes the noisy descent of $E\super{\ell}(\lambda_{\bullet})$ relative to $\mu_+\super{\ell}$, then, modulo stabilizers,
    \begin{equation}
        E\super{\ell} + \hat{\cE}(\noisy{\sigma}\super{\ell}) + \hat{\cE}(\mu_-\super{\ell}) = \sum_{\bullet} h_{\bullet} + \im\delta, \quad h_{\bullet} = E\super{0}(\lambda_\bullet) + E\super{0}(\lambda_{\bullet-1})
    \end{equation}
    Note that by Lemma \ref{lem:RG-stability} and \ref{lem:terminal-dependence}, each difference $h_{\bullet}$ is supported within $\{\succeq w_{\bullet}\}$ (covered by $w_{\bullet}$) at the seed level where $w_{\bullet}\in \cC\super{0}$ is the seed level shadow of $q_{\bullet}$.
    Hence, by local triviality in Remark \ref{rem:setup}, if $\delta h_{\bullet} =0$, then it must be a stabilizer $\in \im \delta$.
    Therefore, we can consider the equivalent representation of the effective seed level physical error
    \begin{equation}
        \ubar{E}\super{0} = \sum_{\bullet: \delta h_{\bullet} \ne 0} h_{\bullet}
    \end{equation}
    Note that if $\noisy{\bm{a}}(\lambda_{\bullet})$ is the noisy history of $\noisy{\sigma}\super{\ell}(\lambda_{\bullet})$, then by Lemma \ref{lem:perturbation} and \ref{lem:active-history}, we have
    \begin{equation}
        \delta h_{\bullet} = \noisy{a}\super{0}(\lambda_{\bullet}) +\noisy{a}\super{0}(\lambda_{\bullet-1})
    \end{equation}
    Hence, removing $h_{\bullet}$ with $\delta h_{\bullet}=0$ from the error $E\super{0}\mapsto \ubar{E}\super{0}$ means that at least one of active histories has a terminal check.
    Note that if $\tau$ is covered by $\ubar{E}\super{0}$ (relative to the induced graph of $C\super{0}$), then it must be covered by $\cB$ (relative to the induced $Z$ graph of $C\super{0}$) as defined in Lemma \ref{lem:admissible-cover}.
    Indeed, if $z\in \tau$, then $z \succ h_i$ for some $\delta h_i\ne 0$. Since $h_i$ is supported in $\{\succeq w_i\}$, we see that $z\succ w_i$.
    Since $\delta h_i \ne 0$, we can assume without loss of generality that there exists $z'$ such that $\noisy{a}\super{0}(\lambda_i)|_{z'}\ne 0$ while $\noisy{a}\super{0}(\lambda_{i-1})|_{z'}= 0$; hence, $z'\in \cB$.
    Since $z' \succ w_i$, we see that $z\in \cN(z')$, where we write $\cN$ for neighborhoods relative to the $Z$ graph, and thus $\tau$ is covered by $\cB$ relative to the $Z$ graph.
    Apply Lemma \ref{lem:admissible-cover} to $\cN(\tau) \cap \cB$, there exists cover $\cS\subseteq \cN(\tau) \cap \cB \subseteq \cN(\tau)$ of $\cN(\tau)\cap \cB$ with 
    \begin{equation}
        \quad \cS\subseteq \noisy{a}\super{0}(\lambda)
    \end{equation}
    for some admissible $\lambda$. 
    Hence, if $\tau$ covered by $\ubar{E}\super{0}$, then there exists admissible $\lambda$ relative to $e\super{\ell},
    \mu_-\super{\ell}$ such that $\cS$ is a subset of terminal checks of $\cJ(\lambda,\mu_+\super{\ell})$, and thus we wish to apply Lemma \ref{lem:terminal-probability}.
    Note, however, that $\cS$ depends on $\tau$ but also the realization of errors $e\super{\ell},\mu_\pm\super{\ell}$, and thus we shall instead consider the following, which is independent of the realization of errors,
    \begin{equation}
        \{\cS\subseteq \cN(\tau): \tau \subseteq \cN_{2}(\cS)\}
    \end{equation}
    Indeed, our argument implies that given $e\super{\ell},\mu_\pm\super{\ell}$ such that $\tau$ is covered by $\ubar{E}\super{0}$, then there exists $\cS$ in the collection such that event $\cA_{\cS}$ occurs. 
    Hence, by Lemma \ref{lem:cover-size} and \ref{lem:terminal-probability}, if $c=\deg \cG_Z(C\super{0})$, which depends only on $\fw,\fd$, then
    \begin{align}
        \Pr[\tau \text{ covered by } \ubar{E}\super{0}] &\le \sum_{\cS} 1\{\cS\subseteq \cN(\tau),\tau \subseteq \cN_2(\cS)\} \Pr[\cA_{\cS}] \\
        &\le 2 q^{|\tau|/c^2} \sum_{\cS} 1\{\cS\subseteq \cN(\tau)\}\\
        &\le 2 q^{|\tau|/c^2} 2^{c|\tau|} 
    \end{align}
    The statement then follows.
\end{proof}

\begin{theorem}[SS Threshold]
    \label{thm:RG-threshold}
    Consider the setup in Remark \ref{rem:setup}, and the noise model in Definition \ref{def:noise-model} such that physical $e\super{\ell}(t)$ and measurement errors $\mu\super{\ell}(t^+)$ for $t\in [0,T]$ in the top level $C\super{\ell}$ are (jointly) local stochastic with rate $p$.
    Then there exists \textbf{threshold} $p_{\rm RG}>0$ and constants $a,c>0$, all of wchich depending only on the constants in Eq. \eqref{eq:constants} (independent of seed decoder), such that if $p/p_{\rm RG} \le a$, then
    \begin{itemize}
        \item If $C\super{0}$ is equipped with an arbitrary meta-check decoder, then the RG decoder satisfies
        \begin{equation}
            \Pr[T \text{ rounds of SS decoding } C\super{\ell} \text{ fails}]= O(T)|\cQ\super{0}| \left(\frac{p}{p_{\rm RG}}\right)^{c\noisy{\kappa}^{-\ell}}
        \end{equation}
        \item If $C\super{0}$ is equipped with a single-shot meta-check decoder, i.e., below threshold $p_c>0$ of the seed decoder,
        \begin{equation}
            \Pr[T \text{ rounds of SS decoding } C\super{0} \text{ fails}] =O(T) \poly (|\cQ\super{0}|) \left(\frac{p}{p_{c}}\right)^{\Omega(n^{\beta})}
        \end{equation}
        then the RG decoder satisfies
        \begin{equation}
            \Pr[T \text{ rounds of SS decoding } C\super{\ell} \text{ fails}]= O(T) \poly(|\cQ\super{0}|) \left(\frac{p}{p_{\rm RG}}\right)^{\Omega(n^{\beta}\noisy{\kappa}^{-\ell})}
        \end{equation}
        provided that $\ell$ is sufficiently large so that
        \begin{equation}
            q = \left(\frac{p}{p_{\rm RG}}\right)^{c \noisy{\kappa}^{-\ell}} < p_c
        \end{equation}
    \end{itemize}
\end{theorem}

\begin{proof}
    By Theorem \ref{thm:noise-transfer}, the probability that the effective seed level errors are nonzero (at least one seed level effective error occurs) is upper bounded by a union bound and thus the first statement follows.
    The case is similar for the second statement.
\end{proof}

%% file: repair.tex
\section{Local Repair}
\label{sec:local-repair}
This section constructs a local syndrome repair rule that satisfies Definition \ref{def:repair}, so that the resulting CSS code $C\super{\ell}$ equipped with the RG decoder is single-shot.
Again, we shall restrict our attention to the $Z$ sector, since the $X$ sector can be similarly derived.
For this section, let $C\mapsto \bar{\Theta}$ be a sparsity-preserving level map as in Definition \ref{def:level-map} with layer local maps $\iota,\pi,\phi$ forming an SDR.
Let $\bar{\bm{C}}$ be the tower obtained with seed code QLDPC $C\super{0}$ so that $\phi=\phi\super{i}$ has a uniformly bounded column and row weight $\|\phi\|_{\rm c},\|\phi\|_{\rm r}$ among all levels $i\in [\ell]$.
We shall then construct a (nonlinear) local repair starting with (linear) $\phi:MZ(\bar{\Theta}) \to Z(\Theta)$, so that the approximately half-fault property in Definition \ref{def:repair} is satisfied.

To do so, we need the following additional geometric properties.
% For simplicity, let us also assume the following hypothesis, which we later show that can be removed by replacing the seed code with $C\super{\ell_0}$ for some constant $\ell_0$.

\begin{definition}[Metric Contraction]
    The level map $\bar{\Theta}$ has \textbf{metric contraction} with constants $\alpha>1,\beta\ge 0$ if for any cells $u,v$ in $\Theta$,
    \begin{equation}
        \dist_C (\sh u,\sh v)\le \frac{1}{\alpha} \dist_{\Theta}(u,v) +\beta
    \end{equation}
    where $\dist_C$ denotes the graph distance on the \textbf{bare graph}, i.e., vertices $x,q,z$ of $C$ and edges $xq,qz$ (similarly for $\dist_{\Theta}$). 
\end{definition}

\begin{definition}[Meta-Check Locality]
    The $Z$-sector meta-checks of level map $\bar{\Theta}$ are \textbf{local} relative to the bare graph $\Theta$ if there exists constant $\gamma$ such that for any input code $C$ and for any pair of $z,z'\in Z(\Theta)$ checks in a common meta-check, 
    \begin{equation}
        \dist_\Theta(z,z')\le \gamma
    \end{equation}
    The case is similar for the $X$ sector.
\end{definition}

\begin{lemma}[Validity]
    \label{lem:local-validity}
    The level map $\bar{\Theta} = \bar{\Gamma}_*\circ \bar{\Gamma}$ has metric contraction with constant $(4,4)$, and $Z$ sector meta-checks are local relative to the bare graph. The case is similar for $\bar{\Theta}_*=\bar{\Gamma}\circ \bar{\Gamma}_*$ and the $X$ sector meta-checks.
\end{lemma}

\begin{proof}
    See Lemmas~\ref{lem:explicit-metric-contraction} and \ref{lem:explicit-meta-locality} in Appendix~\ref{sec:metric}.
\end{proof}

By Lemma \ref{lem:local-validity}, we can then treat the level map as an abstract object whose explicit construction exists and deferred in the Appendix.
For the remainder of this section, let $\bar{\Theta}$ also exhibit $(\alpha,\beta)$ metric contraction and $\gamma$-local meta-checks (relative to $\Theta$).

\subsection{Bravyi-Haah Separation}

The following property is motivated by the Bravyi-Haah clustering algorithm \cite{bravyi2013quantum},
\begin{lemma}[Polynomial Growth]
    \label{lem:poly-growth}
    Let $d=\log_{\alpha}\Delta_{\sh}$.
    If $i\ge \log_{\alpha} r$, then the $i$ level ball on the bare graph $C\super{i}$ has size
    \begin{equation}
        |B_r(a)| \le Vr^d, \quad B_r(a)\equiv \{\dist(a,\cdot)\le r\}
    \end{equation}
    where $a$ is a cell in $C\super{i}$, and $V=\Delta_{\rm ver}V_0$ and $V_0$ is a bounded constant representing the max size of a ball of radius $1+\beta (1-\alpha^{-1})^{-1}$ in any $C\super{i},i\ge 0$.
\end{lemma}
\begin{proof}
    Put $k=\lceil\log_\alpha r\rceil\le i$. 
    Iterating metric contraction maps $B_r(a)$ into a radius-$R_0$ ball at level $i-k$ where $R_0=1+\beta (1-\alpha^{-1})^{-1}$.    
    Each coarse cell has at most $\Delta_{\sh}^k\le\Delta_{\sh}r^d$ preimages. Uniform sparsity bounds the coarse ball by $V$, proving the claim.
\end{proof}
\begin{lemma}[Separation]
    \label{lem:separation}
    Let $\eta>0$, and $R=R_{\eta}\ge 1$ be an integer such that $(1+\eta)^R >(\gamma R)^d V$, and $i_\eta=\log_{\alpha} (\gamma R)$.
    For any syndrome $\noisy{\sigma}$ in the $i\ge i_{\eta}$ level, there exists subset $\noisy{\tau}\subseteq \noisy{\sigma}$ such that
    \begin{equation}
        |\noisy{\tau}|\le \eta |\noisy{\sigma}|
    \end{equation}
    and every connected component (relative to meta-checks) of $\noisy{\sigma}\backslash \noisy{\tau}$ has size $\le (\gamma R)^d V$.
\end{lemma}

\begin{proof}
    Consider $\noisy{\sigma}$ as the induced subgraph of the bare code $C\super{i}$, and all subsets are implicitly subset of the subgraph $\noisy{\sigma}$.
    Start with $\noisy{\tau} =\varnothing$ and let us proceed by induction.
    If $\noisy{\sigma}\ne \varnothing$, then exists $z \in \noisy{\sigma}$.
    Note that there must exist $t\ge  R$ such that
    \begin{equation}
        |B_{\gamma t}(z)| \le (1+\eta) |B_{\gamma (t-1)}(z)|
    \end{equation}
    Indeed, if the converse holds, then
    \begin{equation}
        |B_{\gamma R}(z)|\ge (1+\eta)^R 
    \end{equation}
    However, by Lemma \ref{lem:poly-growth}, we read a contradiction by our choice of $R$. 
    Add the \textit{annulus} $B_{\gamma t}\backslash B_{\gamma (t-1)}$ to $\noisy{\tau}$. 
    Note that by locality of meta-checks,  no meta-check in $MZ$ intersects both $B_{\gamma (t-1)}(z)$ and $\noisy{\sigma}\backslash B_{\gamma t}(z)$.
    Apply induction on the $\noisy{\sigma}\backslash B_{\gamma t}(z)$ and the statement follows.
\end{proof}

\begin{lemma}[Local Decomposition]
    \label{lem:local-decomposition}
    Let $\sigma$ be a clean syndrome at the $i> i_\eta$ level. 
    Then there exists a disjoint partition $\sigma=\sqcup_i  u_i$ such that $v_i=\phi \bar{\delta} u_i$ satisfies
    \begin{equation}
        |u_i + v_i|\le B_{\eta}, \quad \sum_{i} |v_i| \le \ve |\sigma|
    \end{equation}
    where
    \begin{equation}
        B_{\eta} = (1+\|\phi\|_c\bar{\fd}) (\gamma R_{\eta})^d V, \quad \eta = \frac{\ve}{2 \|\phi\|_c \bar{\fd}}
    \end{equation}
\end{lemma}
\begin{proof}
    By Lemma \ref{lem:separation}, we see that $\sigma$ has subset $\noisy{\tau}\subseteq \sigma$ such that $|\noisy{\tau}|\le \eta|\sigma|$ and $\sigma\backslash \noisy{\tau}$ has connected components (relative to meta-checks)  $\{c\}$, each with size $\le (\gamma R)^d V$.
    Let $\{u_i\}$ be the disjoint union of $\{c\}$ and $\noisy{\tau}$, where we treat every $z$-check in $\noisy{\tau}$ as an individual $u_i$.
    Note that
    \begin{equation}
        |u_i+v_i| \le (1 + \|\phi\|_c \bar{\fd})(\gamma R)^d V \equiv B_{\eta} 
    \end{equation}
    By definition,
    \begin{align}
        \sum_{c} |\bar{\delta} c| &\le  \sum_{M} 1\{\exists c, M\in \delta  \bar{c}\} \\
        &\le \sum_{M} \sum_{z\in \noisy{\tau}} 1\{M \in \delta z\} \\
        &\le \bar{\fd} |\noisy{\tau}|
    \end{align}
    where the summation over $c$ is the connected components (relative to meta-checks) of $\sigma\backslash \noisy{\tau}$ and $M$ denote the meta-check cells. 
    Note that we used the fact that a meta-check can interact with at most one connected component $c$, and that $\bar{\delta}\sigma=0$.
    Hence,
    \begin{equation}
        \sum_{i} |v_i| \le 2 \|\phi\|_{\rm c} \bar{\fd} |\noisy{\tau}| \le \ve |\sigma|
    \end{equation}
\end{proof}

\subsection{Local Update Rule}

Let us collect the constants appearing in Lemma \ref{lem:poly-growth}, \ref{lem:separation} and \ref{lem:local-decomposition}, i.e., $B_{\eta},R_{\eta},i_{\eta}$ and $\ve$.

% determined by
% \begin{equation}
%     \eta = \frac{1}{\|\phi\|_{\rm r}} \frac{1}{ \bar{\fw} +\bar{\fd}} \ve
% \end{equation}
% where $\|\cdot\|_{\rm r}$ is the max row weight.

\begin{definition}[Energy]
    Let $\nu,w$ be syndromes. Then define the \textbf{energy} of $w$ relative to $\nu$ as
    \begin{align}
        E_{\nu}(w)&\equiv |\nu +w| -|\nu| \\
        &= |w| - 2|\nu \cap w|
    \end{align}
\end{definition}

\begin{lemma}[Energy]
    \label{lem:energy}
    The energy $E_{\nu}(w)$ is
    \begin{itemize}
        \item additive on disjoint syndromes, i.e., $E_{\nu}(w+v)=E_{\nu}(w)+E_{\nu}(v)$ if $w\cap v=\varnothing$
        \item independent of disjoint updates, i.e., $E_{\nu+w}(v) = E_{\nu}(v)$ if $w\cap v=\varnothing$
        \item 1-Lipschitz, i.e., $|E_{\nu}(w+v)-E_{\nu}(w)|\le |v|$
    \end{itemize}
\end{lemma}

% \begin{algorithm}
% \caption{\textsc{(Sequential) Local Repair}}
% \label{alg:sequential-local-repair}
% \begin{algorithmic}[1]
%     \Require (Noisy) syndrome $\noisy{\sigma}$
%     \Ensure Correction $\nu(\noisy{\sigma})$ such that $\rep \noisy{\sigma} \equiv \noisy{\sigma} +\nu(\noisy{\sigma})$ is clean
%     \State Initialize $\nu = \phi \bar{\delta} \noisy{\sigma}$
%     \ForEach {$c\in \cL$}
%         \State Update $\nu\gets \nu+c$ if $E_\nu(c)<0$
%     \EndFor 
%     \State Output final $\nu=\nu(\tilde{\sigma})$
% \end{algorithmic}
% \end{algorithm}

\begin{definition}[Local Repair]
    \label{def:local-correction}
    Let $\cL\equiv \cL_{\eta}(\bar{\Theta})$ denote the collection of \textbf{local steps}, i.e., clean connected (relative to meta-checks) syndrome $w\in \ker \bar{\delta}\cap Z(\Theta)$ with bound size $\le B_{\eta}$.
    Define $\nu:Z(\Theta)\to Z(\Theta)$ on $\noisy{\sigma}$ as follows. 
    Specifically, start with $\nu = \phi \bar{\delta} \noisy{\sigma}$.
    If there exists \textbf{admissible} local step, i.e., $w\in \cL$ such that the energy $E_{\nu}(w) <0$,  then update $\nu\gets \nu +w$; otherwise, $\nu$ is \textbf{stationary}. Continue until stationary.
    We refer to the final output $\nu$ as the \textbf{local syndrome repair}, which we define for levels $i > i_{\eta}$.
\end{definition}

\begin{remark}[Stationary]
    Note that depending on the sequential order of the local steps, the final output may be different, and thus we only use the property that the output is stationary.
\end{remark}

\begin{remark}[Run Time]
    \label{rem:run-time}
    Note that by sparsity, $|\cL|=O(|Z(\Theta)|)$, and thus a naive \textit{sweep} of all local steps in $\cL$ runs in linear time.
    Since the repair is updated after each local step, it's possible that a previously non-admissible local step becomes admissible at a later time, and thus a single sweep is insufficient.
    Note, however, that each sweep decreases $|\nu|$ by at least one, and thus at most $|\nu|$ sweeps is sufficient so that the local syndrome repair runs in quadratic time $O(|Z(\Theta)|^2)$.
    Though postponed to subsection \ref{sec:parallel} , this can be further optimized to reach linear run time by utilizing Lemma \ref{lem:energy}.
    In fact, for more practical implementation, we will show that a finite $O(1)$ run time cutoff of the algorithm is sufficient for our purposes. 
\end{remark}

\begin{theorem}[Stationary Energy]
    \label{thm:stationary-energy}
    Consider the local update rule in Definition \ref{def:local-correction} at $i > i_{\eta}$ level and assume that there are no more $B_{\eta}$-admissible steps for $\nu$. Then for any clean syndrome $\sigma\in \ker \bar{\delta}\cap Z\super{i}$,
    \begin{equation}
        E_{\nu}(\sigma) \ge -\ve |\sigma|
    \end{equation}
    In particular, if $\fs$ is a connected (relative to meta-checks) component of the repaired syndrome $\rep \noisy{\sigma} \equiv \noisy{\sigma} + \nu(\noisy{\sigma})$, then
    \begin{equation}
        |\fs| \le \frac{2}{1-\ve} |\noisy{\sigma}\cap \fs|
    \end{equation}
\end{theorem}
\begin{proof}
    Note that by Lemma \ref{lem:energy}, the energy $E_{\nu}$ is additive for disjoint syndromes and 1-Lipschitz.   
    Also, by Lemma \ref{lem:local-decomposition}, we can decompose any clean syndrome $\sigma$ into disjoint $u_i$ such that if $v_i=\phi\bar{\delta}u_i$, then $w_i=u_i+v_i$ is clean and of size $\le B_{\eta}$.
    Hence, any connected component (relative to meta-checks) of $w_i$ cannot be $B_{\eta}$-admissible so that
    \begin{equation}
        E_{\nu} (w_i)\ge 0
    \end{equation}
    Note that
    \begin{align}
        E_\nu (u_i) &\ge E_\nu (w_i) -|v_i| \\
        E_\nu (\sigma) &\ge -\ve |\sigma|
    \end{align}
    Finally, note that if $\fs$ is a connected component (relative to meta-checks) of $\rep \noisy{\sigma}$, then the energy lower bound $E_{\nu}(\fs)$ with $\nu=\nu(\noisy{\sigma})$ implies that
    \begin{equation}
        |\fs| \ge \frac{2}{1+\ve} |\nu\cap \fs|
    \end{equation}
    Since $|\fs| =|\noisy{\sigma}\cap \fs | + |\nu\cap \fs|$, the near half fault bound then follows.
\end{proof}

\begin{remark}[$i_{\eta}$ Level is New Seed]
    Note that the local repair is only defined for levels $i>i_{\eta}$ since Lemma \ref{lem:local-decomposition} is proven only for $i>i_{\eta}$.
    One remedy is to simply \textit{shift} the levels so that $C\super{i_{\eta}}$ is regarded as the new seed code.
    Indeed, the RG decoder is single-shot as long as the seed code is locally trivial.
\end{remark}

\begin{lemma}[Any Level is Locally Trivial]
    Let the seed code $C\super{0}$ be locally trivial as in Remark \ref{rem:setup}.
    Then $C\super{i_*},i_*\ge 1$ is locally trivial.
\end{lemma}
\begin{proof}
    Follows from Lemma \ref{lem:RG-stability}.
\end{proof}

\subsection{Parallelization (Optional)}
\label{sec:parallel}

\begin{algorithm}
\caption{\textsc{Parallel Local Repair}}
\label{alg:parallel-local-repair}
\begin{algorithmic}[1]
    \Require (Noisy) syndrome $\noisy{\sigma}$
    \Ensure Correction $\nu(\noisy{\sigma})$ such that $\rep \noisy{\sigma} \equiv \noisy{\sigma} +\nu(\noisy{\sigma})$ is clean
    \State Initialize $\nu = \phi \bar{\delta} \noisy{\sigma}$
    \RepeatN{$t$}
        \State Update $\nu$ based on a rainbow sweep in Definition \ref{def:parallel-local-repair}
    \EndRepeatN
    \State Output final $\nu=\nu(\tilde{\sigma})$
\end{algorithmic}
\end{algorithm}

As discussed in Remark \ref{rem:run-time}, the simplicity of the local syndrome repair rule in Definition \ref{def:local-correction} comes at a cost of having a naive quadratic run time.
Therefore, the goal of this subsection is to show that a variation of the algorithm can be parallelized to reach a stationary solution in linear run time; in fact, a finite run time cutoff will be sufficient for our purposes.
In particular, the parallelized local RG decoder has logarithmic run time as in Theorem \ref{thm:RG-run-time}. 
The parallel version is inspired by the local small-flip decoder in \cite{fawzi2020constant}.

\begin{definition}[Parallel Local Repair]
    \label{def:parallel-local-repair}
    The \textit{key observation} is that the energy is independent of disjoint updates as in Lemma \ref{lem:energy}, and thus whether a local step is admissible or not does not depend on disjoint steps.
    By sparsity,  $\cL=\cL_\eta(\bar{\Theta})$ can be partitioned into a bounded $\chi=\chi_{\eta}$ \textbf{colors} so that local steps $w$ of the same color are disjoint; let $\zeta:\cL\to [\chi]$ denote the \textbf{color map}.
    Hence, a $\zeta$ \textbf{color sweep} in which local steps of color $\zeta$ are determined whether to be admissible can be run in parallel with $O(1)$ time.
    A \textbf{rainbow sweep} in which all colors are cycled through thus runs in $O(1)$ time since $\chi$ is bounded.
    Note that each rainbow sweep decreases the weight of a non-stationary syndrome by $\ge 1$ and thus reaches a stationary output in linear time.
    However, for our purposes, only finite $t$ rainbow sweeps is sufficient, as in Algorithm \ref{alg:parallel-local-repair}, though the output is no longer stationary.
\end{definition}

\begin{lemma}[Rainbow Sweep]
    \label{lem:rainbow-sweep}
    Put $B=B_{\eta}$.
    Let $\bar{\delta}\noisy{\sigma} = \bar{\delta}\nu$ where $\noisy{\sigma},\nu$ are syndromes at $i\ge i_{\eta}$ level, and $\nu'$ be the output of a rainbow sweep with input $\nu$.
    % Let $\nu\mapsto \nu'$ denote a single \textbf{color sweep}, i.e., update $\nu \mapsto \nu'$ via all admissible local steps of the same color simultaneously.
    Then
    \begin{equation}
        |\nu'|\le r|\nu| + b(1+\ve) |\noisy{\sigma}|
    \end{equation}
    where
    \begin{equation}
        r = 1-b(1-\ve), \quad b = \frac{1}{(1+2\chi B)(1+\bar{\fw}\|\phi\|_{\rm r})}
    \end{equation}
    In particular, if the parallel local repair $\nu$ applies $t$ rainbow sweeps to $\noisy{\sigma}$ as in Algorithm \ref{alg:parallel-local-repair}, then the repair
    \begin{equation}
        |\rep \noisy{\sigma}| \le \left(\frac{2}{1-\ve} + \ve'\right)|\noisy{\sigma}|, \quad \ve'\equiv r^t \|\phi\|_c\bar{\fd}
    \end{equation}
\end{lemma}

\begin{proof}
    Note that $\nu+\noisy{\sigma}$ is a clean syndrome, and thus by Lemma \ref{lem:local-decomposition}, there exists a partition of $\nu+\noisy{\sigma}=\sqcup u_i$ of syndromes such that if $w_i =u_i +\phi\bar{\delta} u_i$, then
    \begin{equation}
        |w_i| \le B_{\eta}, \quad \sum_{i} |\phi\bar{\delta} u_i| \le \ve |\nu+\noisy{\sigma}|
    \end{equation}
    Hence, by disjoint additivity and 1-Lipschitz in Lemma \ref{lem:energy},
    \begin{align}
        \sum_{i} E_\nu(w_i) &\le  \sum_i E_\nu (u_i) + \sum_i |\phi\bar{\delta}u_i| \\
        &\le |\noisy{\sigma}| -|\nu| +\ve (|\sigma|+|\nu|) \\
        &\le (1+\ve) |\noisy{\sigma}| -(1-\ve)|\nu|
    \end{align}
    Note that each connected component (relative to meta-checks) $c$ of $w_i$ is a local step in $\cL$. 
    Conversely, if local step $c$ is given and choose $z\in c$, then there exists at most $1+\|\phi\bar{\delta}\|_{\rm r}\le 1+\bar{\fw} \|\phi\|_{\rm r}$ many $i$ such that $z\in w_i$. 
    Hence, $c$ can be the connected component (relative to meta-checks) of at most $1+\bar{\fw} \|\phi\|_{\rm r}$ many $w_i$.
    Hence,
    \begin{align}
        (1+\bar{\fw} \|\phi\|_{\rm r}) \sum_{c\in \cL} \max(0,-E_\nu (c)) &\ge \sum_{c\in \cL} \sum_{i} \max(0,-E_\nu (c)) 1\{c \text{ connected component of } u_i\} \\
        &\ge - \sum_{i} E_{\nu} (w_i) \\
        &\ge (1-\ve)|\nu| -(1+\ve)|\noisy{\sigma}|
    \end{align}
    Also note that for any $\nu_*$, we have
    \begin{align}
        2|\nu \cap c| - |c| &\le 2|\nu_*\cap c| -|c| +2|(\nu+\nu_*)\cap c| \\
        \max(0,-E_{\nu}(c)) &\le \max(0,-E_{\nu_*}(c)) +2|(\nu+\nu_*)\cap c|
    \end{align}
    Let $\nu=\nu_0$ be the input and define inductively
    \begin{equation}
        \nu_i = \nu_{i-1} +\sum_{u\in \cA_i} u
    \end{equation}
    where $\cA_i$ denotes the admissible local steps when sweeping color $i$ so that $\nu_{i-1}\mapsto \nu_i$ is the output of a (parallel) $i$ color sweep.
    Then $\nu_1,...,\nu_\chi=\nu'$ is the sequential output of sweeping colors $1,...,\chi$, and $\cA$ denotes the collection of all admissible local steps during this random sweep.
    Hence, if we replace $\nu_*$ with the updated $\nu_{\zeta(c)-1}$ when evaluating $c$, then 
    \begin{align}
        \max(0,-E_{\nu}(c)) &\le \max(0,-E_{\nu_{\zeta(c)-1}}(c)) +2\sum_{w\in \cA: \zeta(w)<\zeta(c)} |u\cap c|
    \end{align}
    If both sides are summed over all local steps $c$, then the first term on the right side sums to the weight change due to a rainbow sweep, and thus the inequality compares the parallel and sequential local repairs.
    Hence,
    \begin{align}
        \sum_{c\in \cL_B} \max(0,-E_\nu(c)) &\le \sum_{c\in \cL_B} \max(0,-E_{\nu_{\eta(c)}}(c)) + 2 \sum_{c\in \cL_B} \sum_{u\in \cA: \eta(u)<\eta(c)} |u\cap c| \\
        &\le (|\nu|-|\nu'|) + 2\chi B |\cA| \\
        &\le (1+2\chi B) (|\nu|-|\nu'|)
    \end{align}
    Combining the two inequalities, we obtain
    \begin{align}
        b(1-\ve) |\nu| -b(1+\ve) |\noisy{\sigma}| &\le |\nu| -|\nu'|, \quad b = \frac{1}{(1+2\chi B)(1+\bar{\fw} \|\phi\|_{\rm r})}\\
        |\nu'| &\le r|\nu| +b(1+\ve)|\noisy{\sigma}|, \quad  r=1-b(1-\ve)
    \end{align}
    Moreover, if $\nu(\noisy{\sigma})$ is the output of Algorithm \ref{alg:parallel-local-repair} after $t$ rainbow sweeps
    After applying $t$ color sweeps, we see that
    \begin{equation}
        |\nu(\noisy{\sigma})| \le r^t |\phi\bar{\delta}\noisy{\sigma}| + \frac{1+\ve}{1-\ve} |\noisy{\sigma}|
    \end{equation}
    Hence, the statement follows.
\end{proof}
\begin{remark}[Global Half Fault]
    Note that Lemma \ref{lem:rainbow-sweep} shows that a local repair with finite color sweeps exhibits a property similar to the near half fault property in Definition \ref{def:repair}. 
    Unfortunately, the inequality is global, i.e., does not apply to each connected component (relative to meta-checks).
    The remedy is to note that since only finite color sweeps are applied, the interaction \textit{range} can be bounded, and thus if the connectivity of the syndrome graph is \textit{enlarged}, repair commutes with the \textit{enlarged} connectivity so that the inequality holds for \textit{enlarged} connected components, as follows.
\end{remark}

\begin{lemma}[Enlarged Connectivity]
    \label{lem:enlarged-connectivity}
    Let $G(\phi,\bar{\delta})$ denote the graph\footnote{The syndrome graph $G(\bar{\delta})$ is a subgraph of $G(\phi,\bar{\delta})$} consisting of vertices $\cZ(\Theta)$ and edges $zz'$ if $z',z$ is in a common meta-check or if $z'\in \phi\bar{\delta}z$ or vice-versa $z\lr z'$.
    Let $\noisy{\sigma}$ be an $i$-level syndrome and $N_{t\chi B}(\noisy{\sigma})$ is its $t\chi B$-radius neighborhood relative to $G(\phi,\bar{\delta})$.
    If $\cK$ is a connected component of $N_{t\chi B}(\noisy{\sigma})$ relative to $G(\phi,\bar{\delta})$, then
    \begin{equation}
        \nu (\noisy{\sigma}|_{\cK}) = \nu (\noisy{\sigma} )|_{\cK}, \quad \rep (\noisy{\sigma}|_{\cK}) = (\rep \noisy{\sigma})|_{\cK}
    \end{equation}
    where $\nu$ is defined by Algorithm \ref{alg:parallel-local-repair}.
    Consequently,
    \begin{equation}
        |\rep \noisy{\sigma}|_{\cK} \le \left(\frac{2}{1-\ve} + \ve'\right)|\noisy{\sigma}|_{\cK}
    \end{equation}
\end{lemma}

\begin{proof}
    By Algorithm \ref{alg:parallel-local-repair}, we start with $\nu =\phi \bar{\delta} \noisy{\sigma}$ so that $\nu_0 \subseteq N_1(\noisy{\sigma})$.
    Note that after a single color sweep $\nu \gets \nu +w$ where $w$ must be the disjoint union of admissible local steps. 
    In particular, each admissible local step must intersect $\nu$ and thus $\nu +w \subseteq N_{t\chi B}(\noisy{\sigma})$ where we used the fact that each local step is of size $\le B$.
    Hence, after $t$ rainbow sweeps, we see that the output $\nu(\noisy{\sigma}) \subseteq N_{t\chi B}(\noisy{\sigma})$.
    The statement then follows.
\end{proof}

\begin{remark}[Enlarged RG Subgraph]
    \label{rem:enlarged-RG-subgraph}
    To accommodate the enlarged connectivity, we must replace the RG syndrome subgraph $\cG(\noisy{\bm{\sigma}})$ with
    \begin{equation}
        \cG(\noisy{\bm{\sigma}}) = \bigcup_{i\in [\ell]} \sh N_{t\chi B} (\noisy{\sigma}\super{i}) \cup \sh ( \rep\noisy{\sigma}\super{i} \cup f\super{i} \cup \delta f\super{i})
    \end{equation}
    The proof then follows similarly, culminating in Theorem \ref{thm:RG-threshold}, with noisy contraction
    \begin{equation}
        \noisy{\kappa} = \kappa \left(\frac{2}{1-\ve} + \ve'\right) <1
    \end{equation}
    where $\ve,\ve'$ are sufficiently small.
    % In addition, though we do not use the fact, Lemma \ref{lem:enlarged-connectivity} implies that if $zz'$ is an edge in $G(\phi,\bar{\delta})$, then $\sh z,\sh z'$ are adjacent in the RG graph (similar to Remark \ref{rem:meta-check-vs-RG-connectivity}) and that 
    % \begin{equation}
    %     \dnum^{\cK} \rep \noisy{\sigma}\super{i} = \rep (\dnum^{\cK} \noisy{\sigma}\super{i})
    % \end{equation}
    % for any connected component $\cK$ of the enlarged $\cG(\noisy{\bm{\sigma}})$.
\end{remark}

\begin{theorem}[Parallelized RG Decoder]
    \label{thm:RG-run-time}
    The RG decoder with the parallel local syndrome repair in Algorithm \ref{alg:parallel-local-repair} is single-shot as in Remark \ref{rem:enlarged-RG-subgraph} and runs in time 
    \begin{equation}
        \tau_{\rm RG} = O(\ell) + \tau_{\rm seed}
    \end{equation}
    where $\tau_{\rm seed}$ is the run time of the seed decoder.
    In particular, $\ell = \log (|\cQ\super{\ell}|/|\cQ\super{0}|)$.
\end{theorem}

\begin{proof}
    By Algorithm \ref{alg:parallel-local-repair}, the parallel local repair runs in $O(1)$ time for each level. Since the cleaning is layer-local and layers are uniformly bounded in size, it can be parallelized and run in $O(1)$ time. Hence, the $\desc \rep$ map runs in $O(1)$ time for each level and thus the statement follows.
\end{proof}

%% file: iterate.tex
\section{Level Map}
\label{sec:iterate}

Suppose that we are given a CSS code as a cocomplex $C=(C,\delta)$ where
\begin{equation}
    C =  X \to Q \to Z
\end{equation}
and $\bar{C} =(\bar{C},\bar{\delta})$ be that with meta-checks.
This subsection then inductively constructs CSS code $\bar{\Gamma}(C)$ with additional meta-checks only in the $Z$ sector, i.e., $MZ'$.
To construct meta-checks in the $X$ sector, a \textit{dual} iteration is applied via $\bar{\Gamma}_*(C)\equiv \bar{\Gamma}(C^\top)^\top$ so that the level map is given by the composition $\bar{\Theta} = \bar{\Gamma}_*\circ \bar{\Gamma}$ as in Lemma \ref{lem:level-map-composition}.
The following are some assumptions made to simplify the (already tedious) notation, though we remark that they can be removed and dealt with in full generality \cite{balasubramanian2026passive}.

\begin{hypothesis}[Edge-like $x,z$ Pair]
    \label{hyp:edge-like-supp}
    In general, the common support $x\wedge z$ has even cardinality if $x\prec z$. For (notation) simplicity, we shall assume that $|x\wedge z|=2$ whenever $x\prec z$ in $C$.
    Note that the assumption is self-dual, i.e., if $C$ satisfies this property, then so does it transpose $C^\top$.
\end{hypothesis}

\begin{hypothesis}[Connectivity]
    \label{hyp:connected-local-graph}
    Assume that the local graph $G^x$ as in Definition \ref{def:local-cone-x} is connected for all $x$-checks in $C$.
    Note that this assumption is not self-dual, and thus we shall also assume that the transpose $C^\top$ also satisfies, i.e., the local graphs of $z$-checks consisting of vertices $\{\prec z\}$ is connected for all $z$ checks
\end{hypothesis}

% \begin{hypothesis}[$\infty$ Syndrome Distance]
%     \label{hyp:seed-decoder}
%     For simplicity, we shall assume that the seed code $\bar{C}\super{0}$ has $\infty$ syndrome distance for both the $X$ and $Z$ sectors, i.e., all meta-check clean syndrome are valid syndrome. The case where the syndrome distance is simply large can also be dealt with.
%     Note that this property is self-dual and only assumed at the seed level.
% \end{hypothesis}

In particular, we shall show in Lemma \ref{lem:edge-like-support}-\ref{lem:connected-local-graph} that Hypotheses \ref{hyp:edge-like-supp}-\ref{hyp:connected-local-graph} hold inductively, i.e., $\Gamma(C),\Gamma_*(C)$ satisfy the hypotheses if $C$ satisfies them, so that it's sufficient to assume that only the seed code $C\super{0}$ satisfies the hypotheses.

\subsection{Local Graphs and Cones}

\begin{figure}[h]
\centering
\subfloat[\label{fig:input-code}]{%
    \centering
    \includegraphics[width=0.43\columnwidth]{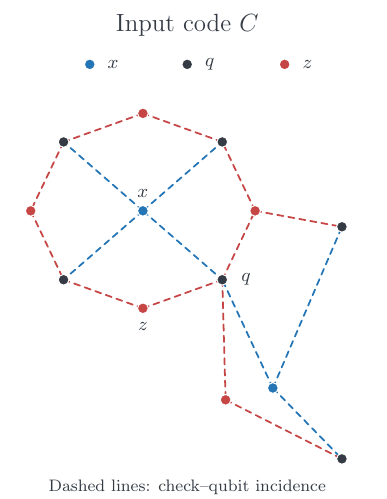}
}
\hspace{0.05\columnwidth}
\subfloat[\label{fig:local-cone}]{%
    \centering
    \includegraphics[width=0.43\columnwidth]{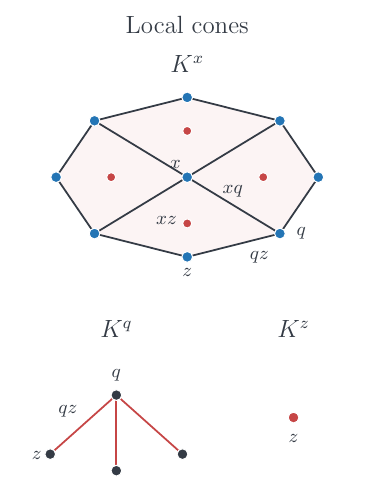}
}
\caption{Input Code and Local Cones. (a) A schematic portion of an input CSS code $C$: blue, black and red dots denote $x$-checks, qubits, and $z$-checks, respectively. Dashed lines denote check--qubit incidence; only selected cells and incidences are shown. (b) The local cones $K^x,K^q,K^z$ in Definitions \ref{def:local-cone-x} and \ref{def:local-cone-q-z}. In $K^x$, blue dots denote vertices $x,q,z$, black lines denote edges $xq,qz$, and red dots mark the shaded faces $xz$. In $K^q$, black dots denote vertices $q,z$ and red lines denote edges $qz$; $K^z$ consists of a single red vertex $z$. Only representative cell labels are shown.
When \textit{glued} together as in Remark \ref{rem:BPT}, we obtain a CSS code with $x$-checks, qubits and $z$-checks, labeled as blue, black and red, respectively.
}
\label{fig:local-view}
\end{figure}

\begin{definition}[Local Cone of $x$, Fig. \ref{fig:local-view}]
    \label{def:local-cone-x}
    Consider an $x$ check in $C$. Let the \textbf{local graph} $G^x=V^x\to E^x$ of $x$ be the (bipartite) graph cocomplex with vertices $\{\succ x\}$, i.e., all cells in $C$ strictly larger than $x$, and edges $qz$ if $q\prec z$. 
    Note that the vertices have degrees
    \begin{equation}
        |\delta_{G^x} q| = |\delta q|, \quad |\delta_{G^x} z| =2
    \end{equation}
    Write $\cone G^x$ as the corresponding cone cocomplex where we denote the extra \textit{cone} vertex as $x$ so that its vertices are all cells $\succeq x$, edges are $xq,qz,xz$ if the endpoints are adjacent, and faces $xqz$ is $x\prec q \prec z$.
    Slightly modify the cone and remove \textit{cone edges} $xz$. To keep the complex structure, for a given pair $x\prec z$, \textit{merge} faces $xqz$ and $xq'z$ (uniquely defined since $|x\wedge z|=2$) and call the corresponding face $xz$.
    We call the resulting cell cocomplex $K^x$ the \textbf{local cone} of $x$ with vertices $x,q,z$, edges $xq,qz$ and faces $xz$.
\end{definition}

\begin{remark}[Weights of $K^x$]
    \label{rem:codegree-Kx}
    Note that the column weights of (co)differential maps in $K^x$ can be easily derived from the following diagram
    \begin{equation}
        \label{eq:Kx-codeg-diag}
        \begin{tikzcd}
        	&& x \\
        	& q & xq \\
        	z & qz & xz
        	\arrow["{|\delta x|}", from=1-3, to=2-3]
        	\arrow["1", from=2-2, to=2-3]
        	\arrow["{|\delta q|}", from=2-2, to=3-2]
        	\arrow["{|\delta q|}", from=2-3, to=3-3]
        	\arrow["2"', from=3-1, to=3-2]
        	\arrow["1"', from=3-2, to=3-3]
        \end{tikzcd}
    \end{equation}
    where we have implicitly assumed that $q,z$ are over all cells $\succ x$.
    Specifically,
    \begin{align}
        |\delta_{K^x} x|=|\delta x|, \quad |\delta_{K^x} q|=1+|\delta q|, \quad |\delta_{K^x} z| =2 \\
        |\delta_{K^x} xq|=|\delta q|, \quad |\delta_{K^x} qz|=1
    \end{align}
    In particular, we had to remove the \textit{cone edges} $xz$ so that the degree of $x$ is \textit{preserved}. 
    The row weights can be simiarly derived, but are trivial since edges $xq,qz$ have two endpoints, and faces $xz$ are \textit{squares}, i.e., four edges.
\end{remark}
\begin{definition}[Local Cone of $q,z$, , Fig. \ref{fig:local-view}]
    \label{def:local-cone-q-z}
    Let the \textbf{local graph} $G^q=V^q$ of $q$ in $C$ be the graph cocomplex with vertices $\{\succ q\}$, i.e., all cells in $C$ strictly larger than $q$, and no edges. 
    Let $K^q=\cone G^q$ be the \textbf{local cone} of $q$ where we denote the extra \textit{cone vertex} as $q$ so that its vertices are all cells $\succeq q$ and edges are $qz$ if endpoints are adjacent.
    Further let the \textbf{local graph} $G^z=\varnothing$ of $z$ be the empty graph cocomplex.
    Let $K^z=\cone G^z$ be the \textbf{local cone} of $z$ where we denote the extra \textit{cone vertex} as $z$ so that its vertex is just $z$ itself.
\end{definition}

\begin{remark}[Weights of $K^q$]
    Similar to Remark \ref{rem:codegree-Kx}, note that the weights of $K^q$ are given by
    \begin{equation}
        \label{eq:Kq-codeg-diag}
        \begin{tikzcd}
        	& q \\
        	z & qz
        	\arrow["{|\delta q|}", from=1-2, to=2-2]
        	\arrow["1"', from=2-1, to=2-2]
        \end{tikzcd}
    \end{equation}
\end{remark}

\begin{remark}[Distinction]
    When we wish to distinguish cells in the local cones, we may write $|c;u\ket$ to denote the cell $c$ in $K^u$ for $u=x,q,z$.
\end{remark}
\begin{remark}[Local Cones are Contractible]
    \label{rem:cone-contractible}
    By construction, the local cones $K^x,K^q,K^z$ are all contractible, i.e., the cohomology vanishes except for the vertices, which have unique basis element, i.e., the sum of all vertices in $K^u$ for $u=x,q,z$.
\end{remark}

\begin{remark}[BPT Saturating]
    \label{rem:BPT}
    As first shown in \cite{lin2023geometrically} and elucidated in \cite{yuan2026unified}, if the purpose was to obtain BPT saturating (up to polylog) codes $\Gamma(C)$ from an input code $C$ by a single step of embedding, one can construct $\Gamma(C)$ via gluing (subdivisions of) $K^x,K^q,K^z$ together since they induce a natural inclusion chain map, due to the diagrams in Eq. \eqref{eq:Kx-codeg-diag}, \eqref{eq:Kq-codeg-diag} and Fig. \ref{fig:local-view}, and sketched below
    \begin{equation}
        \label{eq:K-glue-diag}
        \begin{tikzpicture}[baseline]
        \matrix(a)[matrix of math nodes, nodes in empty cells, nodes={minimum size=5mm},
        row sep=.4em, column sep=.5em,
        text height=1.5ex, text depth=0.25ex]
        {
            &&&&&& x \\
            &&& q && xq &\\
            z && qz && xz &&\\
            &&&q &&& \\
            z && qz &&&& \\
            &&&&&& \\
            z &&&&&& \\
        };
        \path[->,line width=.5mm,violet,font=\scriptsize]
        (a-3-1) edge (a-5-1)
        (a-3-3) edge (a-5-3)
        (a-5-1) edge (a-7-1);
        \path[->,font=\scriptsize]
        (a-2-4) edge (a-2-6)
        (a-3-1) edge (a-3-3)    
        (a-3-3) edge (a-3-5)
        (a-5-1) edge (a-5-3);
        \path[- ,line width=.5mm,violet,font=\scriptsize](a-2-4) edge ([yshift=4pt]a-3-4.center);
        \path[->,line width=.5mm,violet,font=\scriptsize]([yshift=-4pt]a-3-4.center) edge (a-4-4);
        \path[->,font=\scriptsize]
        ([xshift=1pt,yshift=2pt]a-2-4.south west) edge (a-3-3)
        ([xshift=1pt,yshift=2pt]a-2-6.south west) edge (a-3-5)
        ([xshift=1pt,yshift=2pt]a-1-7.south west) edge (a-2-6)
        ([xshift=1pt,yshift=2pt]a-4-4.south west) edge (a-5-3);
        \end{tikzpicture}  
        % \begin{tikzcd}[row sep=small, column sep=scriptsize]
        % 	&&&&&& \textcolor{rgb,255:red,51;green,51;blue,255}{x} \\
        % 	&&& \textcolor{rgb,255:red,51;green,51;blue,255}{q} && \textcolor{rgb,255:red,51;green,51;blue,255}{xq} \\
        % 	\textcolor{rgb,255:red,51;green,51;blue,255}{z} && \textcolor{rgb,255:red,51;green,51;blue,255}{qz} & {} & \textcolor{rgb,255:red,51;green,51;blue,255}{xz} \\
        % 	&&& q \\
        % 	z && qz \\
        % 	\\
        % 	\textcolor{rgb,255:red,255;green,51;blue,61}{z}
        % 	\arrow[color={rgb,255:red,51;green,51;blue,255}, from=1-7, to=2-6]
        % 	\arrow[color={rgb,255:red,51;green,51;blue,255}, from=2-4, to=2-6]
        % 	\arrow[color={rgb,255:red,51;green,51;blue,255}, from=2-4, to=3-3]
        % 	\arrow[color={rgb,255:red,160;green,51;blue,255}, line width=1pt, no head, from=2-4, to=3-4]
        % 	\arrow[color={rgb,255:red,51;green,51;blue,255}, from=2-6, to=3-5]
        % 	\arrow[color={rgb,255:red,51;green,51;blue,255}, from=3-1, to=3-3]
        % 	\arrow[color={rgb,255:red,160;green,51;blue,255}, line width=1pt, from=3-1, to=5-1]
        % 	\arrow[color={rgb,255:red,51;green,51;blue,255}, from=3-3, to=3-5]
        % 	\arrow[color={rgb,255:red,160;green,51;blue,255}, line width=1pt, from=3-3, to=5-3]
        % 	\arrow[color={rgb,255:red,160;green,51;blue,255}, line width=1pt, from=3-4, to=4-4]
        % 	\arrow[from=4-4, to=5-3]
        % 	\arrow[from=5-1, to=5-3]
        % 	\arrow[color={rgb,255:red,160;green,51;blue,255}, line width=1pt, from=5-1, to=7-1]
        % \end{tikzcd}
    \end{equation}
    where purple arrows denote the gluing maps $g^{qx}:K^x \to K^q$ if $q\succ x$ and similarly $g^{zq}$. Note that $g^{zx}=0$.
\end{remark}

\subsection{Local Hollowed Tensor}

To obtain self-correction, \cite{balasubramanian2026passive} extended the idea of Remark \ref{rem:BPT} by first modifying the local cones so that they possess meta-checks before gluing them together as in Eq. \eqref{eq:K-glue-diag}.
The most natural way is to take the tensor product of the local cones $K^{u}$ with the repetition code $R$.
However, this is slightly insufficient, since this would increase the weights of certain cells by $+2$ due to \textit{vertical} extensions of the repetition code as shown in Fig. \ref{fig:hollow-tensor}.
Hence, an iterative embedding over $i$ from the seed input code $C\super{0}$ would result in the final code $C\super{i}$ to be non-LDPC for sufficiently large $i$.
Therefore, \cite{balasubramanian2026passive} circumvents this issue via \textit{hollowing out} the \textit{interior} of the tensor product so that only the \textit{boundary} remains.
Specifically,

\begin{figure}[ht]
\centering
\subfloat[\label{fig:hollow-tensor}]{%
    \centering
    \includegraphics[width=0.43\columnwidth]{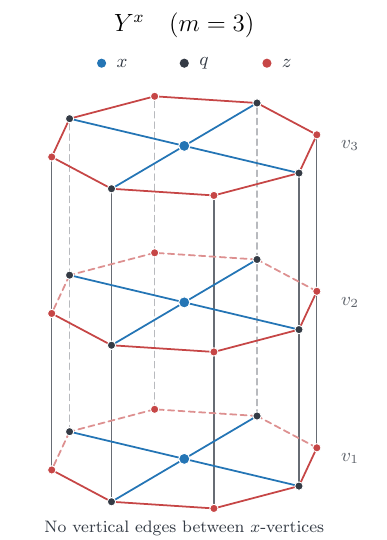}
}
\hspace{0.05\columnwidth}
\subfloat[\label{fig:meta-check}]{%
    \centering
    \includegraphics[width=0.43\columnwidth]{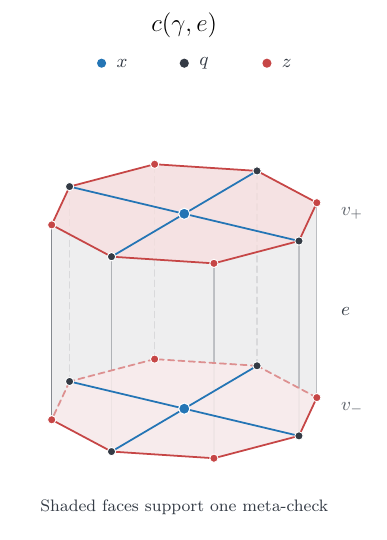}
}
\caption{Hollowed Tensor. (a) denotes $Y^x$ constructed from the local graph and cone of $x$ as in Fig. \ref{fig:local-view}. Note that the coloring of $Y^x$ follows Fig. \ref{fig:input-code} instead of Fig. \ref{fig:local-cone} to emphasize the origin of each vertex and edge. 
In particular, black edges denote vertical edges due to the repetition code, but note that there are no \textit{vertical edges} between (blue) $x$ vertices.
(b) shows the support of a meta-check $c(\gamma,e)$ in $\bar{Y}^x$, where $\gamma$ is the boundary cycle and $\partial^R e=v_-+v_+$. The red shaded caps and gray shaded side faces denote $\sum_{z\in\gamma}xz\otimes\partial^R e$ and $\sum_{qz\in\gamma}qz\otimes e$, respectively. Dashed lines indicate edges behind the displayed surfaces.
}
\label{fig:hollow-tensor-w-meta-checks}
\end{figure}

\begin{definition}[Hollowed Tensor, Fig. \ref{fig:hollow-tensor}]
    \label{def:hollowed-tensor}
    Let $G^u,K^u$ be the local graph and cone of cells $u=x,q,z$. Define the \textbf{hollowed tensor (product)} as the union of cocomplexes as follows
    \begin{equation}
        Y^u = K^u \otimes V + G^u\otimes R
    \end{equation}
    where $R=V\to E$ is the repetition code on $m\ge 2$ vertices, so that $Y^u$ are length 2, 1, 0 cocomplexes for $u=x,q,z$, respectively.
\end{definition}

Currently, the hollowed tensor does not yet have meta-checks; these will be added later. First, we check the homology of the cohollowed tensor. 
This can be verified directly via the Mayer--Vietoris sequence, but since the constructions are cell cocomplexes, it's more elucidating to the topological structure by relating the hollowed tensor to its homotopy equivalent, i.e., wedge sums of suspensions of local graphs.

\begin{lemma}[Cohomology]
    \label{lem:cohomology}
    Let $Y^u$ be the hollowed tensor for $u=x,q,z$. Then it is homotopy-equivalent to
    \begin{equation}
        Y^u \cong \bigvee^{m-1} \Sigma G^u
    \end{equation}
    where $\vee$ denotes the wedge sum of the suspension $\Sigma G^u$. In particular,
    \begin{equation}
        \tilde{H}^j(Y^u) \cong \bigoplus_{m-1} \tilde{H}^{j-1}(G^u)
    \end{equation}
    % If $G^x$ is connected, then $H^1(Y^x)=0$.
\end{lemma}
\begin{proof}
    For nonempty $G^u$, collapse one contractible copy $K^u\times\{v\}$ to a point. Since $R$ is a tree, each of the remaining $m-1$ cone copies contributes a suspension of $G^u$; the intervening cylinders give the attaching maps. Thus the quotient is homotopy-equivalent to the stated wedge. Reduced cohomology gives the second identity. For $u=z$, $Y^z$ is $m$ points: interpret $\Sigma\varnothing=S^0$ and $\tilde H^{-1}(\varnothing)=\F_2$.
\end{proof}

\begin{definition}[Meta-Checks, Fig. \ref{fig:hollow-tensor-w-meta-checks}]
    \label{def:hollowed-meta-checks}
    By Lemma \ref{lem:cohomology}, the cohomology of $Y^u$ can be reduced to that of the suspension of local graphs $G^u$, and thus it becomes clear how to introduce meta-checks in a natural manner.
    \begin{itemize}
        \item First consider $H^2(Y^x)\cong  H^1(G^x)^{m-1}$.
        Let $\gamma$ be a cycle of $G^x$ and edge $e\in R$. By bipartiteness, $\gamma$ must alternate between $q,z$ vertices and thus we can define
        \begin{equation}
            c(\gamma,e) = \sum_{z\in \gamma} xz\otimes \partial^R e + \sum_{qz\in \gamma} qz\otimes e
        \end{equation}
        Choose a cycle basis $\{\gamma\}$ of $H_1(G^x)$ and equip $Y^x$ with 3-cells $\{c(\gamma,e)\}$ so that $\bar{Y}^x$ is now a cell cocomplex with metachecks so that $H^2(\bar{Y}^x)=0$.
        \item Next consider $H^1(Y^q)\cong \tilde{H}^0(G^q)^{m-1}$. For each distinct $z,z'\succ q$ and edge $e\in R$, define
        \begin{equation}
            \gamma(zz',e) = \sum_{u=zz'} (qu \otimes \partial^R e + u\otimes e)
        \end{equation}
        Equip $Y^q$ with 2-cells\footnote{Technically, only a complete pair is necessary, e.g., arbitrarily order $z_0,z_1,...\succ q$ and only consider pairs $\gamma(z_0z_i,e),i\ge 1$ since any pair $\gamma(z_i z_j,e)$ can be obtained from $z_0 z_i, z_0z_j$.} $\{\gamma(zz',e)\}$ so that $\bar{Y}^x$ is now a cell cocomplex of length 2 with $H^1(Y^x)=0$.
        \item Finally, note that $Y^z=z\otimes V$ and thus naturally extend it to $\bar{Y}^z=z\otimes R\cong R$.
    \end{itemize}
\end{definition}

\begin{remark}[Hollowed Tensors are Contractible]
    \label{rem:hollow-contractible}
    In analogy to Remark \ref{rem:cone-contractible}, the meta-check equipped hollowed tensors $\bar{Y}^u,u=x,q,z$ are contractible.
    Note that since the graphs are all local, the equipped meta-checks are naturally local and sparse. See Lemma \ref{lem:sparsity}.
\end{remark}

\begin{remark}[Diagrams]
    Similar to Eq. \eqref{eq:Kx-codeg-diag}, \eqref{eq:Kq-codeg-diag}, the hollowed tensors have (commuting) diagrams
    \begin{equation}
        \label{eq:Y-diag}
        \bar{Y}^x = 
        \begin{tikzpicture}[baseline]
        \matrix(a)[matrix of math nodes, nodes in empty cells, nodes={minimum size=5mm},
        row sep=.4em, column sep=.5em,
        text height=1.5ex, text depth=0.25ex]
        {
            &&&&&& x,v \\
            &&& q,v && xq,v &\\
            z,v && qz,v && xz,v &&\\
            &&&q,e &&& \\
            z,e && qz,e && c(\gamma,e) && \\
        };
        \path[->,blue,font=\scriptsize]
        (a-3-1) edge (a-5-1)
        (a-3-3) edge (a-5-3);
        \path[->,red,font=\scriptsize]
        (a-3-5) edge (a-5-5)
        (a-5-3) edge (a-5-5);
        \path[->,font=\scriptsize]
        (a-2-4) edge (a-2-6)
        (a-3-1) edge (a-3-3)    
        (a-3-3) edge (a-3-5)
        (a-5-1) edge (a-5-3);
        \path[- ,blue,font=\scriptsize](a-2-4) edge ([yshift=4pt]a-3-4.center);
        \path[->,blue,font=\scriptsize]([yshift=-4pt]a-3-4.center) edge (a-4-4);
        \path[->,font=\scriptsize]
        ([xshift=1pt,yshift=2pt]a-2-4.south west) edge (a-3-3)
        ([xshift=1pt,yshift=2pt]a-2-6.south west) edge (a-3-5)
        ([xshift=1pt,yshift=2pt]a-1-7.south west) edge (a-2-6)
        ([xshift=1pt,yshift=2pt]a-4-4.south west) edge (a-5-3);
        \end{tikzpicture},
        \quad
        \bar{Y}^q = 
        \begin{tikzpicture}[baseline]
        \matrix(a)[matrix of math nodes, nodes in empty cells, nodes={minimum size=5mm},
        row sep=.4em, column sep=.5em,
        text height=1.5ex, text depth=0.25ex]
        {
            &&& q,v \\
            z,v && qz,v &\\
            &&& \\
            z,e && \gamma(zz',e) & \\
        };
        \path[->,font=\scriptsize]
        ([xshift=1pt,yshift=2pt]a-1-4.south west) edge (a-2-3)
        (a-2-1) edge (a-2-3);
        \path[->,blue,font=\scriptsize]
        (a-2-1) edge (a-4-1);
        \path[->,red,font=\scriptsize]
        (a-2-3) edge (a-4-3)
        (a-4-1) edge (a-4-3);
        \end{tikzpicture},
        \quad
        \bar{Y}^z = 
        \begin{tikzpicture}[baseline]
        \matrix(a)[matrix of math nodes, nodes in empty cells, nodes={minimum size=5mm},
        row sep=.4em, column sep=.5em,
        text height=1.5ex, text depth=0.25ex]
        {
            z,v \\
            \\
            z,e \\
        };
        \path[->,red,font=\scriptsize]
        (a-1-1) edge (a-3-1);
        \end{tikzpicture}
    \end{equation}
    where the black arrows denote adjacency due to local graphs or cones, blue arrows denote that due to the repetition code in Definition \ref{def:hollowed-tensor}, and red arrows due to the equipped meta-checks in Definition \ref{def:hollowed-meta-checks}.
\end{remark}

\subsection{Induction Step}
\begin{definition}[Induction]
    \label{def:gluing}
    Similar to Remark \ref{rem:BPT}, the hollowed tensors $\bar{Y}^x,\bar{Y}^q,\bar{Y}^z$ induce a natural chain map $g^{qx}:\bar{Y}^x\to \bar{Y}^q$ and $g^{zq}:\bar{Y}^q\to \bar{Y}^z$ if $x\prec q\prec z$, even when the meta-checks are included. The full diagram is omitted here due to its 4D-like structure; instead, we show the flattened diagram below.
    \begin{equation}
        \bar{\Gamma}(C) \equiv 
        X(\Gamma) \to Q(\Gamma) \to Z(\Gamma) \to MZ(\bar{\Gamma}) =
        \begin{tikzpicture}[baseline]
        \matrix(a)[matrix of math nodes, nodes in empty cells, nodes={minimum size=5mm},
        row sep=2em, column sep=2em,
        text height=1.5ex, text depth=0.25ex]
        {Y^x_0 & Y_1^x & Y_2^x & \bar{Y}_3^x\\
         Y^q_0 & Y_1^q & \bar{Y}_2^q & \\
         Y^z_0 & \bar{Y}_1^z &&  \\};
        \path[->,font=\scriptsize]
        (a-1-1) edge (a-1-2)
        (a-1-2) edge (a-1-3)
        (a-1-3) edge (a-1-4)
        (a-2-1) edge (a-2-2)
        (a-2-2) edge (a-2-3)
        (a-3-1) edge (a-3-2);
        \path[->,line width=.5mm,violet,font=\scriptsize]
        (a-1-1) edge (a-2-1)
        (a-1-2) edge (a-2-2)
        (a-1-3) edge (a-2-3)
        (a-2-1) edge (a-3-1)
        (a-2-2) edge (a-3-2);
        \end{tikzpicture}
    \end{equation}
    so that cells in $Y_0^x$ over all $x$ are the $X$-checks in $\Gamma$. 
    Similarly, those in $Y_1^x,Y_0^q$ are the qubits, those in $Y_2^x,Y_1^q,Y_0^z$ are the $Z$-checks, and those in $\bar{Y}_3^x,\bar{Y}_2^q,\bar{Y}_1^z$ are the meta $Z$-checks.
    In particular, each cell $\gamma$ in $\Gamma(C)$ must be in the layer $Y^c$ of some unqiue cell $c$ in $C$ and thus we call $c$ the \textbf{shadow} of $\gamma$, i.e., $c=\sh \gamma$.
    We also write $\Gamma=\Gamma(C)$ to be that without the meta-checks.
    Recall that the \textbf{dual} is $\bar{\Gamma}_*(C)=\bar{\Gamma}(C^{\top})^{\top}$.
\end{definition}

Unlike the BPT case in Remark \ref{rem:BPT}, the embedding $C\to \Gamma(C)$ will occur inductively over $i$, and thus we must make sure that the weights and qubit degrees remain uniformly bounded.
In particular, if the seed input code $C\super{0}$ is QLDPC, the final code must also be QLDPC (in fact, sparse with meta-checks), as guaranteed by the following.

\begin{lemma}[Sparsity]
    \label{lem:sparsity}
    Let $m=2$ ($m\ge 3$). Then $c=3 (4)$ and
    % \begin{equation}
    %     X\super{i} \xrightleftharpoons[\max(\fd_Z,2)]{\max( \weight_X,\fd_Z+3, 5)} Q \xrightleftharpoons[\max(\weight_Z,q_X+2,4)]{\max(\fd_Z,3)} Z\super{i}
    % \end{equation}
    \begin{align}
        \weight_X' &\le \max( \weight_X,\fd_Z+c, 2+c) \\
        \fd_Z' &\le \max(\fd_Z,c)\\
        \weight_Z' &\le \max(\weight_Z,\fd_X+2,4) \\
        \fd_X' &\le \max(\fd_X,2)
    \end{align}
    where the prime notation denotes the $X,Z$ check weights and qubit degrees of $\Gamma(C)$, and those without the prime notation are those of $C$.
    Choose a simple-cycle basis in each $G^x$ and put $r_* =\max_x\dim H_1(G^x)$.
    The meta-check weight $\bar{\fw}'$ and degree $\bar{\fd}'$ of $\bar\Gamma(C)$ then satisfy
    \begin{align}
        \bar\fw'&\le\max (4\weight_X,6+2\fd_X,2+\weight_Z ),\\
        \bar\fd'&\le 2\max (r_*,\fd_Z-1,1),\qquad
        r_*\le\frac{\weight_X\fd_Z}{2}.
    \end{align}
    % In particular, if 
    % \begin{align}
    %     \weight_{\alpha}'' &\le \max(\weight_{\alpha},\fd_{\bar{\alpha}} +c, 2+c) \\
    %     \fd_{\alpha}'' &\le \max(\fd_{\alpha},c)
    % \end{align}
    % where double prime notation denotes those of $(\Gamma_*\circ \Gamma)(C)$. 
    % Hence, if $C\super{0}$ is QLDPC, then so is $(\Gamma_*\circ \Gamma)^\ell(C\super{0})$, independent of $\ell$.
\end{lemma}
\begin{proof}
    Put $r_v=\deg_R(v)\le c-2$. The $X$-check weights for vertices $x\otimes v,q\otimes v,z\otimes v$ in an $x$ layer are respectively $|\delta x|$, $|\delta q|+2+r_v$, and $4+r_v$. Qubits in $x$ layers have $X$-degree two and $Z$-degree at most $\max(\fd_Z,2+\max_v r_v,4)$, where the value four occurs only for a rail $z\otimes e$; for $m=2$ this rail still has degree four. Qubits in $q$ layers have $X$-degree at most $\fd_X$ and $Z$-degree at most $\max(\fd_Z,2+\max_v r_v)$. Finally, $Z$-checks in $x,q,z$ layers have weights at most $4,\fd_X+2,\weight_Z$, respectively. These give the displayed bounds with $\fd_Z'\le\max(\fd_Z,4)$ for both cases. 
    % Under alternating primal and dual steps, both qubit degrees are bounded by $D=\max(\fd_X^{(0)},\fd_Z^{(0)},4)$, and both check weights by $\max(\weight_X^{(0)},\weight_Z^{(0)},D+4,6)$. Hence the bounds are uniform in the number of levels.

    For the meta-check weights, $c(\gamma,e)$ has $2|E(\gamma)|\le4\weight_X$ boundary cells. A $q$-layer meta-check has six local boundary cells and two incoming faces for each $x\prec q$, giving weight at most $6+2\fd_X$. A $z$-layer meta-check has two local vertices and one incoming rail for each $q\prec z$, giving weight at most $2+\weight_Z$.
    For the meta-check degrees, the check types $xz\otimes v,qz\otimes e$ in $x$ layers, $qz\otimes v,z\otimes e$ in $q$ layers, and $z\otimes v$ in $z$ layers have degrees at most
    \begin{equation}
     2r_*,\quad r_*+\fd_Z-1,\quad2(\fd_Z-1),\quad\fd_Z,\quad2,
    \end{equation}
    respectively. Indeed, each local edge or vertex lies in at most $r_*$ basis cycles; an anchor $z_0$ meets $|\delta q|-1$ pair meta-checks, and a non-anchor meets one. The repetition path has degree at most two. Taking the maximum proves the degree bound. Finally, connectivity and $|x\wedge z|=2$ give
    \begin{equation}
     \dim H_1(G^x)=|\{z:x\prec z\}|-|\delta x|+1
     \le |\{z:x\prec z\}|\le\frac{\weight_X\fd_Z}{2}.
    \end{equation}
    Thus the new meta-checks are uniformly sparse under iteration as well.
\end{proof}

\begin{lemma}[Isomorphism]
    \label{lem:isomorphism}
    Let $\bar{\Gamma}=\bar{\Gamma}(C)$ be constructed as in Definition \ref{def:gluing} which contains meta-checks. Then
    \begin{equation}
        H_Q(\bar{\Gamma}) \cong H_Q(C), \quad H_Z(\bar{\Gamma}) \cong H_Z(C) = \coker \delta, \quad H_X(\bar{\Gamma}) \cong H_X(C) =\ker \delta
    \end{equation}
\end{lemma}

\begin{remark}[Renormalization]
    Note that in general, $H_Z(\bar{\Gamma})\not\cong H_Z(\bar{C})$ since meta-checks in $\bar{C}$ are discarded during the induction step.
    However, this does not obstruct our decoder, since the main point is that given measurement error $\mu\super{i}$, we apply a correction so that the repaired syndrome is meta-check clean and thus can be mapped to $[\mu]\in \coker \partial$ with representation $\mu$ chosen to satisfy certain properties. 
    This provides the renormalization step $\mu\super{i}\mapsto \mu$.
    See Definition \ref{def:RG-decoder} and Lemma \ref{lem:RG-noisy-decoder}.
\end{remark}
\begin{proof}
    The statement follows from the main result in \cite{yuan2026unified} and Remark \ref{rem:hollow-contractible}.
\end{proof}

\begin{lemma}[SDR]
    \label{lem:SDR}
    Let $\bar{\Gamma}=\bar{\Gamma}(C)$, and consider the diagram below
    \begin{equation}
    \begin{tikzpicture}[baseline]
    \matrix(a)[matrix of math nodes, nodes in empty cells, nodes={minimum size=25pt},
    row sep=2em, column sep=3em,
    text height=1.25ex, text depth=0.25ex]
    {X& Q  & Z & \\
     X(\Gamma) & Q(\Gamma)  & Z(\Gamma) & MZ(\bar{\Gamma})\\};
    \path[-left to,font=\scriptsize,transform canvas={yshift=0.2ex}]
    (a-1-1) edge (a-1-2)
    (a-1-2) edge (a-1-3)
    (a-2-1) edge (a-2-2)
    (a-2-2) edge (a-2-3)
    (a-2-3) edge (a-2-4);
    \path[left to-,font=\scriptsize,transform canvas={yshift=-0.2ex}]
    (a-2-1) edge node[below]{$\phi$}  (a-2-2)
    (a-2-2) edge node[below]{$\phi$}  (a-2-3)
    (a-2-3) edge node[below]{$\phi$}  (a-2-4);
    \path[-left to,font=\scriptsize,transform canvas={xshift=0.2ex}]
    (a-1-1) edge node[right]{$\iota$}  (a-2-1)
    (a-1-2) edge node[right]{$\iota$}  (a-2-2)
    (a-1-3) edge node[right]{$\iota$}  (a-2-3);
    \path[left to-,font=\scriptsize,transform canvas={xshift=-0.2ex}]
    (a-1-1) edge node[left]{$\pi$}  (a-2-1)
    (a-1-2) edge node[left]{$\pi$}  (a-2-2)
    (a-1-3) edge node[left]{$\pi$}  (a-2-3);
    \end{tikzpicture}
    \end{equation}
    Let $\iota:C\to \Gamma$ denote the inclusion map defined so that $\iota u$ is the collection of all vertices in $Y^u$ for any cell $u=x,q,z$.
    Then there exists layer-local $\pi,\phi$ so that $\iota,\pi$ form a SDR (see Definition \ref{def:level-map}) between $C,\bar{\Gamma}$.
\end{lemma}

\begin{proof}
    Choose within each layer a contraction $(i_0,p_0,h_0)$ onto its constant 0-cochain, with $h_0i_0=p_0h_0=h_0^2=0$. Here $p_0$ on vertices may be evaluation at a chosen reference vertex. The gluing $g$ strictly raises the shadow, so $(gh_0)^3=0$. The finite perturbation formulas give
    \begin{equation}
     \iota=i_0,\qquad \pi=p_0(1+gh_0)^{-1},\qquad
     \phi=h_0(1+gh_0)^{-1}.
    \end{equation}
    Indeed, $gi_0=i_0\delta_C$, so the transferred differential is $\delta_C$ and the inclusion is unchanged. The contraction identities give the SDR identities in Definition~\ref{def:level-map}. All factors preserve or raise the shadow, proving layer-locality; bounded layers and sparsity give uniform row and column bounds. 
\end{proof}

\begin{lemma}[Edge-like Support]
    \label{lem:edge-like-support}
    If Hypothesis \ref{hyp:edge-like-supp} holds for $C$, then it also holds for $\Gamma(C)$ (and thus $\Gamma_*(C)$)
\end{lemma}

\begin{proof}
    Note that by Definition \ref{def:gluing}, $X$-checks in $\Gamma(C)$ are vertices in $Y^x$, while $Z$-checks are faces in $Y^x$, edges in $Y^q$ and vertices in $Y^z$.
    Since every face in $Y^x$ is a square, it's clear that given a pair of $X,Z$-check in $Y^x$, the hypothesis is satisfied for $\Gamma(C)$.
    The other cases can be similarly checked.
\end{proof}

\begin{lemma}[Connected Local Graph]
    \label{lem:connected-local-graph}
    If Hypothesis \ref{hyp:connected-local-graph} holds for $C$, then it also holds for $\Gamma(C)$ (and thus $\Gamma_*(C)$)
\end{lemma}

\begin{proof}
    For a new $X$-check $x\otimes v$, its local graph is a subdivision of $G^x$. For $q\otimes v$ in an $x$ layer, the cone edge $xq\otimes v$ connects through the faces $xz\otimes v$ to all hinges $qz\otimes v$; the adjacent rails and the copied $q$-layer vertex connect to these hinges. This uses a neighbor $z\succ q$. For $z\otimes v$ in an $x$ layer, the two hinges are joined by $xz\otimes v$, and all rails and copied vertices attach to them.

    For the transpose, a new $Z$-check in an $x$ layer is a square, so its local graph is connected. A $Z$-check in a $q$ layer is an edge; for each $x\prec q$, its copied edge and the two copied endpoint checks join its two endpoints, so the local graph is connected. Finally, the local graph of $z\otimes v$ in a $z$ layer is a subdivision of the original local graph at $z$. All new checks are nonempty and each new qubit meets both sectors. Hence both incidence and connectivity conditions persist, including after transposition.
\end{proof}

%% file: contraction.tex
\section{Cleaning}
\label{sec:RG-reduction}

As suggested by Definition \ref{def:level-map}, we must consider the $Z$-sector cleaning $f$ of $\Gamma$ and compute the contraction $\kappa$ and fill bound $c_f$.
% Since we will be considering the alternating construction $(\Gamma_*\circ \Gamma)\super{\ell}(C\super{0})$, we note that the treatment of $X$ and $Z$ syndromes are similar, and thus for simplicity, we shall focus on the $Z$ syndromes. 
In \cite{balasubramanian2026passive}, the authors showed that if $m=2$, then for valid syndromes $\sigma\in Z(\Gamma)$, a layer-local cleaning procedure $f$ can be constructed to obtain the coarse-grained valid syndromes $\desc \sigma\in Z(C)$ such that
\begin{equation}
    |\desc \sigma|\le \kappa |\sigma|,\quad  \kappa =\frac{1}{2}
\end{equation}
Since the aim of \cite{balasubramanian2026passive} was to prove self-correction, which only deals with physical errors, $m=2$ and $\kappa <1$ are sufficient to derive an RG argument.
However, as suggested by the noisy contraction in Remark \ref{rem:setup}, when dealing with both physical and measurement errors, we further require the reduction factor to be $\kappa<1/2$.
By Remark \ref{rem:m-sheets}, one way of accomplishing this is to use the conventional $m=2$ sheet construction and group multiple (two) levels together as a single pass.
However, since it is more natural to have a single construction, we shall extend the argument to $m=3$ and show that $\kappa <1/2$.
Note that some properties extend easily to general $m\ge 2$, while others are nontrivial, and thus we will emphasize the distinction.
Motivated by Definition \ref{def:gluing} and the Cleaning Lemma \cite{yuan2026unified}, we define
\begin{definition}
    Let $\sigma\ne 0$ be a syndrome in $Z(\Gamma)$ and a coboundary of $Y^v$ for $v=x,q,z$. Then the \textbf{relative (syndrome) coexpansion} for $\sigma$ is defined as
    \begin{equation}
        \|\sigma\|^{UV} = \min_{\delta^V f=\sigma} \frac{|g^{UV} f|}{|\delta^V f|}
    \end{equation}
    Further define the \textbf{relative (syndrome) coexpansion} $\zeta^{UV}$ as the maximum of $\|\sigma\|^{UV}$ over syndromes $\sigma$ which are also coboundaries  of $Y^v$.
\end{definition}

\subsection{$x$ Layer Cleaning}

\begin{figure}[ht]
\centering
\subfloat[\label{fig:hollow-tensor-copy}]{\includegraphics[height=0.40\columnwidth]{hollow-tensor-clean.pdf}}
\hspace{0.015\columnwidth}
\subfloat[\label{fig:Yx-q-block}]{\includegraphics[height=0.40\columnwidth]{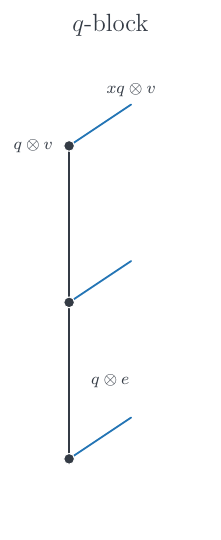}}
\hspace{0.015\columnwidth}
\subfloat[\label{fig:Yx-z-block}]{\includegraphics[height=0.40\columnwidth]{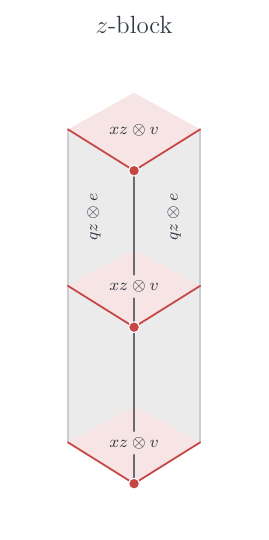}}
\hspace{0.015\columnwidth}
\subfloat[\label{fig:Yx-ze-block}]{\includegraphics[height=0.40\columnwidth]{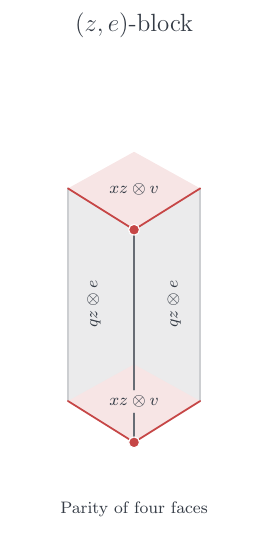}}
\caption{$Y^x$ blocks. (a) An $x$-layer $Y^x$ as in Fig.~\ref{fig:hollow-tensor}. (b) A $q$-block, with vertices $q\otimes v$ and edges $xq\otimes v$, $q\otimes e$. (c) A $z$-block: red shaded faces denote $xz\otimes v$ and gray shaded faces denote $qz\otimes e$. (d) A $(z,e)$-block for one repetition edge $e$; its parity is the $\F_2$ sum of the occupied faces among the two caps and two side faces. Open-ended edges in (b) meet $x$-vertices outside the block.}
\label{fig:Yx-blocks}
\end{figure}

In this subsubsection, we focus on the $x$-layer $Y^x= (Y^x,\delta^x)$ depicted by the diagram in Eq. \eqref{eq:Y-diag}.

\begin{remark}[$Y^x$ Geometric Reduction]
    \label{rem:x-geometric-reduction}
    Note that there exists a natural geometric reduction of $x$ layers as shown in Fig. \ref{fig:Yx-blocks}.
    From an algebraic perspective, this corresponds to flattening the diagram as follows
    \begin{equation}
    \begin{tikzpicture}[baseline]
    \matrix(a)[matrix of math nodes, nodes in empty cells, nodes={minimum size=5mm},
    row sep=1.5em, column sep=1.5em,
    text height=1.5ex, text depth=0.25ex]
    {&& x,v\\
    & q,v  &  xq,v \oplus q,e\\
    z,v  & qz,v\oplus z,e & xz,v\oplus qz,e\\};
    \path[->,font=\scriptsize]
    (a-2-2) edge  (a-2-3)
    (a-3-1) edge  (a-3-2)
    (a-3-2) edge  (a-3-3)
    (a-1-3) edge  (a-2-3)
    (a-2-2) edge  (a-3-2)
    (a-2-3) edge  (a-3-3);
    \end{tikzpicture}
    \end{equation}
    where the row cocomplexes are shown in in Fig. \ref{fig:Yx-q-block} and \ref{fig:Yx-z-block} and, in particular, given by the direct sum of $q$ \textbf{blocks} and $z$ \textbf{blocks}, respectively.
    Using either Mayer-Vietoris\footnote{Note that when using Mayer-Vietoris to compute the cohomology of the bottom row complex corresponding to Fig. \ref{fig:Yx-z-block}, one must take the transpose, since by definition, the intersection must be a well-defined complex.} or the Snake-Lemma, it's straightforward to check that the middle row complex has cohomology $(0,q\otimes E)$, while the bottom row complex has homology $(0,0,z\otimes E)$.
    Therefore, one can apply \cite{yuan2026unified} so that $Y^x$ (without meta-check 3-cells) is homotopically equivalent to the column cocomplex as shown below 
    \begin{equation}
    \begin{tikzpicture}[baseline]
    \matrix(a)[matrix of math nodes, nodes in empty cells, nodes={minimum size=25pt},
    row sep=1.5em, column sep=1.5em,
    text height=1.25ex, text depth=0.25ex]
    {Y_0^x & Y_1^x  & Y_2^x  \\
     x\otimes V& q\otimes E  & z\otimes E \\};
    \path[->,font=\scriptsize,transform canvas={yshift=0.2ex}]
    (a-1-1) edge (a-1-2)
    (a-1-2) edge (a-1-3)
    (a-2-1) edge (a-2-2)
    (a-2-2) edge (a-2-3);
    \path[-left to,font=\scriptsize,transform canvas={xshift=0.2ex}]
    (a-1-1) edge node[right]{$\iota$}  (a-2-1)
    (a-1-2) edge node[right]{$\iota$}  (a-2-2)
    (a-1-3) edge node[right]{$\iota$}  (a-2-3);
    \path[left to-,font=\scriptsize,transform canvas={xshift=-0.2ex}]
    (a-1-1) edge node[left]{$\pi$}  (a-2-1)
    (a-1-2) edge node[left]{$\pi$}  (a-2-2)
    (a-1-3) edge node[left]{$\pi$}  (a-2-3);
    \end{tikzpicture}
    \end{equation}
    where $x\otimes v \prec q\otimes e$ if $x\prec q$ in $C$ and $v\prec e$ in repetition $R$, while $q\otimes e \prec z\otimes e$ if $q\prec z$, and $\pi,\iota$ denote the projection and inclusion chain maps.
    From a geometric standpoint, the projection $\pi:Y_2^x \to z\otimes E$ computes the \textbf{flux} (parity) on the $(z,e)$ \textbf{block} as shown in Fig. \ref{fig:Yx-ze-block}, while $\pi:Y_1^x \to q\otimes E$ computes the \textbf{flux} (parity) on each $(q,e)$ \textbf{block}.
\end{remark}

\begin{lemma}[Flux Cleaning]
    \label{lem:flux-cleaning}
    Let $\sigma$ be a 2-coboundary in $Y^x$. For any subset $V'$ of $m-1$ vertices of $V$, there exists $f$ supported on \textit{cone lines} $xq\otimes v$ where $v\in V'$ such that $\sigma +\delta^x f$ has trivial parity in each $(z,e)$ block of $Y^x$; equivalently, the restriction of $\sigma +\delta ^x f$ is a 2-coboundary for each $z$ block.
\end{lemma}

\begin{proof}
    By geometric reduction in Remark \ref{rem:x-geometric-reduction}, if $\sigma\in Y^x$ is a 2-coboundary, i.e., $\in \im \delta^x$, then $\pi \sigma$ must be a 2-coboundary of the embedded column cocomplex.
    In particular, $(\pi \sigma)_e$ can be induced by some $f_e \subseteq \{q:q\succ x\}\times e$ for every edge $e$ of $R$.
    Note that $f_e$ is naturally corresponds to of a subset of 1-cells $\{q\otimes e:q\succ e\}$ in $Y^x$, and thus with slight abuse of notation, we also write $f_e\subseteq \{q\otimes e:q\succ e\}$ so that $(\pi \sigma)_e =\pi f_e$.
    In particular $\pi(\sigma+f_E)=0$ where $f_E$ is the (direct) sum of $f_e,e\in E$.
    
    Note that $f_e$ is supported on the \textit{rails} $q\otimes e$.
    For our purposes, we would instead prefer a solution concentrated on the \textit{cone lines} $xq\otimes v$.
    To do this, note that for any fixed $q$, the collection $f_E(q)\equiv \{f_e(q):e\}$ can be regarded a subset of edges of $R$ and since $H_1(R)=0$, there must exist subset $f_V(q)\equiv \{f_v(q):v\}$ of vertices such that $\delta^R f_V(q)= f_E(q)$.
    In fact, we can choose $f_V(q)$ to be supported on any $m-1$ vertices uniformly over $q$.
    With slight abuse of notation, write $f_V$ as the sum of cone lines $xq\otimes v$ if $f_v(q)\ne 0$, and thus it's straightforward to check that $\pi f_E=\pi f_V$ so that $\pi (\sigma +f_V)=0$.
\end{proof}

The purpose of Lemma \ref{lem:flux-cleaning} is so that one may focus on a single $z$ block of $Y^x$.
\begin{lemma}[$z$ Block Cleaning]
    \label{lem:z-block-cleaning}
    Let $\sigma$ be a 2-coboundary of a fixed $z$ block in $Y^x$.
    Then there exists 1-chain $f$ supported on hinges\footnote{For $m=3$, there is no loss in generality considering $f$ supported only on hinges $qz\otimes v$. Indeed, if $f$ is supported on some rail $z\otimes e$, then it contributes $+2$ to $|g^{Qx} f|$.
    This is equivalent to instead assigning $f$ a support on hinges $q^\pm z\otimes v$ where $v\prec e$ is the leaf of the $m=3$ repetition code, since each hinge contributes $+1$ to $|g^{Qx} f|$.} $qz\otimes v$ (but not on \textit{rails} $z\otimes e$) of the $z$ block such that $\delta^x f =\sigma$ and
    \begin{equation}
        |g^{Qx} f| \le 
        \begin{cases}
            \min ( |c^+|+|c^-|, 2m -|c^+|-|c^-|) & \beta =0 \\
            m - ||c^+|-|c^-|| & \beta =1
        \end{cases}
    \end{equation}
    where $c^\pm, \beta$ are defined via $\sigma$ as follows.
    Let $q^\pm$ denote the qubits in the $z$ block, i.e., $x\prec q^\pm \prec z$. 
    Choose a reference vertex $v_0$ on the repetition graph and let $\beta \in \F_2$ be the value of $\sigma$ on $xz\otimes v_0$. 
    Let $c_v^\pm \in \F_2$ be the summation of $\sigma$ on $q^\pm z\otimes e$ over edges $e$ along the unique path from $v_0 \to v$ in $R$ so that $c^\pm \in V$.
\end{lemma}

\begin{proof}
    For notation simplicity, let $\sigma_v$ denote the value of $\sigma$ on $xz\otimes v$; similarly, let $f_v^\pm$ denote the value of $f$ on $q^\pm z\otimes v$.
    Since $\sigma$ has zero parity on each $(z,e)$ block, we see that $\beta_v \equiv c_v^- +c_v^+ +\sigma_v$ is a constant for all $v$ and equal to $\beta$.
    Also note that by definition, if $\delta^x f=\sigma$, then for any edge $vv'$ of the repetition code, we must have
    \begin{equation}
        f_v^\pm +f_{v'}^{\pm} = c_v^\pm +c_{v'}^{\pm}
    \end{equation}
    Hence, using the reference point $v_0$, we see that for any vertex $v$,
    \begin{equation}
        f_v^\pm =f_{v_0}^\pm +c_v^\pm
    \end{equation}
    Also note that if we desire $\delta^x f =\sigma$, then $f_{v_0}^+ +f_{v_0}^-=\beta$.
    The converse is also clear, i.e., any $f$ of the form $f^\pm =c^\pm +\beta^\pm \cV$ where $\cV$ denotes the subset of all vertices in $V$ with $\beta^+ +\beta^- =\beta$ satisfies $\delta^x f=\sigma$ and is supported on hinges $qz\otimes v$.
    In particular, due to the gluing in Definition \ref{def:gluing}, we see that
    \begin{align}
        |g^{Qx} f| =|f| &=\min_{\beta^+ +\beta^-=\beta} (|c^+ +\beta^+ \cV| +|c^-+\beta^- \cV|)
    \end{align}
    And thus the statement follows by considering the cases $\beta=0,1$ separately.
\end{proof}

\begin{theorem}[$x$ Layer Cleaning]
    \label{thm:x-layer-cleaning}
    Let $\sigma$ be a 2-coboundary of $Y^x$ with $m=3$. Then there exists $f$ supported on $Y^x$, but not on \textit{rails} $z\otimes e$ for any $z\succ x$ and $e\in E$, such that $\delta^x f= \sigma$ and
    \begin{equation}
        |g^{Qx} f| \le \frac{9}{7} |\sigma|
    \end{equation}
    In particular, $\zeta^{QX}\le 9/7$
\end{theorem}

% \begin{remark}[$x$ Layer Tightness]
%     \label{rem:x-layer-tight}
%     Note the factor $9/7$ is tight in the sense that there exists a local graph $G^x$ such that the resulting relative syndrome expansion is $=9/7$. Indeed, consider the local graph $G^x= B_3$. \red{***}
% \end{remark}
\begin{proof}
    Fix a reference vertex $v^\star$ of $R$.
    By Lemma \ref{lem:flux-cleaning}, there exists $f^{\star}$ supported on cone lines $xq\otimes v$ with $v\ne v^{\star}$ such that $\sigma +\delta^x f^{\star}$ is a 2-coboundary for each $z$ block, whose restriction we denote as $\hat{\sigma}^z$.
    By Lemma \ref{lem:z-block-cleaning}, we can find $f^z$ supported on hinges in the $z$ block such that $\hat{\sigma}^z=\delta^x f^z$.
    Moreover, let $v^{\star}$ be the reference point to upper bound $|g^{Qx} f^z|$.
    Note that by using the same reference point $v^{\star}$ for both Lemma \ref{lem:flux-cleaning} and \ref{lem:z-block-cleaning}, we see that
    \begin{equation}
        \hat{\sigma}^z_e=\sigma^z_e,\quad \hat{\sigma}^z_{v^{\star}} =\sigma^z_{v^{\star}}
    \end{equation}
    where the subscripts denote the restriction of $\hat{\sigma}^z$ (or $\sigma^z$) to the corresponding face, e.g., $\hat{\sigma}^z_{e}$ is that on $qz\otimes e,q\succ z$ and $\hat{\sigma}^z_{v}$ is that on $xz\otimes v$.
    Hence,
    \begin{equation}
        |g^{Qx} f^z| \le 
        \begin{cases}
            2|\sigma^z_{01}| +  |\sigma^z_{12}| +3|\sigma^z_0| & v^{\star} =v_0  \\
            |\sigma^z_{01}| +  |\sigma^z_{12}| +3|\sigma^z_1| & v^{\star} =v_1  \\
            |\sigma^z_{01}| +  2|\sigma^z_{12}| +3|\sigma^z_2| & v^{\star} =v_2
        \end{cases}
    \end{equation}
    Where $v_0,v_1,v_2$ are the $m=3$ vertices of $R$ and $v_0,v_2$ are the leaves and $v_1$ is the center vertex.
    Let
    \begin{equation}
        f =f^{\star} +\sum_{z\succ x} f^z
    \end{equation}
    Since $g^{Qx} f^{\star}=0$, we see that $\delta^x f=\sigma$ and
    \begin{equation}
        |g^{Qx} f| \le 
        \begin{cases}
            2|\sigma_{01}| +  |\sigma_{12}| +3|\sigma_0| & v^{\star} =v_0  \\
            |\sigma_{01}| +  |\sigma_{12}| +3|\sigma_1| & v^{\star} =v_1  \\
            |\sigma_{01}| +  2|\sigma_{12}| +3|\sigma_2| & v^{\star} =v_2
        \end{cases}
    \end{equation}
    Note that the reference $v^{\star}$ can be arbitrarily chosen, and instead of taking the minimum over the three options, we can take the average, with weighted options $2/7,3/7,2/7$ so that the statement then follows.
\end{proof}

\subsection{$q$ Layer Cleaning}

\begin{figure}[ht]
\centering
\subfloat[\label{fig:q-layer}]{\includegraphics[height=0.40\columnwidth]{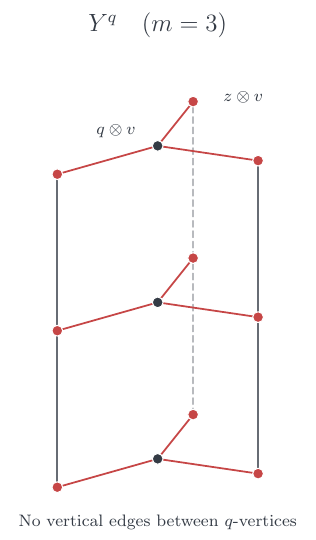}}
\hspace{0.015\columnwidth}
\subfloat[\label{fig:Yq-z-block}]{\includegraphics[height=0.40\columnwidth]{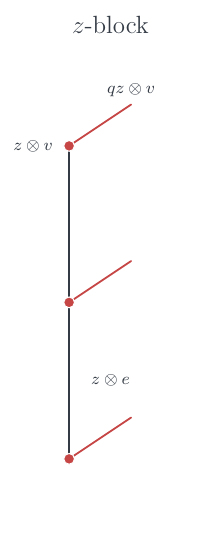}}
\hspace{0.015\columnwidth}
\subfloat[\label{fig:Yq-ze-block}]{\includegraphics[height=0.40\columnwidth]{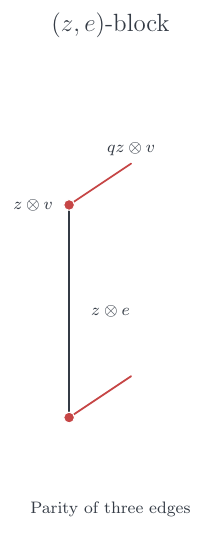}}
\caption{$Y^q$ blocks. (a) A $q$-layer $Y^q$. (b) A $z$-block, with vertices $z\otimes v$, spokes $qz\otimes v$, and rails $z\otimes e$. (c) A $(z,e)$-block for one repetition edge $e$; its parity is the $\F_2$ sum of the occupied edges among the two spokes and one rail. Open-ended spokes in (b), (c) meet $q$-vertices outside the block.}
\label{fig:Yq-blocks}
\end{figure}

In this subsubsection, we focus on the $q$-layer $Y^q= (Y^q,\delta^q)$ depicted by the diagram in Eq. \eqref{eq:Y-diag}.
\begin{remark}[$Y^q$ Geometric Reduction]
    \label{rem:q-geometric-reduction}
    Similar to Remark \ref{rem:x-geometric-reduction}, there also exists a natural geometric reduction of $q$ layers as shown in Fig. \ref{fig:Yq-blocks}.
    From an algebraic perspective, this corresponds to the flattened diagram
    \begin{equation}
    \begin{tikzpicture}[baseline]
    \matrix(a)[matrix of math nodes, nodes in empty cells, nodes={minimum size=5mm},
    row sep=1.5em, column sep=1.5em,
    text height=1.5ex, text depth=0.25ex]
    {& q,v\\
     z,v  &  qz,v\oplus z,e\\};
    \path[->,font=\scriptsize]
    (a-2-1) edge  (a-2-2)
    (a-1-2) edge  (a-2-2);
    \end{tikzpicture}
    \end{equation}
    where the row cocomplex is shown in in Fig. \ref{fig:Yq-z-block} and, in particular, given by the direct sum of $z$ \textbf{blocks} in $Y^q$, respectively.
    Similar to Remark \ref{rem:x-geometric-reduction}, we can compute the cohomology of the row cocomplex as $(0,z\otimes E)$ so that $Y^q$ (without meta-check 2-cells) is homotopically equivalent to the column cocomplex as shown below \cite{yuan2026unified}
    \begin{equation}
    \begin{tikzpicture}[baseline]
    \matrix(a)[matrix of math nodes, nodes in empty cells, nodes={minimum size=25pt},
    row sep=1.5em, column sep=1.5em,
    text height=1.25ex, text depth=0.25ex]
    {Y_0^q & Y_1^q \\
     q\otimes V  & z\otimes E \\};
    \path[->,font=\scriptsize,transform canvas={yshift=0.2ex}]
    (a-1-1) edge (a-1-2)
    (a-2-1) edge (a-2-2);
    \path[-left to,font=\scriptsize,transform canvas={xshift=0.2ex}]
    (a-1-1) edge node[right]{$\iota$}  (a-2-1)
    (a-1-2) edge node[right]{$\iota$}  (a-2-2);
    \path[left to-,font=\scriptsize,transform canvas={xshift=-0.2ex}]
    (a-1-1) edge node[left]{$\pi$}  (a-2-1)
    (a-1-2) edge node[left]{$\pi$}  (a-2-2);
    \end{tikzpicture}
    \end{equation}
    where $q\otimes v \prec z\otimes e$ if $q\prec z$ in $C$ and $v\prec e$ in repetition $R$, and $\pi,\iota$ denote the projection and inclusion chain maps.
    From a geometric standpoint, the projection $\pi:Y_1^q \to z\otimes E$ computes the \textbf{flux} (parity) on the $(z,e)$ \textbf{block} as shown in Fig. \ref{fig:Yq-ze-block}.
\end{remark}

\begin{lemma}[Flux Freedom]
    \label{lem:flux-freedom}
    If $\sigma$ is a 1-coboundary in $Y^q$, then the flux of $\sigma$ on $(z,e)$ blocks is constant over all $z\succ q$ for fixed $e$.
    In particular, there exists a bijection between 1-coboundaries $\sigma$ of $Y^q$ and $(\Phi,r_z,z\succ q)$ where $\Phi\in E,r_z\in V$, so that $\Phi$ is the flux of $\sigma$ and $r_z$ is the restriction of $\sigma$ on \textit{cone lines} $qz\otimes V$.
\end{lemma}
\begin{proof}
    Follows trivially from the homotopy equivalence between $Y^q$ and $q\otimes V\to z\otimes E$.
\end{proof}

\begin{theorem}[$q$ Layer Cleaning]
    \label{thm:q-layer-cleaning}
    Let $\sigma$ be a 1-coboundary of $Y^q$. If $m=3$, then
    \begin{equation}
        \|\sigma\|^{ZQ} |\sigma|\le |\sigma|_{\rm spokes} + \frac{4}{3}|\sigma|_{\rm rails}
    \end{equation}
    where $|\sigma|_{\rm spokes},|\sigma|_{\rm rails}$ denotes the weights of $\sigma$ restricted to spokes $qz\otimes v$ and rails $z\otimes e$, respectively.
    In particular, $\zeta^{ZQ}\le 4/3$.
\end{theorem}

% \begin{remark}[$m=3$ Special]
    % Note that the $ZQ$ relation coexpansion can be computed exactly as in Theorem \ref{thm:ZQ-expansion} of the Appendix and thus $m=3$ is special.
% \end{remark}
\begin{proof} 
    By Lemma \ref{lem:flux-freedom}, we see that a 1-coboundary $\sigma$ is specified by $(\Phi,r_z,z\succ q)$.
    Hence, $f$ is a 0-chain in $Y^q$ such that $\delta^q f=\sigma$ iff
    \begin{align}
        f_{z,v} &= r_{z,v} + f_{q,v} \\
        f_{q,v} &= f_{q,v^{\star}} + \sum_{e\in v^{\star} \to v} \Phi_e
    \end{align}
    where $v^{\star}$ is some chosen fixed reference vertex in $R$ and $v^{\star} \to v$ is the unique path in $R$ connecting the two vertices.
    In particular, we see that $f$ is unique up to freedom $f_{q,v^{\star}}\in \F_2$, which also follows from $H^0(Y^q) \cong \F_2$.
    Let
    \begin{equation}
        y_v = \sum_{e\in v^{\star}\to v} \Phi_e
    \end{equation}
    so that $\Phi \mapsto y$ is injective with the image being $y\in V$ such that $y_{v^{\star}}=0$. 
    Set $u_z= r_z+y$ and define the total change of $u_z$ as
    \begin{equation}
        \ch u_z \equiv \sum_{e=vv'\in E} |u_{z,v} + u_{z,v'}|
    \end{equation}
    Then it's straightforward to check that
    \begin{equation}
        \|\sigma\|^{ZQ} = \frac{1}{|\sigma|}\min\left(\sum_{z\succ q} |u_z|, \sum_{z\succ q} (3-|u_z|)\right)
    \end{equation}
    where
    \begin{equation}
        |\sigma| = \sum_{z\succ q} (|u_z+y| +\ch u_z)
    \end{equation}
    Note that
    \begin{equation}
        \min\left(\sum_{z\succ q} |u_z|, \sum_{z\succ q} (m-|u_z|)\right) \le \sum_{z\succ q} (\lambda |u_z| +(1-\lambda)(3-|u_z|))
    \end{equation}
    Choose $\lambda = 1-|y|/m$. Then
    \begin{equation}
        |\sigma|\|\sigma\|^{ZQ} \le \sum_{z\succ q} \left(|y| + \left(1-\frac{2}{3} |y|\right) |u_z|\right)
    \end{equation}
    Then it's sufficient to show that for any 1-chain $u\in V$, we have
    \begin{equation}
        |y| + \left(1-\frac{2}{3} |y|\right) |u_z| \le |u+y| +\frac{4}{3} \ch u
    \end{equation}
    Equivalently, it's sufficient to  show that
    \begin{equation}
        3 |u\cap y| -|u| |y| \le 2 \ch u 
    \end{equation}
    where we note that the left side is the covariance $3\Cov(u,y)$ between $u,y$ vectors.
    By symmetry between $y$ and $y+\cV$, it's sufficient to consider the cases where $|y|=0,1$. 
    The former case is trivial and thus we are only left with the singleton case $|y|=1$.
    In this case, if $u$ is a uniform constant so that $\ch u=0$, then it's clear that $\Cov(u,y)=0$ and thus the inequality holds. 
    Conversely, if $\ch u\ne 0$, then we see that the left-side is $\le 2$ and thus the inequality holds.
    Hence, the statement follows.
\end{proof}

\begin{remark}[Tightness]
    Since the local graph of $q$ is simply a set of disjoint vertices, one can further obtain $\zeta^{ZQ}$ for arbitrary $m\ge 2$ as follows.
    \begin{equation}
        \zeta^{ZQ} = 
        \begin{dcases}
            \left\lfloor \frac{m}{2} \right\rfloor & m\ne 3 \\
            \frac{6}{5} & m=3
        \end{dcases}
    \end{equation}
    The proof follows similarly to Theorem \ref{thm:q-layer-cleaning} but optimizes the constant using linear programming (LP).
    Since we do not use this fact, we omit the proof.
\end{remark}

\subsection{$x,q$ Cleaning}

\begin{theorem}[Contraction]
    \label{thm:contraction}
    Let $\Gamma$ be constructed with $m=3$. Then there exists a layer-local map $f$ which maps clean syndromes $\sigma \in \ker\bar{\delta}$ to qubits  such that
    \begin{equation}
        \sigma +\delta f(\sigma) \in \im \iota
    \end{equation}
    so that the \textbf{descent} $\desc \sigma = \pi (\sigma+\delta f(\sigma))$ is a syndrome in $C$, and 
    \begin{equation}
        |\desc\sigma| \le \kappa |\sigma|, \quad \kappa = \frac{4}{9} < \frac{1}{2}
    \end{equation}
    Moreover, $f(\sigma +\iota \noisy{\tau}) =f(\sigma)$ for any syndrome $\noisy{\tau}$ in $C$.
\end{theorem}

\begin{remark}[Single-Step Contraction]
    Follow the proof of Theorem \ref{thm:contraction}, as discussed previously, one may deduce a factor $\kappa = \zeta^{ZQ} \zeta^{QX}/m$ by a two step cleaning procedure.
    However, by Theorem \ref{thm:x-layer-cleaning} and Theorem \ref{thm:q-layer-cleaning}, this would result in $\kappa > 1/2$ for $m=3$. 
    % Since the factor $>1/2$, it would be insufficient to prove single-shotness as discussed in Section \red{***}.
    Hence, we optimize for the slightly more convoluted statements of Theorem \ref{thm:x-layer-cleaning} and \ref{thm:q-layer-cleaning}.
\end{remark}
\begin{proof}
    Let $\sigma= \sigma$ be a meta-check clean syndrome in $\Gamma$.
    By Definition \ref{def:gluing}, the restriction $\sigma|^x$ to the $x$-layer must be a 2-cocycle of $Y^x$.
    Since $Y^x$ is contractible, we see that it must also be a 2-coboundary and thus by Theorem \ref{thm:x-layer-cleaning}, there exists $f^x$ supported on $Y^x$, but not on rails $z\otimes e$, such that $\delta^x f=\sigma^x$ and
    \begin{equation}
        |g^{Qx} f^x| \le \frac{9}{7} |\sigma^x|
    \end{equation}
    Let $f^X$ denote the sum of $f^x$ over all $x\in X\super{i-1}$ and note that $\sigma^1=\sigma +\delta f^X$ has zero support on $x$-layers, while its restriction $\sigma^1|^q$ to a $q$ layer must be a 1-cocycle of $Y^q$.
    Since $Y^q$ is contractible, it must also be a 1-coboundary and thus by Theorem \ref{thm:q-layer-cleaning}, we see that there exists $f^q$ supported on $Y^q$ such that $\delta^q f^q=\sigma^1|^q$ and
    \begin{equation}
        |g^{Zq}f^q|\le |\sigma^1|^q|_{\rm spokes} +\frac{4}{3} |\sigma^1|^q|_{\rm rails}
    \end{equation}
    Note that by Theorem \ref{thm:x-layer-cleaning}, $\sigma^1|^q$ and $\sigma^q$ are equal on rails. Hence,
    \begin{equation}
        |g^{Zq}f^q|\le |\sigma|^q + g^{qX} f^X|_{\rm spokes} +\frac{4}{3} |\sigma|^q|_{\rm rails}
    \end{equation}
    Let $f^Q$ be the sum of $f^q$ over all qubits $q$ in $C$ and let $f=f^X+f^Q$.
    It's then clear that $\sigma\mapsto f$ is layer local and $\sigma^2 \equiv \sigma +\delta f$ is only supported on $z$ layers.
    Hence, we see that
    \begin{align}
        |\sigma^2| &= \sum_z |\sigma^2|^z \\ 
        &\le \sum_z |\sigma|^z +\sum_{q} |g^{Zq} f^q| \\
        &\le |\sigma|^Z + \sum_{q} |\dnum^q \sigma + g^{qX} f^X|_{\rm spokes} +\frac{4}{3} |\sigma|^q_{\rm rails} \\
        &\le |\sigma|^Z + \frac{4}{3} |\sigma|^Q + \sum_x |g^{Qx} f^x|  \\
        &\le |\sigma|^Z + \frac{4}{3} |\sigma|^Q + \frac{9}{7}|\sigma|^X  \\
        &\le \frac{4}{3} |\sigma|
    \end{align}
    Also note that $\sigma+\delta f\in \ker \bar{\delta}$ and thus by Definition \ref{def:gluing} must be in the image of the inclusion map $\iota$ so that
    \begin{equation}
        \sigma +\delta f =\iota \sigma'
    \end{equation}
    for some syndrome $\sigma'$ in $C$. Since $|\iota \sigma'| =3|\sigma'|$, we see that the inequality follows.
    Note that if $\sigma$ is a valid syndrome (not just meta-check clean), then $\sigma =\delta e$ for some $e$. Hence,
    \begin{equation}
        \sigma' = \delta \pi (e+f) \in \im \delta
    \end{equation}
    Hence, $\sigma'$ is a valid syndrome in $C$.
    Also note that by definition, $f$ satisfies the passthrough property in Definition \ref{def:cleaning}.
\end{proof}

% \begin{corollary}[Linear mod coBoundaries]
%     \label{cor:linear-mod-coboundaries}
%     Although the descent map is nonlinear in general, it is linear up to $\im \delta$, i.e.,
%     \begin{equation}
%         \desc(\sigma_1\super{i}+\sigma_2\super{i}) = \desc \sigma_1\super{i} +\desc \sigma_2\super{i} + \im \delta
%     \end{equation}
% \end{corollary}
% \begin{proof}
%     Let $\iota$ be the inclusion map. Then by Theorem \ref{thm:contraction},
%     \begin{equation}
%         \iota (\desc(\sigma_1\super{i}+\sigma_2\super{i}) = \desc \sigma_1\super{i} +\desc \sigma_2\sigma)\in  \im \delta
%     \end{equation}
%     Let $\pi:\Gamma\to C$ denote the projection chain map (without meta-checks) so that $\pi \iota $ is the identity, and thus the statement follows.
% \end{proof}

\subsection{Fill Bound}

As suggested by Definition \ref{def:level-map}, we must obtain the fill bound factor $c_f$.

\begin{theorem}[Fill Bound]
    \label{thm:fill-bound}
    Let $f$ be the cleaning map in Theorem \ref{thm:contraction} for $\Gamma$. Then for any clean syndrome $\sigma$ with $f=f(\sigma)$,
    \begin{equation}
        |\sh \{\succeq f\} \backslash \sh \sigma| \le |\sh f(\sigma) \backslash \sh \sigma| \le c_f |\sigma|
    \end{equation}
    where $c_f<3$
\end{theorem}

\begin{proof}
    Note that
    \begin{align}
        |\sh \{\succeq f\} \backslash \sh \sigma| &= \sum_{u} 1\{u\in \sh\{\succeq f\},\sigma|^u =0\} \\
        &= \sum_{u=q,z} 1\{u\in \sh\{\succeq f\},\sigma|^u =0\}
    \end{align}
    where we do not need to consider $x$ checks since if $\sigma|^x=0$, then by layer-locality, we have $f|^x=0$.
    First consider $u=q$. If $q\in \sh\{\succeq f\}$, then we see that either $f|^q\ne 0$ or there exists $x\prec q$ such that $g^{qx} f|^x\ne 0$.
    Since $\sigma|^q=0$, we see that if $g^{qX} f|^X=0$, then by construction, we would have chosen $f|^q=0$ and thus either possibility would imply that there exists $x\prec q$ such that $g^{qx} f|^x\ne 0$.
    Hence,
    \begin{align}
        \sum_{q} 1\{q\in \sh\{\succeq f\},\sigma|^q=0\} &\le \sum_{q} \sum_{x\prec q} 1\{ g^{qx} f|^x\ne 0,\sigma|^q=0\} \\
        &\le \sum_{x} \sum_{q\succ x} |g^{qx} f^X| \\
        &\le \frac{9}{7} |\sigma^X|
    \end{align}
    Now consider $u=z$. By definition, if $z\in \sh\{\succeq f\}$, then $g^{zq}f|^q\ne0$ for some $q\prec z$. Hence,
    \begin{align}
        \sum_{z} 1\{z\in \sh\{\succeq f\},\sigma|^z=0\} &\le \sum_{z} |g^{zQ} f^Q| 1\{\sigma|^z=0\} \\
        &\le |\iota \desc \sigma| \\
        &\le \frac{4}{3} |\sigma|
    \end{align}
    Combining the two results in $c_f = 4/3+ 9/7 <3$
\end{proof}

%% file: dual.tex
\section{Dual Step $\Gamma_*$}
\label{sec:dual}
In this section, we collect the explicit facts/notation that is captured implicitly in $\bar{\Gamma}_*=\bar{\Gamma}_*(C)$.
Specifically, write
\begin{equation}
    \bar{\Gamma}_*(C) = MX(\bar{\Gamma}_*)\to X(\Gamma_*)\to Q(\Gamma_*)\to Z(\Gamma_*)
\end{equation}
Then Lemma \ref{lem:SDR} is replaced with the diagram
\begin{equation}
    \begin{tikzpicture}[baseline]
    \matrix(a)[matrix of math nodes, nodes in empty cells, nodes={minimum size=25pt},
    row sep=2em, column sep=3em,
    text height=1.25ex, text depth=0.25ex]
    {&X& Q  & Z \\
     MX(\bar{\Gamma}_*) &X(\Gamma_*)& Q(\Gamma_*)  & Z(\Gamma_*)\\};
    \path[-left to,font=\scriptsize,transform canvas={yshift=0.2ex}]
    (a-1-2) edge (a-1-3)
    (a-1-3) edge (a-1-4)
    (a-2-1) edge (a-2-2)
    (a-2-2) edge (a-2-3)
    (a-2-3) edge (a-2-4);
    % \path[left to-,font=\scriptsize,transform canvas={yshift=-0.2ex}]
    % (a-2-2) edge node[below]{$\phi$}  (a-2-3)
    % (a-2-3) edge node[below]{$\phi$}  (a-2-4);
    \path[-left to,font=\scriptsize,transform canvas={xshift=0.2ex}]
    (a-1-2) edge node[right]{$\iota$}  (a-2-2)
    (a-1-3) edge node[right]{$\iota$}  (a-2-3)
    (a-1-4) edge node[right]{$\iota$}  (a-2-4);
    \path[left to-,font=\scriptsize,transform canvas={xshift=-0.2ex}]
    (a-1-2) edge node[left]{$\pi$}  (a-2-2)
    (a-1-3) edge node[left]{$\pi$}  (a-2-3)
    (a-1-4) edge node[left]{$\pi$}  (a-2-4);
    \end{tikzpicture}
\end{equation}
where $\pi=\pi_{\Gamma_*}:\Gamma_*\to C$ denote the parity map (transpose of the inclusion) so that $(\pi \gamma)|_u$ is the parity of vertices in $\gamma|^{u}$ relative to the layer $Y^{*u}$ for any cell $u=x,q,z$, and $\iota=\iota_{\Gamma_*}$ is a layer-local map such that $\iota,\pi$ form a SDR of $C,\Gamma_*^\circ$.
Furthermore, one can obtain a cleaning for $Z$ syndromes in $\Gamma_*$ as follows

\begin{theorem}[Dual Step Cleaning]
    \label{thm:dual}
    There exists a layer-local cleaning $f$ on clean syndromes $\sigma$ in $\Gamma_*$ such that $\sigma + \delta f(\sigma)\in \im \iota$ and 
    \begin{equation}
        |\desc \sigma| \le |\sigma|
    \end{equation}
    where $\desc \sigma = \pi (\sigma+\delta f(\sigma))$.
    Moreover,
    \begin{equation}
        |\sh \{\succeq f(\sigma)\} \backslash \sh \sigma| =0
    \end{equation}
    and
    \begin{equation}
        f(\sigma +\iota \noisy{s}) = f(\sigma)
    \end{equation}
\end{theorem}

\begin{proof}
    Note that a syndrome $\sigma$ in $Z(\Gamma_*)$ is simply a subset of vertices in (possibly many) $z$ layers.
    For each $z$ layer, we can choose a fixed reference vertex $v^{\star}$, and let $f$ denote the linear extensions of mapping each vertex $v$  to a minimum length path in the $z$ layer connecting $v^{\star}\lr v$.
    It's then clear that $f$ is layer-local and $\sigma +\delta f(\sigma)\in \im \iota$.
    In particular, $\desc \sigma$ is simply the parity of $\sigma|^z$ at each $z$-layer and thus the inequality follows.
    Also note that if $u\in \sh\{\succeq f\}$, then $u=z$ is a $z$-check and thus by definition, we must have $\sigma|^z\ne 0$. Hence, $c_f=0$.
    Since $f$ is linear, we see that $f(\sigma+\iota \noisy{s}) =f(\sigma) +f(\iota \noisy{s})=f(\sigma)$.
\end{proof}

%% file: metric.tex
\section{Metric Contraction}
\label{sec:metric}

We verify the geometric properties used in Section~\ref{sec:local-repair}. All distances below are on the bare graph, without meta-check vertices or edges. The repetition size and input sparsity bounds are fixed. 
% For meta-check locality, we omit zero stabilizer generators, so every input check is incident to a qubit. Isolated qubits may be treated separately. These conventions are preserved by the construction.

\begin{lemma}[Metric Contraction]
    \label{lem:explicit-metric-contraction}
    Each half-step $F=\Gamma,\Gamma_*$ satisfies
    \begin{equation}
        \dist_C(\sh u,\sh v)\le \tfrac12\dist_F(u,v)+2.
    \end{equation}
    Consequently, $\Theta=\Gamma_*\circ\Gamma$ and $\Theta_*=\Gamma\circ\Gamma_*$ have metric contraction $(4,3)$.
\end{lemma}
\begin{proof}
    First consider $\Gamma$. Besides the shadow $s(a)=\sh a$, associate to each bare cell $a$ its terminal input cell $t(a)$: a local vertex labeled $b$ has $t=b$, an edge labeled $ab$ has $t=b$, and a face labeled $xz$ has $t=z$. Tensoring with a repetition vertex or edge does not change $t$. Thus $\dist_C(s(a),t(a))\le2$.

    An edge of the bare graph is either internal to a layer or a gluing incidence. Along internal edges, $s$ is constant and $t$ changes by at most one input edge. Along gluing edges, $t$ is constant and $s$ changes by one input edge. This follows directly from the cell labels in Definition~\ref{def:gluing}; repetition incidences change neither label.
    If a path from $u$ to $v$ contains $I$ internal and $J$ gluing edges, its two projected walks give
    \begin{equation}
        \dist_C(s(u),s(v))\le \min\{J,I+4\}
        \le \tfrac12(I+J)+2.
    \end{equation}
    Taking a shortest path proves the claim. Transposition preserves the bare graph, so the same argument applies to $\Gamma_*$. Applying the two inequalities successively gives
    \begin{equation}
        \dist_C(\sh u,\sh v)
        \le \tfrac14\dist_\Theta(u,v)+\tfrac12\cdot2+2,
    \end{equation}
    and similarly for $\Theta_*$.
\end{proof}

\begin{lemma}[Bounded Lifts]
    \label{lem:metric-bounded-lifts}
    For either half-step $F$, let $\Delta_F$ bound the number of bare cells in one layer and put $D_F=2\Delta_F$. Then
    \begin{equation}
        \dist_F(u,v)\le 2D_F+3\dist_C(\sh u,\sh v).
    \end{equation}
\end{lemma}
\begin{proof}
    Fix a repetition vertex $v_0$ and represent each input cell $a$ by the vertex $r_a=a\otimes v_0$ in its own layer. For an input edge $ab$, the cone edge $ab\otimes v_0$ and the gluing incidence give a three-edge path
    \begin{equation}
        r_a\;\longleftrightarrow\; ab\otimes v_0\text{ in }Y^a
        \;\longleftrightarrow\; b\otimes v_0\text{ in }Y^a
        \;\longleftrightarrow\;r_b.
    \end{equation}
    Hence an input path lifts with length at most three times its length.

    Every bare cell in $Y^a$ connects to $r_a$ using at most two layers. Within a sheet, cone edges connect all vertices, and each bare edge or face is incident to this connected part. Different sheets connect through boundary rails. For an $x$-layer, use a rail at any incident $q$; for a $q$-layer with a $z$ neighbor, use a $z$ rail. If $q$ has no $z$ neighbor, use a $q$ rail in an incident $x$-layer and the gluing edges. For a $z$-layer, use a $z$ rail in an incident $q$-layer. The required neighbor exists under the stated convention. These paths lie in a connected graph of at most $2\Delta_F$ bare cells, so their length is at most $D_F$. Joining the two endpoint paths to the lifted input path proves the bound. The dual case follows by transposition.
\end{proof}

\begin{lemma}[Locality of Meta-Checks]
    \label{lem:explicit-meta-locality}
    The $Z$-sector meta-checks of $\bar\Theta=\bar\Gamma_*\circ\bar\Gamma$ are local relative to the bare graph. One may take
    \begin{equation}
        \gamma_\Gamma=2D_\Gamma+12,\qquad
        \gamma=2D_{\Gamma_*}+3\gamma_\Gamma.
    \end{equation}
    These constants depend only on the fixed repetition size and sparsity bounds. The analogous statement holds for the $X$ sector of $\bar\Theta_*$.
\end{lemma}
\begin{proof}
    First consider a meta-check in $\bar\Gamma(C)$. By Definition~\ref{def:gluing}, its incident checks lie in its own layer or in layers with shadows preceding its shadow. Any two such shadows have input distance at most four: each is within distance two of the meta-check's shadow. Lemma~\ref{lem:metric-bounded-lifts} therefore bounds the diameter of its support in $\Gamma$ by $\gamma_\Gamma$.

    Write $D=\Gamma(C)$. The dual half-step has no new $Z$-sector meta-checks; those of the composition are precisely the rows of $\bar\delta_\Gamma\pi_{\Gamma_*}$, as in Lemma~\ref{lem:level-map-composition}. Here $\pi_{\Gamma_*}$ on $Z$ is the parity map on the vertices of each $z$-layer (Appendix~\ref{sec:dual}). Consequently, checks $u,v$ in a common composed meta-check have shadows in a common meta-check of $\bar\Gamma$, and
    \begin{equation}
        \dist_\Theta(u,v)
        \le 2D_{\Gamma_*}+3\dist_D(\sh u,\sh v)
        \le\gamma.
    \end{equation}
    Uniform layer-size bounds follow from the fixed repetition size and Lemma~\ref{lem:sparsity}. Transposition gives the other sector. Only the intermediate meta-checks within this two-step composition are retained; the claim does not concern the full inherited hierarchy in Remark~\ref{rem:hidden-meta-checks}.
\end{proof}